\documentclass[11pt]{article}
\usepackage{fullpage}
\usepackage{amssymb, amsmath, graphicx, amsthm}
\usepackage{subfigure}  
\usepackage{enumerate}
\usepackage{newclude}
\usepackage{wrapfig}
\usepackage{hyperref}
\usepackage{breakurl}
\usepackage{pifont}

\usepackage{times}

      \theoremstyle{plain}

      \newtheorem{theorem}{Theorem}[section]
      \newtheorem{lemma}[theorem]{Lemma}
      \newtheorem{claim}[theorem]{Claim}
      \newtheorem{corollary}[theorem]{Corollary}
      
      \newtheorem{observation}[theorem]{Observation}
      \newtheorem{conjecture}[theorem]{Conjecture}

      \theoremstyle{definition}

      \theoremstyle{remark}

\def\cNP{\hbox{\rm \sffamily NP}}

\begin{document}

\title{Toward the Hanani--Tutte Theorem for Clustered Graphs}

\author{Radoslav Fulek\thanks{IST Austria,
\texttt{radoslav.fulek@gmail.com}.
The author gratefully acknowledges support from the Swiss National Science Foundation Grant PBELP2\_146705. 
The research leading to these results has received funding from the People Programme (Marie Curie Actions) of the European Union's Seventh Framework Programme (FP7/2007-2013) under REA grant agreement no [291734].
}}

\date{}

\maketitle

\thispagestyle{empty}


\begin{abstract}
The weak variant of Hanani--Tutte theorem says that a graph is planar, if it can be
drawn in the plane so that every pair of edges cross an even number of times.
Moreover, we can turn such a drawing into an embedding without changing the order
in which edges leave the vertices.
We prove a generalization of the weak Hanani--Tutte theorem  that also easily implies the monotone variant of the weak Hanani--Tutte theorem by Pach and T\'oth. Thus, our result
can be thought of as a common generalization of these two neat results.
In other words, we prove the weak Hanani-Tutte theorem for strip clustered graphs, whose clusters are linearly
ordered vertical strips in the plane and edges join only vertices in the same cluster or in neighboring clusters with respect to this order.
In order to prove our main result we first obtain a forbidden substructure characterization of embedded strip planar clustered graphs.

The Hanani--Tutte theorem says that a graph is planar, if it can be
drawn in the plane so that every pair of edges not sharing a vertex cross an even number of times.
We prove the variant of Hanani--Tutte theorem for strip clustered graphs if the underlying abstract graph is (i) three connected,
(ii) a tree; or (iii) a finite set of internally disjoint paths joining a pair of vertices. 
In the case of trees our result implies that c-planarity for flat clustered graphs with three clusters
is solvable in a polynomial time if the underlying abstract
graph is a tree.
\end{abstract}

\newpage

\thispagestyle{empty}

\tableofcontents

\newpage
\setcounter{page}{1}

\section{Introduction}

A \emph{drawing} of a graph $G=(V,E)$ is a representation of $G$ in the plane where every vertex
 in $V$ is represented by a unique point and every
edge $e=uv$ in $E$ is represented by a simple arc joining the two points that represent $u$ and $v$. If it leads to no confusion, we do not distinguish between
a vertex or an edge and its representation in the drawing and we use the words ``vertex'' and ``edge'' in both contexts. We assume that in a drawing no edge passes through a vertex,
no two edges touch and every pair of edges cross in finitely many points. A drawing of a graph is an \emph{embedding} if no two edges cross in the drawing. A graph is
\emph{planar} if it admits an  embedding.

\subsection{Hanani--Tutte theorem}
The Hanani--Tutte theorem~\cite{C34,T70} is a classical result that provides an algebraic characterization of planarity with interesting algorithmic consequences~\cite{FKP12J}. 
The (strong) Hanani--Tutte theorem says that a graph is planar as soon as it can be drawn in the plane so that no pair of edges that do not share a vertex
cross an odd number of times.
Moreover, its variant known as the weak Hanani--Tutte theorem~\cite{CN00,PT00,PSS06} states that if we have a drawing $\mathcal{D}$ of a graph $G$ where every pair of edges cross an even number of times then $G$ has an embedding that preserves the cyclic order of edges at vertices from $\mathcal{D}$.
Note that the weak variant does not directly follow from the strong Hanani--Tutte theorem.
For sub-cubic graphs, the weak variant implies the strong variant.

Other variants of the Hanani--Tutte theorem in the plane were proved for $x$-monotone drawings~\cite{FPSS12,PT04_monotone},
partially embedded planar graphs, simultaneously embedded planar graphs~\cite{S12+}, and two--clustered graphs~\cite{FKP12J}.
As for the closed surfaces of genus higher than zero, the weak variant is known to hold in all closed surfaces~\cite{PSS09}, and the strong variant was proved only
for the projective plane~\cite{PSS09c}.
It is an intriguing open problem to decide if the strong Hanani--Tutte theorem holds
for closed surfaces other than the sphere and projective plane.

To prove a strong variant for a closed surface it is enough to prove it for all the minor minimal graphs (see e.g.~\cite{D10} for the definition of a graph minor) not embeddable in the surface. Moreover, it is known that the list of such graphs is finite for every closed surface, see e.g.~\cite[Section 12]{D10}. Thus, proving or disproving the strong Hanani--Tutte theorem on a closed surface boils down to a search for a counterexample among a finite number of graphs. That sounds quite promising, since checking a particular graph is reducible to a finitely many, and not so many, drawings, see e.g.~\cite{SS13}. However, we do not have a complete list of such graphs for any surface besides the sphere and projective plane.

On the positive side, the list of possible minimal counterexamples for each surface was recently narrowed down to vertex two-connected graphs~\cite{SS13}.
See~\cite{S12} for a recent survey on applications of the Hanani--Tutte theorem and related results.

\subsection{Notation}
In the present paper we assume that $G=(V,E)$ is a (multi)graph as opposed to a simple graph.
Thus, we allow a graph to have multiple edges and loops unless stated otherwise.
We refer to an embedding of $G$ as to a \emph{plane} graph $G$.
The \emph{rotation} at a vertex $v$ is the clockwise cyclic order of the end pieces of edges incident to $v$. The \emph{rotation system} of a graph is the set of rotations at all its vertices.
An embedding of $G$ is up to an isotopy and the choice of an orientation and an outer (unbounded) face described by its rotation system.
Two embeddings of a graph are the \emph{same} if they have the same rotation system.
A pair of edges in a graph is \emph{independent} if they do not share a vertex.
An edge in a drawing is \emph{even} if it crosses every other edge an even number of times.
An edge in a drawing is \emph{independently even} if it crosses every other non-adjacent edge an even number of times.
A drawing of a graph is \emph{(independently) even} if all edges are (independently) even. Note that an embedding is an even drawing.
Let $x(v)$ and $y(v)$, respectively, denote the $x$-coordinate and $y$-coordinate of a vertex in a drawing.

Throughout the appendix we use the notation from~\cite{D10} to denote paths and walks in an abstract graph.
Thus, if $P$ is a path in $G$, i.e., a sequence of vertices without a repetition such that between every two consecutive
vertices there exists an edge in $G$, we write $P=v_1\ldots v_kPu_1\ldots u_l$ to label the first $k$ and the last $l$ vertices of $P$.
We use the same notation for walks that differ from paths by the possibility of revisiting vertices.
 When talking about oriented walks, i.e, walks traversed in a fixed direction, the orientation of the sub-walks of a walk $W$ is
inherited the orientation of $W$.
A walk is \emph{closed} if its first and last vertex coincide.
We denote by $Wvu_1\ldots u_kvW$ or $WvW'vW$ the concatenation of a walk $W$ with a closed walk $W'=vW'v=vu_1\ldots u_kv$ at the vertex $v$ of $W$.
A cycle in a graph is a closed walk without repetitions except for the first and last vertex that are the same.

\subsection{Strip clustered graphs}

A \emph{clustered graph}\footnote{This type of clustered graphs is usually called flat clustered graph in the graph drawing literature. We choose this simplified notation in order not to overburden the reader with unnecessary notation.}  is an ordered pair $(G,T)$, where $G$ is a graph, and
$T=\{V_i| i=1,\ldots, k\}$ is a partition of the vertex set of $G$ into $k$ parts. We call the sets $V_i$ clusters.
A drawing of a \emph{clustered graph} $(G,T)$ is clustered if vertices in $V_i$ are drawn inside
a topological disc $D_i$ for each $i$ such that $D_i\cap D_j=\emptyset$ and every edge of $G$
intersects the boundary of every disc $D_i$ at most once.
We use the term ``cluster $V_i$'' also when referring to the topological disc $D_i$ containing the vertices in $V_i$.
A clustered graph $(G,T)$ is \emph{clustered planar} (or briefly \emph{c-planar}) if $(G,T)$ has a clustered embedding.

A strip clustered graph is a concept introduced recently
 by Angelini et al.~\cite{ADDF13+}\footnote{The author was interested in this planarity variant independently prior
  to the publication of~\cite{ADDF13+} and adopted
 the notation introduced therein.} For  convenience  we slightly alter their definition and define ``strip clustered graphs''
as  ``proper'' instances of ``strip planarity'' in~\cite{ADDF13+}.
In the present paper we are primarily concerned with the following subclass of clustered graphs.
A \emph{clustered graph} $(G,T)$ is \emph{strip clustered}
 if $G=\left(V_1\cup \ldots \cup V_k, E\subseteq \bigcup_i{V_i \cup V_{i+1} \choose 2}\right)$, i.e., the edges
 in $G$ are either contained inside a part or join vertices in two consecutive parts.
A drawing of a strip clustered graph $(G,T)$ in the plane is \emph{strip clustered} if $i<x(v_i)<i+1$ for
all $v_i\in V_i$, and every line of the form $x=i$, $i\in \mathbb{N}$, intersects every edge at most once.
 Thus, strip clustered drawings constitute a restricted class
of clustered drawings.
We use the term ``cluster $V_i$'' also when referring to the vertical strip containing the vertices in $V_i$.
A strip clustered graph $(G,T)$ is \emph{strip planar}
if $(G,T)$ has a strip clustered embedding in the plane.
In the case when $G$ is given by an embedding with a given outer face we say that $(G,T)$ is an embedded strip clustered graph.
Then $(G,T)$ is \emph{strip planar}  if $(G,T)$ has a strip clustered embedding in the plane
in which the embedding of $G$ and the outer face of $G$ are as given.

The notion of clustered planarity appeared for the first time
in the literature in the work of Feng, Cohen and Eades~\cite{Feng95,Feng95+} in 1995 under the name of c-planarity
and its variant was considered already by Lengauer~\cite{L89} in 1989.
 See, e.g.,~\cite{BAT05,Feng95,Feng95+} for the general definition of c-planarity. Here, we consider only a special case of it.
 See, e.g.,~\cite{BAT05} for further references.
 We remark that it has been an intriguing open problem already for almost two decades to decide if c-planarity is \cNP-hard,
 despite of considerable effort of many researchers and that already for strip clustered graphs the problem is open~\cite{ADDF13+}.

To illustrate the difficulty of c-planarity testing we mention that
already in the case of three clusters~\cite{Cor05}, if $G$  is a cycle, the polynomial time algorithm for c-planarity is not trivial,
while if $G$ can be any graph, its existence is still open.
For comparison, if $G$ is a cycle then a strip clustered graph $(G,T)$ is trivially c-planar.
Moreover, this case dominates the strip planarity testing complexity-wise.

\begin{lemma}
\label{lemma:strip3}
The problem of strip planarity testing is reducible in linear time
to the problem of c-planarity testing in the case of flat clustered graphs with three clusters.
\end{lemma}

\begin{proof}
Given an instance of $(G,T)$ of strip clustered graph we construct a clustered graph $(G,T')$ with three clusters $V_0', V_1'$ and  $V_2'$ as follows. We put $T'=\{V_j'|  \ V_j'=\bigcup_{i; \ i \mod 3=j}V_i  \}$. Note that without loss of generality
we can assume that in a drawing of $(G,T')$ the clusters are drawn as regions bounded by a pair of rays emanating from the origin. By the inverse of a projective transformation taking the origin to the vertical infinity we can also assume that the same is true for a drawing of $(G,T)$.
 Notice that such clustered embedding of $(G,T)$ can be continuously deformed by a rotational
transformation of the form $(\phi,r)\rightarrow (k\phi, c\phi \cdot r)$ for appropriately chosen $k,c>1$, which is expressed in polar coordinates,
 so that we obtain a clustered embedding of $(G,T')$.
 We remark that  $(x,y)$ in Cartesian coordinates corresponds to $(\phi,\sqrt{x^2+y^2})$ such that $\sin \phi=\frac{y}{\sqrt{x^2+y^2}}$
and $\cos \phi=\frac{x}{\sqrt{x^2+y^2}}$ in polar coordinates.
On the other hand, it is not hard to see that if $(G,T')$ is c-planar then there exists a clustered embedding of $(G,T')$ with the following property.
For each $i=0,1,2$ and $j$ the vertices of $V_i'$ belonging to $V_j$
and the parts of their adjacent edges in the region representing $V_j'$ belong to a topological disc
 $D_j$ such that $D_j\cap D_{j'}=\emptyset$ for $j\not=j'$ fully contained in this region.
To this end we proceed as follows.
Let $E_i=E[V_{i-1},V_{i}]$ denote the edges in $G$ between $V_{i-1}$ and $V_{i}$.
Let $r_i$ denote the ray emanating from the origin that separates $V_{i-1}'$ from $V_{i}'$.
Given a clustered drawing of $(G,T')$, $p_e(\mathcal{D})$, for $e\in E_i$, is the intersection point of $e$
with the ray $r_{i \mod 3}$. Let $p$ denote the origin. Let $|pq|$ for a pair of points in the plane denote the Euclidean distance between $p$ and $q$.
Recall that $G$ has clusters $V_1,\ldots, V_k$. We obtain a desired embedding of $(G,T')$ inductively as $\mathcal{D}_{k}$ starting with
$\mathcal{D}_4$.
For $\mathcal{D}_i=(G,T')$, $i=5,\ldots, k$, we maintain the following invariant.
For each $j, \ 5\leq j\leq i$, we have $$\max_{e\in E[V_{j-4},V_{j-3}]}|pp_e| < \min_{e\in E[V_{j-1},V_{j}]}|pp_e| \  \ (*).$$
Let $\mathcal{D}$ denote a clustered embedding of $(G,T')$.
We start with a clustered embedding of $\mathcal{D}_4$ of $(G[V_1\cup V_2\cup V_3\cup V_4],T')$ inherited from $\mathcal{D}$.  In the $i^{\mathrm{th}}$ step of the induction we extend
$\mathcal{D}_{i-1}$ of $(G[V_1\cup V_2\ldots \cup V_{i-1}],T')$ inside the wedge corresponding to $V_{i-1 \mod 3}'$ and $V_{i \mod 3}'$
 thereby obtaining an embedding $\mathcal{D}_{i}$ of
$(G[V_1\cup V_2\ldots \cup V_{i}],T')$ so that the resulting embedding $\mathcal{D}_{i}$
is still clustered, and $(*)$ is satisfied. Since by induction hypothesis we have
$\max_{e\in E[V_{i-3j-2},V_{i-3j-1}]}|pp_e| < \min_{e\in E[V_{i-2},V_{i-1}]}|pp_e|$, for all possible $j$,
in $\mathcal{D}_{i-1}$ we have $G[V_{i-1}]$ drawn in the outer face of  $G[V_1 \cup \ldots \cup V_{i-2}]$.
Thus, we can extend the embedding of $\mathcal{D}_{i-1}$ into $\mathcal{D}_{i}$ in which all the edges
of $E_i$ cross $r_{i \mod 3}$ in the same order as in $\mathcal{D}$ while maintaining the invariant $(*)$
and the rotation system inherited from $\mathcal{D}$.
The obtained embedding $\mathcal{D}_{k}$ of $(G,T')$  can be easily transformed into a strip clustered embedding.

 Thus, $(G,T)$ is strip planar if and only if $(G,T')$ is c-planar.
\end{proof}

If $G$ is a tree also the converse of Lemma~\ref{lemma:strip3} is true. In other words,
 given an instance of clustered tree $(G,T)$ with three clusters $V_0, V_1$ and  $V_2$ we can easily construct a strip clustered tree $(G,T')$
with the same underlying abstract graph
such that $(G,T')$ is strip planar if and only if $(G,T)$ is c-planar. Indeed, the desired equivalent instance
 is obtained by partitioning the vertex set of $G$ into clusters thereby obtaining $(G,T'=\{V_i'|\ i\in I\subset \mathbb{N}\})$ as follows.
In the base case, pick an arbitrary vertex $v$ from a non-empty cluster $V_i$ of $G$ into $V_{i}'$,
and no vertex is processed.

In the inductive step we pick an unprocessed vertex $u$ that was already put into a set $V_j'$ for some $j$.
 We put neighbors of
$u$ in $V_{j \mod 3}$ into $V_j'$, neighbors in $V_{j+1 \mod 3}$ into $V_{j+1}'$,
 and neighbors of $v$ in $V_{j-1 \mod 3}$ into $V_{j-1}'$. Then we mark $u$ as processed.
 Since $G$ is a tree, the partition $T'$ is well defined.
Now, the argument of Lemma~\ref{lemma:strip3} gives us the following.

\begin{lemma}
\label{lemma:striptree}
The problem of c-planarity testing in the case of flat clustered graphs with three clusters is reducible in linear time
to the strip planarity testing if the underlying abstract graph is a tree.
\end{lemma}

\subsection{Hanani--Tutte for strip clustered graphs}

We show the following generalization of the weak Hanani--Tutte theorem for strip clustered graphs. See Figure~\ref{fig:strips3} and~\ref{fig:strips4} for an illustration.

\bigskip
\begin{figure}[htp]
\centering
\subfigure[]{\includegraphics[scale=0.7]{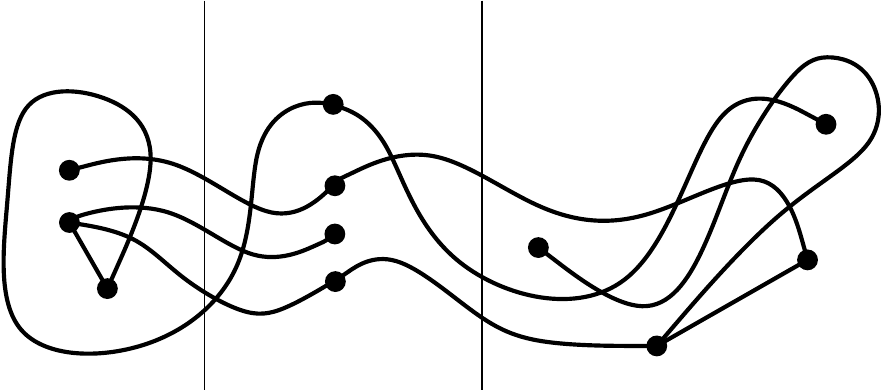}
    \label{fig:strips3}
	} \hspace{2cm}
\subfigure[]{\includegraphics[scale=0.7]{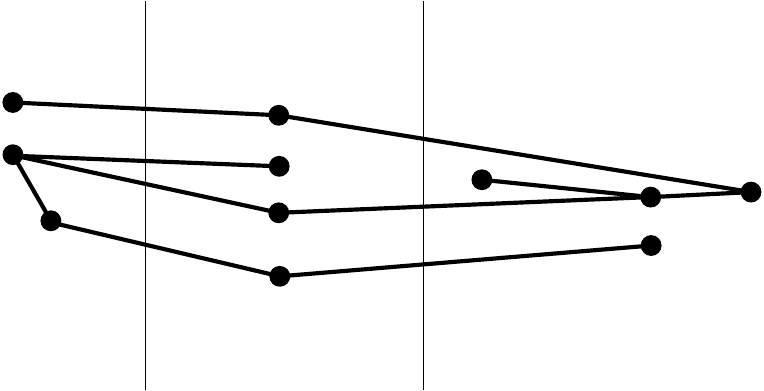}
	\label{fig:strips4}
	}
\caption{(a) Even clustered drawing of a strip clustered graph; (b) Strip clustered embedding of the same clustered graph.}
\end{figure}

\begin{theorem}\label{thm:linearly}
If a strip clustered graph $(G,T)$ admits an even strip clustered drawing $\mathcal{D}$ then $(G,T)$ is strip planar.
Moreover, there exists a strip clustered embedding of $(G,T)$ with the same rotation system as in $\mathcal{D}$.
\end{theorem}

Due to the family of counterexamples in~\cite[Section 6]{FKP12J}, Theorem~\ref{thm:linearly} does not leave too much room for straightforward generalizations.
Let $(G,T)$ denote a clustered graph, and let $G'=G'(G,T)$ denote a graph obtained from $(G,T)$ by contracting every cluster to a vertex and deleting all the loops and multiple edges.
If $(G,T)$ is a strip clustered graph, $G'$ is a sub-graph of a path.
Note that the converse is not true.
In this sense, the most general variant of Hanani--Tutte, the weak or strong one, we can hope for, is the one for the class of clustered graphs $(G,T)$, for which $G'$ is path.
In other words, in this variant $(G,T)$ is a strip clustered graph and a clustered drawing
is not necessarily strip clustered.

Our proof of Theorem~\ref{thm:linearly} is slightly technical, and combines a characterization of upward planar digraphs from~\cite{BBLM94} and Hall's theorem~\cite[Section 2]{D10}.
Using the result from~\cite{BBLM94} in our situation is quite natural, as was already observed in \cite{ADDF13+}, where they solve an intimately related algorithmic question discussed below.
The reason is that deciding the c-planarity for embedded strip clustered graphs is, essentially, a special case of the upward planarity testing.
The technical part of our argument augments the given drawing with subdivided edges so that we are able to apply Hall's Theorem.
Hence, the real novelty of our work lies in Lemma~\ref{marriage:lemma} that implies
the marriage condition, which makes the characterization do the work for us.
Moreover, as a byproduct of our proof we obtain a forbidden substructure characterization for embedded strip planar clustered graphs.

Our characterization verifies the following conjecture for $n=2$
stated for $n$-dimensional polytopal piecewise linearly embedded complexes~\footnote{The notion of the polytopal complex is not used anywhere else in the paper, and thus, a reader not interested in higher dimensional analogs of embedded graphs can skip following paragraphs.} generalizing embedded graphs.

Let $M_1$ and $M_2$, respectively, be $n_1$ and $n_2$-dimensional orientable manifold (possibly with boundaries) such that $n_1+n_2=n$.
Assume that $M_1$ and $M_2$ are PL embedded into $\mathbb{R}^n$ such that they are in general position,
i.e., they intersect in a finite set of points. Let us fix an orientation on $M_1$ and $M_2$.
The algebraic intersection number $i_A(M_1,M_2) = \sum_{p}o(p)$, where we sum over all intersection
points $p$ of $M_1$ and $M_2$ and $o(p)$ is 1 is if the intersection point is positive and -1 if the 
intersection point is negative with respect to the chosen orientations.
If $M_1$ and $M_2$ are not in a general position $i_A(M_1,M_2)$ denotes
$i_A(M_1',M_2')$, where $M_1'$ and $M_2'$, respectively, is slightly perturbed $M_1$ and $M_2$.
(A perturbation eliminates ``touchings'' and does not introduce new ``crossings''.)
Note that $i_A(M_1,M_2)=0$ is not affected by the choice of orientation.

Let $X$ be an an $n$-dim.~polytopal complex PL embedded in $\mathbb{R}^n$ simplicial up to the dimension $n-1$.
 Let $\mathcal{C}=\mathcal{C}(X,\mathbb{Z}_2)$ be the corresponding chain complex, and let $\gamma: X^0 \rightarrow \{1,\ldots, k\}$, where $X^0$ denote the set of vertices of $X$.
The complex $\mathcal{C}$ is \emph{compatible with $\gamma$} if for every pair of pure chains $(C_1,C_2)$ such that \\ 
(i) $dim(C_1)+dim(C_2)=n$; and  \\ (ii) the support of both $C_1$ and $C_2$ is homeomorphic to a ball of the corresponding dimension, \\

we have $i_A(C_1,C_2)= 0$ whenever $\gamma(C_1) \cap \gamma(\partial C_2) = \emptyset$ and $\gamma(\partial C_1) \cap \gamma(C_2) = \emptyset$.

\begin{conjecture}
\label{conj:conj1}
Suppose that $\mathcal{C}=\mathcal{C}(X,\mathbb{Z}_2)$,
where  $X$ is  an $n$-dim.~polytopal complex simplicial up to dimension $n-1$ PL embedded in $\mathbb{R}^n$,  is compatible with $\gamma: V \rightarrow \{1,\ldots, k\}$.
There  
exists a PL  embedding in $\mathbb{R}^n$ of $X$ (ambiently) isotopic
to the given embedding of $X$ such that  \\
(i) every non-empty intersection of an $i$-face, $0<i<n$, with a hyperplane $x_1=a$, for $a\in \mathbb{R}^n$, is homeomorphic to a $j$-dim.~ball, $j<i$; and\\
(ii) for every $u,v\in X^0$, $x_1(u)<x_1(v)$ if $\gamma(u)<\gamma(v)$, where $x_1(.)$
denotes the first coordinate.
\end{conjecture}

We remark that it is not clear to us if Conjecture~\ref{conj:conj1} is the right generalization of Theorem~\ref{thm:characterization} and also we do not have 
any support the conjecture except for the fact that it holds for $n=2$.
The fact that in condition~(ii) of the conjecture we do not require 
edges to have end vertices mapped by $\gamma$ at most one unit apart
does not make any difference since in in order to apply 
Theorem~\ref{thm:characterization} we would just subdivide edges
if necessary.

An edge $e$ of a topological graph is \emph{$x$-monotone} if every vertical line intersects
$e$ at most once.
Pach and T\'oth~\cite{PT04_monotone} (see also~\cite{FPSS12} for a different proof of the same result) proved the following theorem.

\begin{theorem}\label{thm:mono}
Let $G$ denote a graph whose vertices are totally ordered.
Suppose that there exists a drawing $\mathcal{D}$ of $G$, in which $x$-coordinates of vertices respect their order, edges are $x$-monotone and every pair of edges cross an even number of times.
Then there exists an embedding of $G$, in which the vertices are drawn as in $\mathcal{D}$, the edges are $x$-monotone, and the rotation system is the same as in $\mathcal{D}$.
\end{theorem}

We show that Theorem~\ref{thm:linearly} easily implies  Theorem~\ref{thm:mono}.
Our argument for showing that suggests a slightly different variant of Theorem~\ref{thm:linearly} for not necessarily clustered drawings that directly implies Theorem~\ref{thm:mono} (see Section~\ref{sec:bounded}).

The strong variant of Theorem~\ref{thm:linearly}, which is conjectured
to hold, would imply the existence
of a polynomial time algorithm for the corresponding variant of the
c-planarity testing~\cite{FKP12J}.
To the best of our knowledge, a polynomial time algorithm
was given only in the case, when the underlying planar graph has a prescribed planar embedding~\cite{ADDF13+}.
Our weak variant gives a polynomial time algorithm if $G$ is sub-cubic, and in the same case as~\cite{ADDF13+}.
Nevertheless, we think that the weak variant is interesting in its own right.

To support our conjecture we prove the strong variant of Theorem~\ref{thm:linearly} under
the condition that the underlying abstract graph $G$ of a clustered graph is a subdivision of a vertex three-connected graph
or a tree.
In general, we only know that this variant is true for two clusters~\cite{FKP12J}.

\begin{theorem}\label{thm:linearlyStrong}
Let $G$ denote a subdivision of a vertex three-connected graph.
If a strip clustered graph $(G,T)$ admits an independently even strip clustered drawing  then $(G,T)$ is strip planar.\footnote{The argument in the proof of Theorem~\ref{thm:linearlyStrong}
proves, in fact, a strong variant even in the case, when we require the vertices participating in a cut or two-cut to have the maximum degree three. Hence, we obtained a polynomial time algorithm even in the case of sub-cubic cuts and two-cuts.}
\end{theorem}

The proof of Theorem~\ref{thm:linearlyStrong} reduces to Theorem~\ref{thm:linearly} by correcting the rotations at the vertices
of $G$ so that the theorem becomes applicable.
By combining our characterization for embedded strip planar clustered graphs with the characterization of 0--1 matrices with consecutive
ones property due to Tucker~\cite{Tucker72} we prove the strong variant also for strip clustered graphs $(G,T)$,
where $G$ is a tree.

\begin{theorem}\label{thm:tree}
Let $G$ be a tree.
If a strip clustered graph $(G,T)$ admits an independently even strip clustered drawing  then $(G,T)$ is  strip planar.
\end{theorem}

As we noted above, the weak Hanani--Tutte theorem fails already for three clusters.
The underlying graph in the counterexample is a cycle~\cite{FKP12J}, and thus, the strong variant fails as well in
general clustered graphs without imposing additional restrictions.
On the other hand, by an argument analogous to the one yielding Lemma~\ref{lemma:striptree} and~\cite[Lemma 10]{FKP12J}
an independently even clustered drawing of a tree with three clusters gives us an independently even strip clustered drawing
of the same tree which is strip planar by Theorem~\ref{thm:tree}. Hence, the former clustered graph is c-planar. Thus, we have  the following variant
of the (strong) Hanani--Tutte theorem.
Due to the previously mentioned counterexample it cannot be generalized to a more general class of graphs without additional restrictions.

\begin{corollary}\label{thm:treeHT}
Let $G$ be a tree.
If a flat clustered graph $(G,T)$ with three clusters admits an independently even clustered drawing  then $(G,T)$ is c-planar.
\end{corollary}

By Lemma~\ref{lemma:striptree}, and the fact that a variant of the (strong) Hanani--Tutte theorem gives a polynomial time algorithm
 for the corresponding case of strip c-planarity testing, we have the following corollary.
However, we present a cubic time algorithm as a byproduct of our proof of Theorem~\ref{thm:tree}.

\begin{theorem}\label{thm:ahahah}
The c-planarity testing is solvable in a cubic time for flat clustered graphs with three clusters in the case
when the underlying abstract graph is a tree.
\end{theorem}

We are not aware of any other polynomial time algorithm solving this particular case of c-planarity testing.
In order to prove Theorem~\ref{thm:tree} we give an algorithm for the corresponding
strip planarity testing. Then the characterization due to Tucker~\cite{Tucker72}
is used to conclude that the algorithm recognizes an instance admitting
an independently even clustered drawing as positive.
The algorithm works, in fact, with 0--1 matrices having some elements ambiguous, and
can be thought of as a special case of Simultaneous PQ-ordering considered recently by
Bl{\"{a}}sius and Rutter~\cite{BR14}. However, we not need any  result from~\cite{BR14}
in the case of trees.

Using a more general variant of Simultaneous PQ-ordering we prove that strip planarity is 
polynomial time solvable also when the abstract graph is a set of internally vertex disjoint paths joining a pair
of vertices. We call such a graph a \emph{theta-graph}.
Unlike in the case of trees, in the case of theta-graphs we crucially rely on the main result
of~\cite{BR14}.
The following theorem follows immediately from Theorem~\ref{thm:theta_alg}.

\begin{theorem}\label{thm:theta}
The strip planarity testing is solvable in a quartic time if the underlying abstract graph is a 
theta-graph.
\end{theorem}

Similarly as for trees we are not aware of any previous algorithm with a polynomial running time in this case.
The corresponding Hanani--Tutte variant is then obtained similarly as Theorem~\ref{thm:treeHT}

\begin{theorem}\label{thm:thetaHT}
Let $G$ be a theta-graph.
If a strip clustered graph $(G,T)$ admits an independently even strip clustered drawing  then $(G,T)$ is  strip planar.
\end{theorem}

To prove Hanani--Tutte variants Theorem~\ref{thm:treeHT} and~\ref{thm:thetaHT} we use the characterization due to Tucker~\cite{Tucker72} of matrices with consecutive ones property
to conclude that the corresponding algorithm recognizes as ``yes'' instance those that admit
an independently even clustered drawing.
One might wonder why the Hanani--Tutte approach in various planarity variants
does not require us to deal with the rotation system, whereas the PQ-tree or SPQR-tree approach
is usually about deciding if a rotation system of a given graph satisfying certain conditions exists.
Since 0--1 matrices with consecutive ones property are combinatorial analogs of PQ-trees, our
 proof sheds some light on why independently even drawings guarantee that a desired rotation system exists.
 We believe that this connection
between PQ-trees and independently even drawings deserves further exploration.

\subsection{Relation to Level Planarities}

In a recent work Angelini et al.~\cite{Angelini20151} consider a related notion of clustered-level planarity
introduced by Forster and Bachmaier~\cite{FB2004}. Therein they  prove a tractability result for clustered-level planarity testing using Simultaneous PQ-ordering (via an intermediate reduction to another problem)
in the case of ``proper'' instances and \cNP-hardness in general.
For proving the \cNP-hardness result the authors reduce the problem of finding a total ordering~\cite{O79},
whose  variant we use to prove all of our algorithmic result, to theirs.
A reader might wonder if strip planarity is not just a closely related variant of clustered-level 
planarity, whose tractability can be determined just by applying same techniques.

The tractability of clustered-level planarity (in the case of proper instances), and other types of level planarity~\cite{BBF04,JLM98} is, perhaps most naturally, obtained via a variant of PQ-ordering with only a very limited use of the topology of the plane. Hence, in this sense level planarities seems to be much more closely related to representations
of graphs in one dimensional topological spaces such as interval or sub-tree representations.
In the case of strip planarity we do not see how to apply the theory of PQ-ordering, or related
SPQR-trees~\cite{DT89}, without using Theorem~\ref{thm:characterization} that turns the problem into a ``1-dimensional one''. Our proof of Theorem~\ref{thm:characterization} relies crucially on Euler's formula, and
Hall's theorem.
So far we were not able to solve the 1-dimensional one problem in its full generality.
However, we also believe that we did not exhaust potential of our approach, and suspect that
a resolution of the tractability status of strip planarity or c-planarity must use the topology of the plane in an ``essential way'', e.g., by using its topological invariants such as  Euler characteristic that is usually not exploited
in approaches based on PQ-tree style data structures.
In fact, we are not aware of any prior work in the context of algorithmic theory of planarity variants combining an approach relying on Euler's formula with PQ-tree style data structure,
except when Euler's formula is merely used to count the number of edges or faces in the running time analysis.


\section{Preliminaries}

\label{sec:prelim}

\subsection{Basic tricks and definitions}

\label{sec:even}

We will use the following well-known fact about closed curves in the plane.
Let $C$ denote a closed (possibly self-crossing) curve in the plane.

\begin{lemma}
\label{lemma:twocolor}
The regions in the complement of $C$ can be two-colored so that
two regions sharing a non-trivial part of the boundary receive opposite colors.
\end{lemma}

\begin{proof}
The two-coloring is possible, since we are coloring a graph in which every  cycle
can be written as the symmetric difference of a set of cycles of even length.
Hence, every cycle in our graph has en even length, and thus the graph is bipartite.
\end{proof}

Let us two-color the regions in the complement of $C$ so that
two regions sharing a non-trivial part of the boundary receive opposite colors.
A point not lying on $C$ is \emph{outside} of $C$, if it is contained in the region with the same color as the unbounded region. Otherwise,
such a point is \emph{inside} of $C$.
As a simple corollary of Lemma~\ref{lemma:twocolor} we obtain a well-known fact that a pair of closed curves in the plane cross an even
number of times.
 We use this fact tacitly throughout the paper.

Let $G$ denote a planar graph.
Since in the problem we study connected components of $G$ can be treated separately, we can afford to assume that $G$ is connected.
A \emph{face} in an embedding of $G$ is a connected component of the complement of the embedding of $G$ (as a topological space) in the plane.
 The \emph{facial walk} of $f$ is the walk in $G$ that we obtain by traversing the boundary of $f$.
A pair of consecutive edges $e$ and $e'$ in a facial walk $f$ create a \emph{wedge} incident to $f$ at their common vertex.
The cardinality $|f|$ of $f$ denotes the number of edges (counted with multiplicities) in the facial walk of $f$.
Let $F$ denote a set of faces in an embedding. We let $G[F]$ denote the sub-graph of $G$ induced by the edges incident to the faces of $F$.
A vertex or an edge is \emph{incident} to a face, if it appears on the facial walk of $f$.
The \emph{interior} and \emph{exterior}, respectively, of a cycle in an embedded graph is the bounded and unbounded connected component
of its complement in the plane.
Similarly, the \emph{interior} and \emph{exterior}, respectively, of an inner face in an embedded graph is the bounded and unbounded connected component
of the complement of its facial walk in the plane, and vice-versa for the outer-face.

Let $\gamma:V \rightarrow \mathbb{N}$ be a labeling of the vertices of $G$ by integers. Given a face $f$ in an embedding of $G$,
a vertex $v$ incident to $f$ is a \emph{local minimum} (\emph{maximum}) of $f$ if in the corresponding facial walk $W$ of $f$ the value
of $\gamma(v)$ is not bigger (not smaller) than the value of its successor and predecessor on $W$. A minimal and maximal, respectively, local
minimum and  maximum of $f$ is called  \emph{global minimum} and \emph{maximum} of $f$.
The face $f$ is \emph{simple} with respect to $\gamma$ if $f$ has exactly one local minimum and one local maximum.
The face $f$ is \emph{semi-simple} (with respect to $\gamma$) if $f$ has exactly two local minima and these minima have the same value, and two local maxima and these maxima have the same value.
A path $P$ is \emph{(strictly) monotone with respect to $\gamma$} if the labels of the  vertices on $P$ form a (strictly) monotone sequence if ordered
in the correspondence with their appearance on  $P$.

Given a strip clustered graph $(G,T)$ we naturally associate with it a labeling $\gamma$ that for each vertex $v$ returns
the index $i$ of the cluster $V_i$ that $v$ belongs to. We refer to the cluster whose vertices get label $i$ as to the $i^{\mathrm{th}}$ cluster.
Let $(\overrightarrow{G},T)$ denote the directed strip clustered graph obtained from $(G,T)$ by orienting every edge
$uv$ from the vertex with the smaller label to the vertex with the bigger label, and in case of a tie orienting $uv$ arbitrarily.
A \emph{sink} and \emph{source}, respectively, of $\overrightarrow{G}$ is
a vertex with no outgoing and incoming edges.

In our arguments we use a continuous deformation in order to transform a given drawing into a drawing with desired properties.
Observe that during such transformation of a drawing of a graph
the parity of crossings between a pair of edges is affected only when an edge $e$ passes over a vertex $v$,
in which case we change the parity of crossings of $e$ with all the edges incident to $v$. Let us call such an event an \emph{edge-vertex switch}.

\paragraph{Edge contraction and vertex split.}
A \emph{contraction} of an  edge $e=uv$ in a topological graph is an operation that turns
$e$ into a vertex
by moving $v$ along $e$ towards $u$ while dragging all the other edges incident to $v$ along $e$.
Note that by contracting an edge in an even drawing, we obtain again an even drawing.
By a contraction we can introduce multi-edges or loops at the vertices.

We will also often use the following operation which can be thought of as the inverse operation of the edge contraction
in a topological graph.
A \emph{vertex split} in a drawing of a graph $G$ is the operation that replaces a vertex $v$ by two vertices $v'$ and $v''$
drawn in a small neighborhood of $v$ joined by a short crossing free edge so that the neighbors of $v$ are partitioned into two parts
according to whether they are joined with $v'$ or $v''$ in the resulting drawing, the rotations at $v'$ and $v''$ are inherited from the
rotation at $v$, and the new edges are drawn in the small  neighborhood of the edges they correspond to in $G$.

\paragraph{Bounded Edges.}

\label{sec:bounded}

Theorem~\ref{thm:linearly} can be extended to more general clustered graphs $(G,T)$ that are not necessarily strip clustered, and drawings that
are not necessarily clustered.
The clusters $V_1,\ldots, V_k$ of $(G,T)$ in our drawing $\mathcal{D}$
 are still linearly ordered and
drawn as vertical strips respecting this order. An edge $uv\in E(G)$, where $u\in V_i, v\in V_j$, can join any two vertices of $G$, but it must be drawn so that it intersects only clusters $V_l$ such that $i\leq l \leq j$. We say that the edge $uv$ is \emph{bounded}, and the drawing \emph{quasi-clustered}.

A similar extension of a variant of the Hanani--Tutte theorem is also possible in the case of $x$-monotone drawings~\cite{FPSS12}.
In the $x$-monotone setting instead of the $x$-monotonicity of edges in an (independently) even drawing it is only required that the vertical projection of each edge is bounded by the vertical projections of its vertices. Thus, each edge stays between its end vertices.

In the same vein as for $x$-monotone drawing the extension of our result
to drawings $\mathcal{D}$ of clustered graphs with bounded edges can be proved by a reduction to the original claim, Theorem~\ref{thm:linearly}.
To this end we just need to subdivide every edge $e$ of $(G,T)$ violating conditions of strip clustered drawings so that newly created edges join the vertices in the same or neighbouring clusters, and perform edge-vertex switches in order to restore the even parity of the number of crossings between every pair of edges.
The reduction is carried out by the following lemma that is also used in the proof of
Theorem~\ref{thm:mono}.

\bigskip
\begin{figure}[htp]
\centering
\subfigure[]{\includegraphics[scale=0.7]{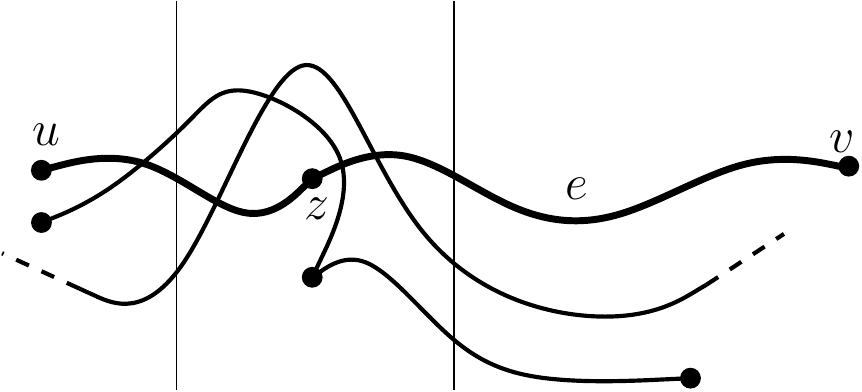}
    \label{fig:strips1}
	} \hspace{2cm}
\subfigure[]{\includegraphics[scale=0.7]{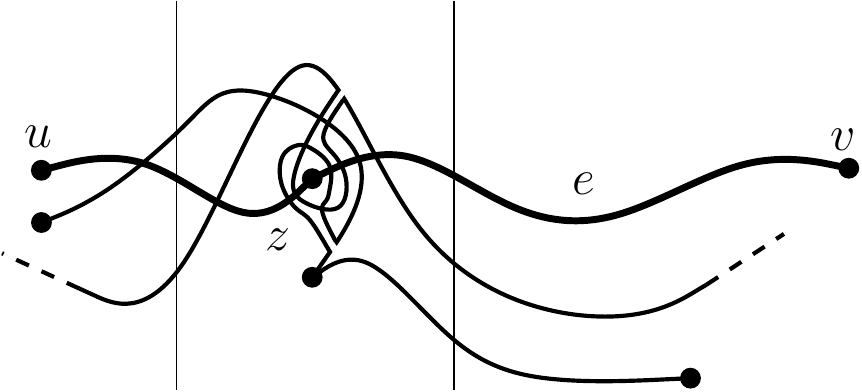}
	\label{fig:strips2}
	}
\caption{(a) Subdivision of the edge $e$ by the vertex $z$ resulting into odd crossing pairs; (b) Restoration of the evenness by performing edge-vertex switches with $z$.}
\end{figure}

\begin{lemma}
\label{lemma:bounded}
Let $\mathcal{D}$ denote an even quasi-clustered drawing of a clustered graph $(G,T)$.
Let $e=uv$, where $u\in V_i, v\in V_j$ denote an edge of $G$.
Let $G'$ denote a graph obtained from $G$ by subdiving $e$ by $|i-j|-1$ vertices.
Let $(G',T')$ denote the clustered graph, where $T'$ is inherited from $T$ so that the subdivided edge $e$ is turned into a strictly monotone path w.r.t. $\gamma$.
There exists an even quasi-clustered drawing $\mathcal{D}'$ of $(G',T')$, in which each new edge crosses the boundary of a cluster exactly once
and in which no new intersections of edges with boundaries of the clusters are introduced.
\end{lemma}

\begin{proof}
Refer to Figure~\ref{fig:strips1} and~\ref{fig:strips2}.
First, we continuously deform $e$ so that $e$ crosses the boundary of every cluster it visits at most twice. During the deformation
we could change the parity of the number of crossings between $e$ and some edges of $G$.
This happens when $e$ passes over a vertex $w$. We remind the reader that we call this event an edge-vertex switch.
Note that we can further deform $e$ so that it performs another edge-vertex switch
with each such vertex $w$,
while introducing new crossings with edges ``far'' from $w$ only in pairs.
Thus, by performing the appropriate edge-vertex switches of $e$ with
vertices of $G$
we maintain the parity of the number of crossings
of $e$ with the edges of $G$ and we do not introduce
intersections of $e$ with the boundaries of the clusters.

Second, if $e$ crosses the boundary of a cluster twice,
we subdivide $e$  by a vertex $z$ inside the cluster thereby turning $e$
into two edges, the edge joining $u$ with $z$ and the edge joining $z$ with $v$.
After we subdivide $e$ by $z$, the resulting drawing is not necessarily even.
However, it cannot happen that an edge crosses an odd number
 of times exactly one edge incident to $z$, since prior
  to subdividing the edge $e$ the drawing was even.
Thus, by performing edge-vertex switches of $z$ with edges
 that cross both edges incident to $z$ an odd number of times we restore the even parity of crossings between all pairs of edges.
By repeating the second step until we have no edge that crosses the boundary of a cluster twice we obtain a desired drawing of $G'$.
\end{proof}

\subsection{From strip clustered graphs to the marriage condition}

\label{sec:labelings}

\begin{wrapfigure}{r}{.5\textwidth}
\centering
\includegraphics[scale=0.7]{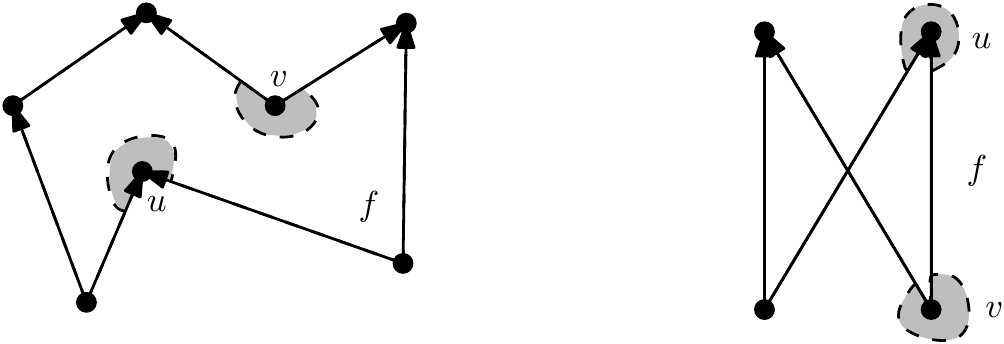}
\caption{Concave angles at $u$ and $v$ inside $f$ in an upward embedding (left), and concave angles in $f$ not admissible in an upward embedding (right).}
\label{fig:poincare1}
\end{wrapfigure}

The main tool for proving Theorem~\ref{thm:characterization} is~\cite[Theorem 3]{BBLM94} by Bertolazzi et al. that characterizes
embedded directed planar graphs, whose embedding can be straightened (the edges turned
into straight line segments) so that all the edges are directed upward,
 i.e., every edge is directed towards the vertex with a higher $y$-coordinate.
Here, it is not crucial that the edges are drawn as straight line segments, since
we can straighten them as soon as they are $y$-monotone~\cite{PT04_monotone}.
The theorem says that an embedded directed planar graph  $\overrightarrow{G}$ admits such an embedding, if there exists a mapping
from the set of sources and sinks of
$\overrightarrow{G}$ to the set of  faces of $\overrightarrow{G}$ that is easily seen to be necessary for such a drawing to exist.
Intuitively, given an upward embedding a sink or source $v$ is mapped to a face $f$
if and only if a pair of edges $vw$ and $vz$, incident to $f$ form inside $f$ a concave angle,
i.e., an angle bigger than $\pi$
(see Figure~\ref{fig:poincare1} for an illustration).
Thus, a vertex can be mapped to a face only if it is incident to it. First, note that the number of sinks incident to a face $f$ is the same as the number of sources incident to $f$. The mentioned easy necessary condition for the existence of an upward embedding is that (i) an internal and external face with $2k$ extremes (sinks or sources) have precisely $k-1$ and  $k+1$, respectively, of them mapped to it, and that (ii) the rotation at each vertex can be split into two parts consisting of incoming and outgoing edges. The embeddings satisfying the latter are dubbed \emph{candidate embeddings} by~\cite{BBLM94}.

Assuming that in $(G,T)$ each cluster forms an independent set, we would like to prove that $(\overrightarrow{G},T)$ satisfies this condition if $(G,T)$ does
not contain certain forbidden substructures in the hypothesis of Theorem~\ref{thm:characterization}. That would give us the desired clustered drawing by an easy geometric argument.
However, we do not know how to do it directly if faces have arbitrarily many sinks and sources. Thus, we first augment the given embedding by adding
edges and vertices so that (i) the outer face in $\overrightarrow{G}$ is incident to at most one sink and one source; (ii) each
internal face, that is not simple, is incident to exactly two sinks and two sources; and (iii) 
the hypothesis of Theorem~\ref{thm:characterization} is still satisfied.
 Let $(G',T')$  denote the resulting strip clustered graph.
This reduces the proof to showing that there exists a bijection between the set of internal semi-simple faces, and the set of sinks and sources in $\overrightarrow{G}'$ excluding the source and sink incident to the outer face.

By~\cite[Lemma~5]{BBLM94} the total number of sinks and sources is exactly the
total demand  by all the faces (in our case, the number of semi-simple faces plus two) in a candidate embedding.
Hence, by Hall's Theorem the bijection exists
 if every subset of internal semi-simple faces
 of size $l$ is incident to at least $l$ sinks
 and sources. The heart of the proof is then showing that
 the hypothesis of Theorem~\ref{thm:characterization}
 guarantees that this condition is satisfied (Lemma~\ref{marriage:lemma}).

\subsection{Necessary conditions for strip planarity}

\label{sec:crossing}

We present two necessary conditions that an embedded strip planar clustered graph
has to satisfy. In Section~\ref{sec:char} we show that the conditions
are, in fact, also sufficient.
For the remainder of this section let $(G,T)$ denote an embedded strip clustered graph. Let us assume that
$(G,T)$ is strip planar and let $\mathcal{D}$ denote the corresponding embedding
of $G$ with the given outer face.

In what follows we define the notion $i_A(P_1,P_2)$ of \emph{algebraic intersection number}~\cite{CN00} of a pair of oriented paths $P_1$ and $P_2$ in an embedding of a graph. We orient $P_1$ and $P_2$ arbitrarily.
 Let $P$ denote the sub-graph of $G$ which is the union of $P_1$ and $P_2$.
We define $cr_{P_1,P_2}(v)=+1$ ($cr_{P_1,P_2}(v)=-1$) if $v$ is a vertex of degree four in $P$ such that the paths $P_1$ and $P_2$ alternate in the rotation at $v$
and at $v$ the path $P_2$ crosses $P_1$ from left to right (right to left) with respect
to the chosen orientations of $P_1$ and $P_2$.
We define $cr_{P_1,P_2}(v)=+1/2$ ($cr_{P_1,P_2}(v)=-1/2$) if $v$ is a vertex of degree three in $P$ such that at $v$ the path
$P_2$ is oriented towards $P_1$ from left, or from $P_1$ to right (towards $P_1$ from right, or from $P_1$ to left) in the direction of $P_1$.
The algebraic intersection number of $P_1$ and $P_2$ is then the sum of $cr_{P_1,P_2}(v)$ over all vertices of degree three and four in $P$.

We extend the notion of algebraic intersection number to oriented walks as follows.
Let $i_A(W_1,W_2)=\sum cr_{u_1v_1w_1,u_2v_1w_2}(v_1)$, where the sum runs over all pairs $u_1v_1w_1\subseteq W_1$ and $u_2v_1w_2\subseteq W_2$
of oriented sub-paths of $W_1$ and $W_2$, respectively. (Sub-walks of length two in which $u_1=w_1$ or $u_2=w_2$ does not
 have to be considered in the sum, since their contribution towards the algebraic intersection number is zero anyway.)
Note that $i_A(W_1,W_2)$ is zero for a pair of closed walks.
Indeed, $i_A(C_1,C_2)=0$ for any pair of closed continuous curves in the plane which can be proved by observing that the statement
is true for a pair of non-intersecting curves and preserved
under a continuous deformation.
Whenever talking about algebraic intersection number of a pair of walks we tacitly assume that the walks are oriented. The actual orientation is not important to us since
in our arguments only the parity of the algebraic intersection number matters.

We will need the following property of $i_A(P_1,P_2)$ for a pair of paths in $G$.
Let $W$ denote a closed walk containing a vertex $v$.
Let $P$ denote a path having both end vertices (strictly) inside a single component of the complement of the embedding of $W$
in the plane.

\begin{lemma}
\label{lemma:symmDif}
For any walk $W'$ passing through $v$ we have
$i_A(W'vWvW',P)=i_A(W',P)$
\end{lemma}
\begin{proof}
In the proof $v$ stands only for the given occurrence of $v$ on $W'$.
We note that $i_A(W,P)=0$.
Let $u$ and $u'$, respectively, denote a vertex of $G$ that is the predecessor and successor of $v$ on $W'$.
Let $w$ and $z$, respectively, denote a vertex of $G$ that is the predecessor and successor of $v$ on $W$.
By the definition of $i_A$ we have
\begin{eqnarray}
\nonumber
i_A(W'vWvW',P)& = & i_A(W'uvzWwvu'W',P)\\
\nonumber
&= &i_A(W',P)-i_A(uvu',P)+i_A(W,P)-i_A(wvz,P)+ \\
\nonumber
& &     i_A(uvz,P)+i_A(wvu',P) \\
\nonumber
  &=& i_A(W',P)-i_A(uvu',P)
+i_A(uvz,P)+i_A(zvw,P)+i_A(wvu',P)\\
\nonumber
&=& i_A(W',P)-i_A(uvu',P)+i_A(uvu',P)=i_A(W',P).
\end{eqnarray}
\end{proof}

Given a strip clustered graph $(G,T)$ we naturally associate with it a labeling $\gamma$ that for each vertex $v$ returns
the index $i$ of the cluster $V_i$ that $v$ belongs to.
Let $G'\subseteq G$. Let $\max (G')$ and $\min (G')$, respectively, denote the maximal and minimal value of $\gamma(v)$, $v\in V(G')$.

\begin{figure}[t]
  \centering
\centering
{
\includegraphics[scale=0.7]{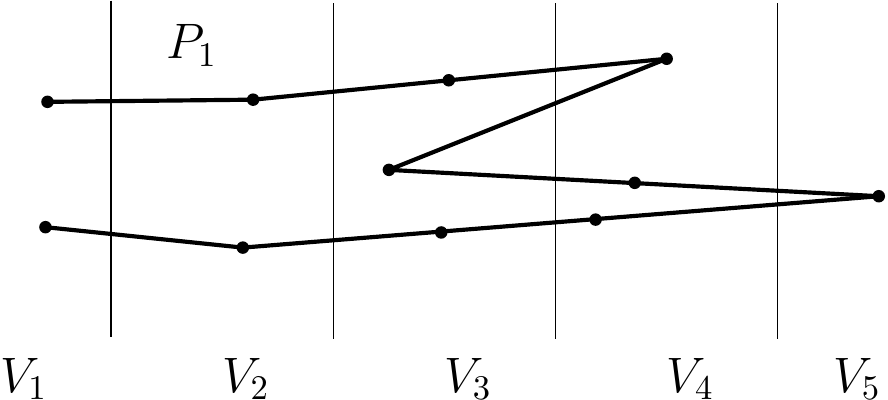}
    	} \hspace{10px}
{ 
\includegraphics[scale=0.7]{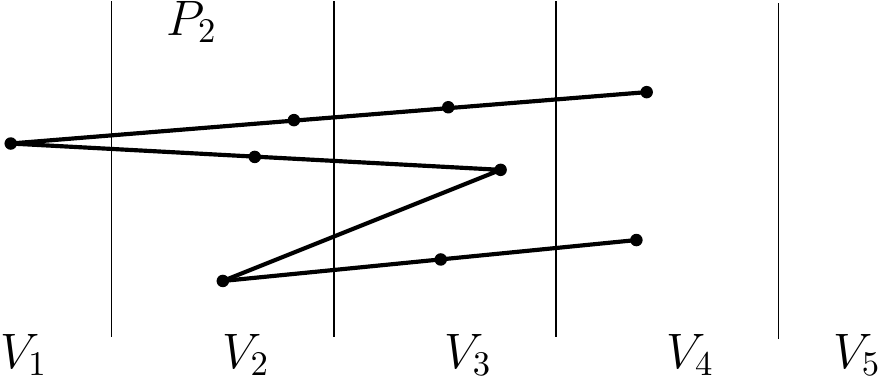}
		}
\caption{A path $P_1$ that is 1-cap (top); and a path $P_2$ that is a 4-cup (bottom).}
\label{fig:cupcup}
\end{figure}

\begin{figure}[t]
  \centering
\centering
{
\includegraphics[scale=0.7]{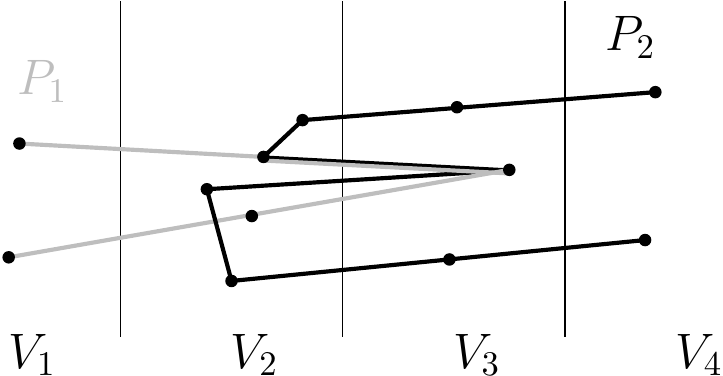}
    	} \hspace{10px}
{
\includegraphics[scale=0.7]{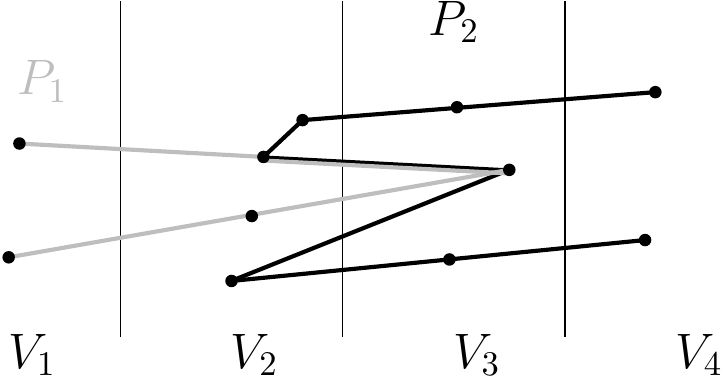}
		}
\caption{An infeasible pair of an 1-cap $P_1$ and a 4-cup  $P_2$ (top); and a feasible pair of an 1-cap $P_1$ and a 4-cup  $P_2$ (bottom).}
\label{fig:cupcup2}
\end{figure}

\paragraph{Definition of an $i$-cap and $i$-cup.}
A path $P$ in $G$ is an \emph{$i$-cap} and \emph{$j$-cup} if for the end vertices $u,v$ of $P$ and all $w\not=u,v$ of $P$
we have $\min (P) = \gamma (u) =\gamma(v)=i\not=\gamma(w)$ and $\max (P) = \gamma (u) =\gamma(v)=j\not=\gamma(w)$, respectively, (see Figure~\ref{fig:cupcup}).
A pair of  an \emph{$i$-cap} $P_1$ and \emph{$j$-cup} $P_2$ is \emph{interleaving} if (i) $\min(P_1)< \min(P_2)\le \max(P_1) < \max(P_2)$; and (ii) $P_1$ and $P_2$ intersect in a path (or a single vertex).
An interleaving pair of an oriented $i$-cap $P_1$ and $j$-cup $P_2$ is \emph{infeasible}, if  $i_A(P_1,P_2)\not=0$, and \emph{feasible},
otherwise (see Figure~\ref{fig:cupcup2}).
Thus, feasibility does not depend on the orientation.
Note that $i_A(P_1,P_2)$ can be either $0,1$ or $-1$. Throughout the paper by an infeasible and feasible pair of paths we mean an infeasible and feasible, respectively,  interleaving pair of an {$i$-cap} and {$j$-cup}.


\begin{observation}
\label{obs:crossing}
In $\mathcal{D}$ there does not exist an unfeasible interleaving pair $P_1$ and $P_2$ of an $i$-cap and $j$-cup, $i+1<j$.
\end{observation}

As a special case of Observation~\ref{obs:crossing} we obtain the following.

\begin{observation}
\label{obs:alternate}
The incoming and outgoing edges do not alternate at any vertex $v$ of $\overrightarrow{G}$  (defined in Section~\ref{sec:even}) in the rotation given by $\mathcal{D}$, i.e., the
incoming and outgoing edges incident to $v$ form two disjoint intervals in the rotation at $v$.
\end{observation}

Observation~\ref{obs:alternate} implies that the corresponding embedding of the upward digraph  $\overrightarrow{G}$
is a candidate embedding, if the clusters are independent sets.

We say that a vertex $v\in V(G)$ is \emph{trapped} in the interior of a cycle $C$
if in $\mathcal{D}$ the vertex $v$ is in the interior of $C$ and we have $\min(C) > \gamma(v)$ or $\gamma(v)> \max (C) $, where
$\max (C)$ and $\min (C)$, respectively, denotes the maximal and minimal label of a vertex of $C$.
A vertex $v$ is trapped if it is trapped in the interior of a cycle.

\begin{observation}
\label{obs:trap}
In $\mathcal{D}$ there does not exist a trapped vertex.
\end{observation}

\section{Marriage condition}

\label{sec:labeling}

In the present section we prove a statement, Lemma~\ref{marriage:lemma}, about planar bipartite graphs that is used later to show that
the marriage condition for applying Theorem 3 from~\cite{BBLM94} is satisfied.
Lemma~\ref{marriage:lemma} says that if necessary conditions for an embedded strip clustered graph for c-planarity from Section~\ref{sec:crossing}
holds then for each subset of semi-simple faces $F$ we have sufficiently many sources and sinks in $\overrightarrow{G}$  (defined in Section~\ref{sec:even}) incident to the faces in $F$.
The lemma works only for a special class of clustered graphs. Thus, in order to apply it in Sections~\ref{sec:normalized} we need to normalize our clustered graph so that it has a special form.

Let $G=(V,E)$ denote a planar bipartite connected graph given by the isotopy class
of an embedding. Let $V_1$ and $V_2$ denote the two parts of the bipartition of $G$.
All the sub-graphs in the present section are given by 
the isotopy class of an embedding inherited from $G$.

\bigskip

\begin{figure}[htp]
\centering
\subfigure[]{\includegraphics[scale=0.7]{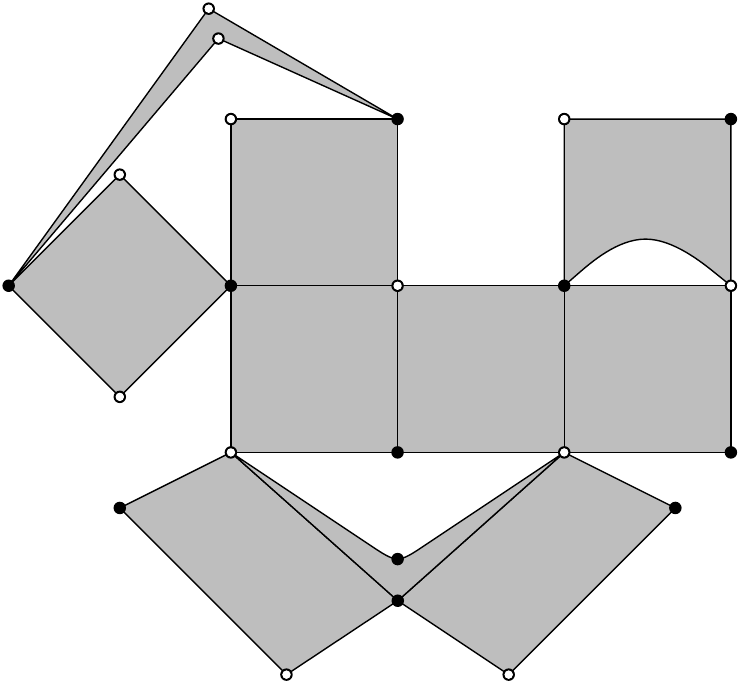}
    \label{fig:4faces}
	} \hspace{2cm}
\subfigure[]{\includegraphics[scale=0.7]{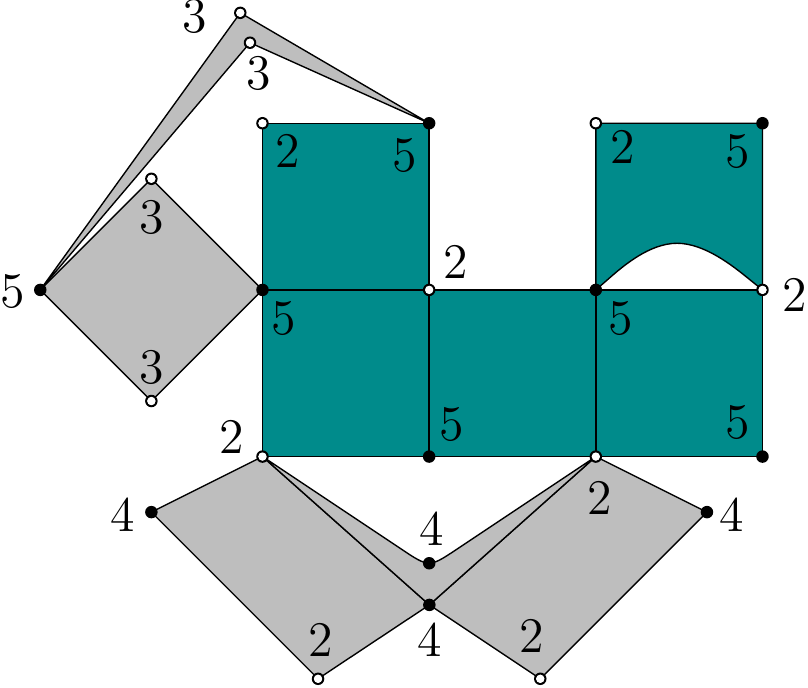}
	\label{fig:4faces1}
	}
\caption{(a) A set of fancy faces $F$ of $G$ filled by grey color. The vertices in $V_1$ and $V_2$, respectively, are marked by empty discs and full discs; (b)
The labeling $\gamma$ of the vertices in $V(G)$; the faces in an $\approx$-class are highlighted.}

\end{figure}

Refer to Figure~\ref{fig:4faces} and~\ref{fig:4faces1}.
For the remainder of this section let $\gamma:V \rightarrow \mathbb{N}$ be a labeling of the vertices of $G$ by integers and let $F$ denote a subset of inner faces of $G$
all of which are four-cycles such that

\begin{enumerate}[(i)]
\item
\label{it:semisimple}
all faces in $F$ are semi-simple; (defined in Section~\ref{sec:even})
\item
\label{it:v1v2}
for every $f\in F$
the vertices incident to $f$ in $V_1$ receive a smaller label than the vertices incident to $f$ belonging to $V_2$. Thus,  local minima of $f$ belong to $V_1$
and local maxima of $f$ to $V_2$; and
\item
\label{it:cycle}
$G$ does not contain a vertex $v$ trapped in the interior of a cycle. (defined in Section~\ref{sec:crossing})
\end{enumerate}

Let $F'$ denote the set of remaining faces of $G$, i.e., the faces not in $F$.
We say that $F$ is a set of \emph{fancy} faces in $G$ (with respect to $\gamma$).
The set $F$ is the one for which we show the marriage condition. Thus, we assume that $G$ is connected and $G=G[F]$.
Let $\mathcal{F}$ denote the closure of the union of faces in $F$.
A vertex $v$ of $G$ is a \emph{joint}, if no intersection of a disc neighborhood of $v$ with $\mathcal{F}$ is homeomorphic to a disc.

Refer to Figure~\ref{fig:4faces3}. The \emph{cardinality} of a relation $R$ is the number of pairs $(a,b)\in R$.
We define an equivalence relation $\approx$ on the set $F$ as the transitive closure of a relation $\sim$, where $f_1\sim f_2$,
if $f_1$ and $f_2$ share at least one vertex and their vertices in both $V_1$ and $V_2$ have the same label.
Let $\approx_1$ (resp. $\approx_2$) denote the minimal (with respect to its cardinality) transitive relation on $F$ such that $f_1\approx_1 f_2$ (resp. $f_1\approx_2 f_2$)
if $f_1\approx f_2$, or $f_1$ and $f_2$ are from two different $\approx$-classes that share a vertex from $V_1$ (resp. $V_2$), which are necessarily joints.

\begin{figure}[htp]
\centering
\subfigure[]{\includegraphics[scale=0.7]{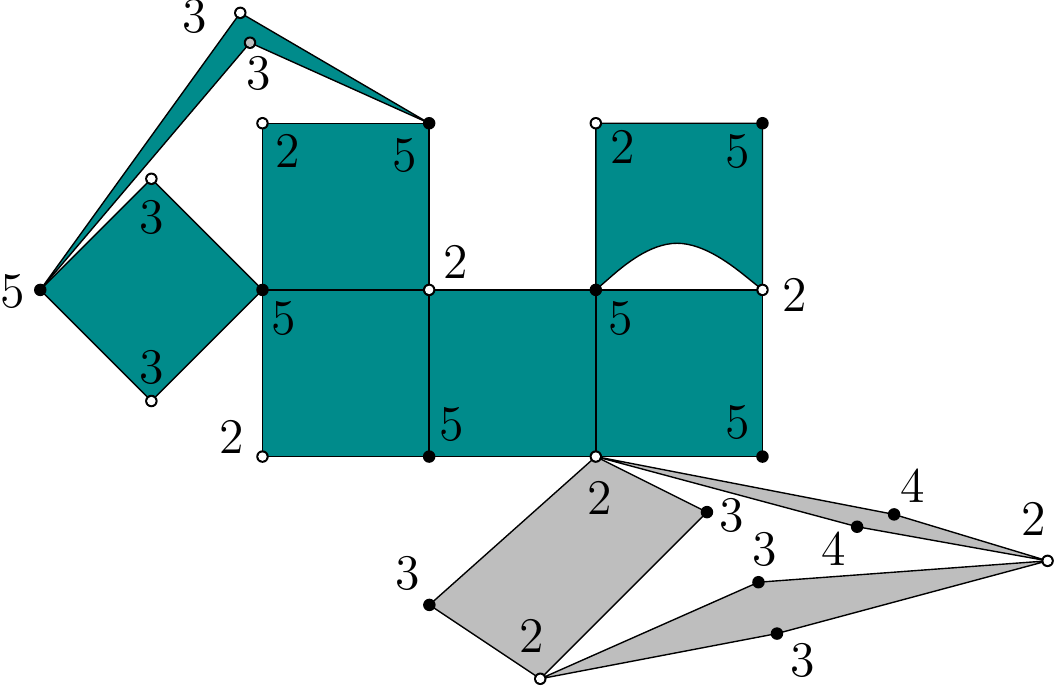}
    \label{fig:4faces3}
	} \hspace{2cm}
\subfigure[]{\includegraphics[scale=0.7]{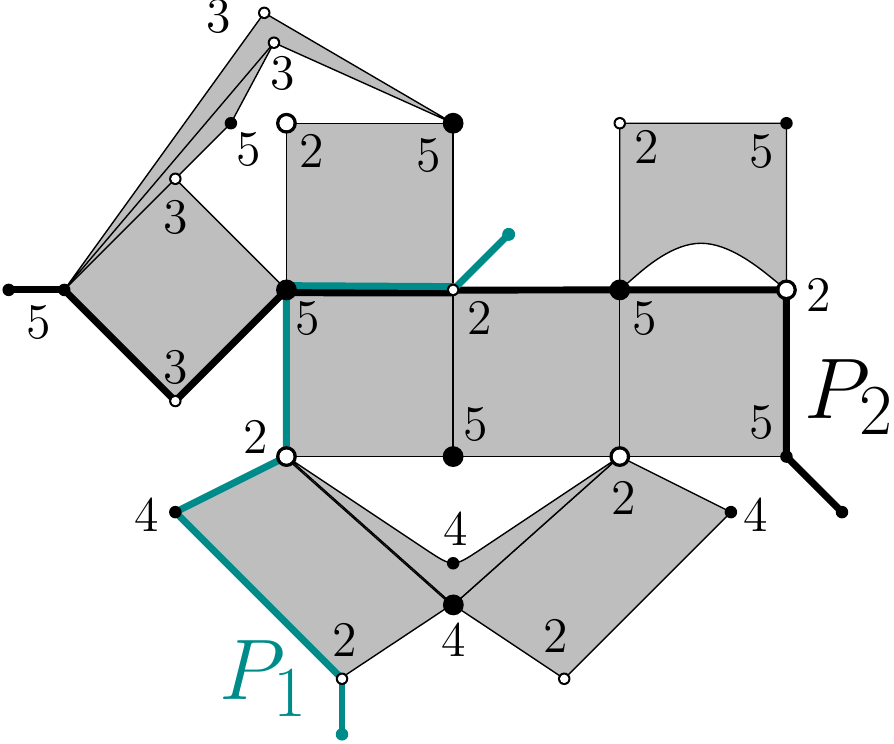}
	\label{fig:4faces-paths}
	}
\caption{(a) The faces in an $\approx_2$-class are highlighted.
(b) A pair of paths $P_1$ and $P_2$ that alternate, $P_1$ is a 2-cap and $P_2$ is a 5-cup.
 The marked vertices are enlarged.
 The marking violates property~(\ref{it:paths}) of $M$. }

\end{figure}

Refer to Figure~\ref{fig:4faces-paths}. In what follows we define a \emph{forbidden pair} of paths $P_1$ and $P_2$ that is obtained
from an unfeasible pair by deleting one edge from each end of a path
in the pair, but satisfying some additional properties.
A wedge $\omega$ at a vertex $v\in V_1$ in $G$ is \emph{good} if
$\omega$ is incident to a face in $F'$ and we can add to $G$ an edge $e$ incident to $v$ with the other end vertex $u$ of degree one,  that eliminates $\omega$ and has $\gamma(u)=\gamma(v)-1$, without
violating~(iii).
We define a good wedge at $v$ if $v\in V_2$ in the same way
except that we require $\gamma(u)=\gamma(v)+1$.

Let $P_1$ and $P_2$ denote an $i$-cap and $j$-cup, $i<j$, respectively,  intersecting in a path, whose each end vertex is incident to at least one 
good wedge  in $F'$. In the case an end vertex $v_1$
of $P_1$ is contained in $P_2$ we additionally require that the set of all edges creating good wedges at $v_1$ do not alternate with any pair of consecutive edges of $P_2$ in the rotation, and vice-versa. Let us 
attach to each end vertex of $P_1$ and $P_2$ a new additional edge eliminating a good wedge while keeping the drawing crossing free. (Note the short protruding edges in the figure.) The newly introduced vertices have degree one and each belongs to the interior of a face in $F'$.
Let $e(P_1)$ and $e(P_2)$ denote the paths obtained from $P_1$ and $P_2$, respectively, by extending them by newly added edges.
(We care about pairs of paths $e(P_1)$ and $e(P_2)$  yielding an unfeasible pair. Thus, we will be interested only in pairs $P_1$ and $P_2$
ending in the vertices that do not represent sinks or sources in our directed graph defined in Section~\ref{sec:even}.)

The pair of paths $P_1$ and $P_2$ is \emph{forbidden} if

\begin{enumerate}[(A)]
\item
$P_1$ and $P_2$ is contained in the sub-graph of $G$ corresponding to an $\approx_1$ and $\approx_2$, respectively, equivalence class,
and its end vertices belong to $V_1$ and $V_2$;
\item
$i_A(e(P_1),e(P_2))\not=0$  (note 
that the value of $i_A(e(P_1),e(P_2))$ is the same regardless of how
we extended $P_1$ and $P_2$ into $e(P_1)$ and $e(P_2)$ due to the definition of a forbidden pair); and
\item
$\min(P_1)\le \min(P_2)\le \max(P_1) \le \max(P_2)$.
\end{enumerate}

Next, we define a subset of $V$ that will be substituted by the set of the sinks and sources in $G$,
when using the result of the present section, Lemma~\ref{marriage:lemma}, later in Section~\ref{sec:normalized}.
A subset $M$ of $V$ is called the subset of \emph{marked} vertices of $V$ and has the following properties.

\begin{enumerate}[(a)]
\item
\label{it:interior}
Let $C$ denote a cycle contained in a sub-graph induced by faces in an $\approx_1$ or $\approx_2$ equivalence class.
For every such $C$, $M$ must contain all the vertices with label $\min(C)$ and $\max(C)$ that are either incident to $C$ or in the interior of $C$
except for vertices incident to a face of $F'$ in the exterior of $C$.
\item
\label{it:paths}
The graph $G$ does not contain a forbidden pair of an $i$-cap $P_1$ and $j$-cup $P_2$, $i<j$, none of whose end vertices belongs to $M$.
\end{enumerate}

Note that $M$ does not have to exist. Thus, the following applies only to graphs $G$ with a labeling $\gamma$ for which $M$ exists.
When applying Lemma~\ref{marriage:lemma} condition~(\ref{it:interior}) is ``enforced'' by Observation~\ref{obs:trap},
since $M$ corresponds to the set of sinks and sources of
 $\overrightarrow{G}$.
Condition~(\ref{it:paths}) is ``enforced'' by Observation~\ref{obs:crossing}.

The lemma bounds from below the size of $M$ by the size of $F$ and constitutes the heart of the proof of
characterization of embedded strip planar clustered graphs.

\begin{lemma}
\label{marriage:lemma}
Suppose that $M$ exists.
We have $|M|\ge |F|$.
\end{lemma}

The key idea in the proof of the lemma is the combination of properties~(\ref{it:interior}) and~(\ref{it:paths}) of $M$ with the following simple observation.

\begin{observation}
\label{obs:simple}
Let $C=v_1v_2\ldots v_{2a}$, $a\ge 2$, denote an even cycle. Let $V'$ denote a subset of the vertices of $C$ of size
at least $a+2$. Then $V'$ contains four vertices $v_i,v_j,v_k$ and $v_l$, where $i<j<k<l$, such that $i,k$ is odd and $j,l$ is even (or vice versa).
\end{observation}
\begin{proof}
For the sake of contradiction we assume that $V'$ does not contain four such vertices.
Let $V_0$ and $V_1$, respectively, denote the vertices of $V$ with even and odd index.
Similarly, let $V_0'$ and $V_1'$, respectively, denote the vertices of $V'$ with even and odd index.
Suppose that $2 \le |V_0'|\ge |V_1'|$ and fix a direction in which we traverse $C$.
Between every two consecutive vertices of $V_0'$ along $C$ except for at most one pair of consecutive vertices
we have a vertex in $V_1-V_1'$. Thus, $|V_1-V_1'|\geq |V_0'|-1$.
On the other hand, $|V_1-V_1'|=a-|V_1'|\leq a-(a+2-|V_0'|)=|V_0'|-2$ (contradiction).
\end{proof}

Before we turn to the proof of the lemma, let us illustrate how Observation~\ref{obs:simple}
and properties~(\ref{it:interior}) and~(\ref{it:paths}) of $M$ implies the lemma in the case, when all the faces of $G[F]$  are in $F$
except for the outer face $f_o$, and $G[F]$ is two-connected. Note that in this case all the faces of $F$ are in the same $\approx$ class.
We have $2|E|=4|F|+\sum_{f'\in F'} |f'|$, and thus, by Euler's formula we obtain $2|V|+2|F|=2|E|+2=4|F|+|f_o|+2$.
Thus, $|V|-\frac 12|f_o|-1=|F|$. If $|M|<|F|$, by property~(\ref{it:interior}), we have at least $\frac 12|f_o|+2$ vertices
incident to $f_o$ not belonging to $M$, and hence, by Observation~\ref{obs:simple} we find four vertices incident to $f_o$, whose existence
yields a pair of paths violating property~(\ref{it:paths}).

\begin{proof}[{\it Proof of Lemma~\ref{marriage:lemma}}]
We extend the previous illustration for the case
when all the fancy faces of $G$ are in the same $\approx$
to the general one by devising a charging scheme recursively assigning marked vertices to sub-graphs
induced by fancy faces in a single $\approx$ class.
Similarly as above, by Euler's formula $|V|+|F|+|F'|=|E|+2$ and the identity $4|F|+\sum_{f'\in F'}|f'|=2|E|$.
We obtain the following
\begin{equation}
\label{eqn:euler}
|V|=|F|+\sum_{f'\in F'}(|f'|/2 -1)+2
\end{equation}

Let $U$ denote $V\setminus M$, i.e., the set of vertices of $V$ that are not marked.
By (\ref{eqn:euler}), it is enough to show that $|U|\le \sum_{f'\in F'}(|f'|/2 -1)+2$.

Refer to Figure~\ref{fig:split}.
A \emph{separation} of $G$ at a joint is an operation that replaces $v$ by as many vertices as there are
arc-connected components in an intersection of a small punctured disc neighborhood of $v$ with $\mathcal{F}$ so that these vertices form an independent set; and for each face $f$ in $F$ incident to $v$ we obtain a new face in $F$ by replacing $v$ with a copy of $v$ corresponding to the arc-connected component meeting $f$.
Thus, a separation preserves the number of faces of $F$ in $G$.
As a pre-processing step, we first perform the separation at every joint of $G$ which is not marked, i.e., not in $M$.
By slightly abusing the notation we denote the resulting graph by $G$.
 By performing the separation we definitely cannot violate~(\ref{it:interior}) in the resulting embedded graph $G$, but we have also the following.

\bigskip
\begin{figure}[htp]
\centering
\subfigure[]{\includegraphics[scale=0.7]{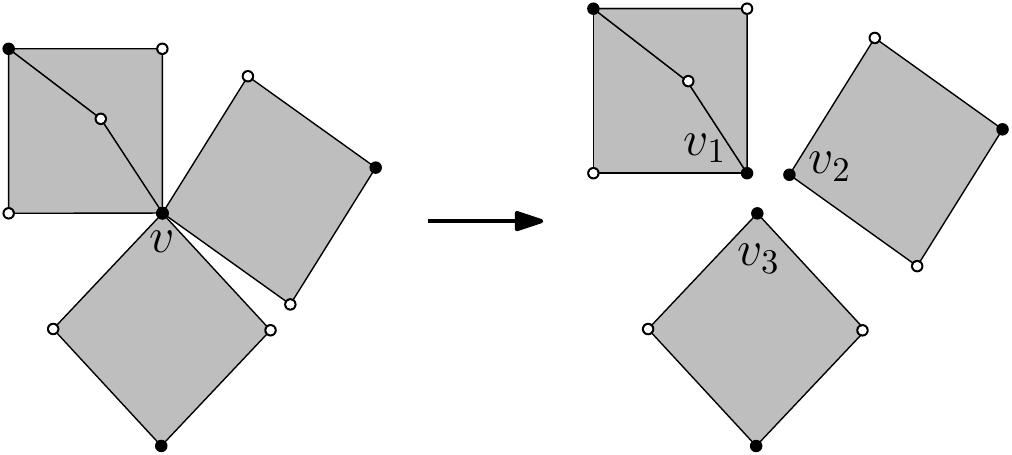}
    \label{fig:split}
	} \hspace{2cm}
\subfigure[]{\includegraphics[scale=0.7]{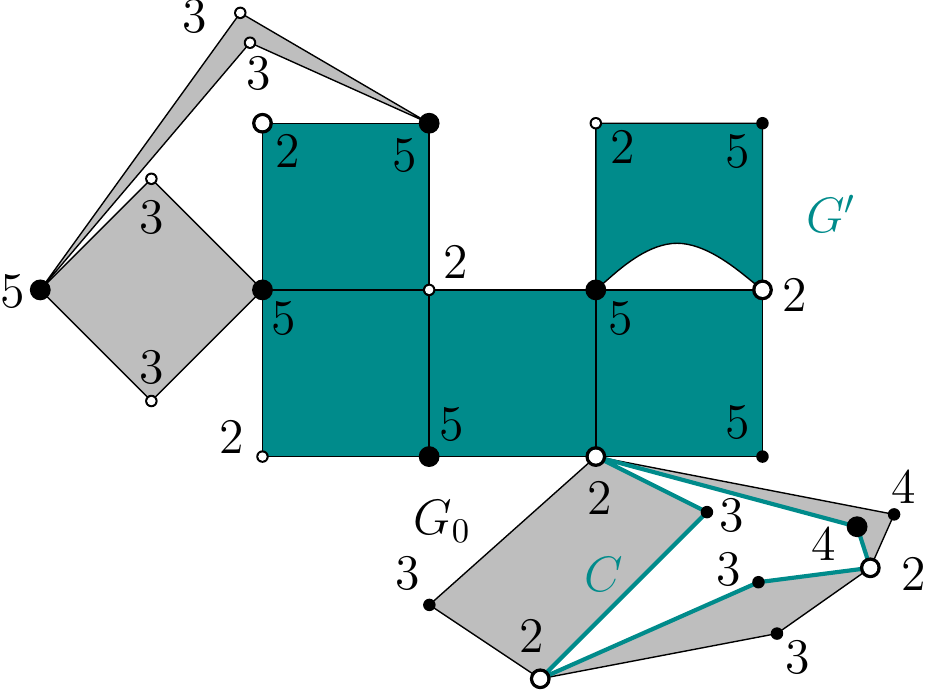}
	\label{fig:facesProof}
	}
\caption{(a) Separation at the vertex $v$; (b) The sub-graph $G'$ of $G$ with an $\approx_1$ equivalence class corresponding to $G_0$ attached to it. The marked vertices are enlarged.}
\end{figure}

\begin{figure}[htp]
\centering
\includegraphics[scale=0.7]{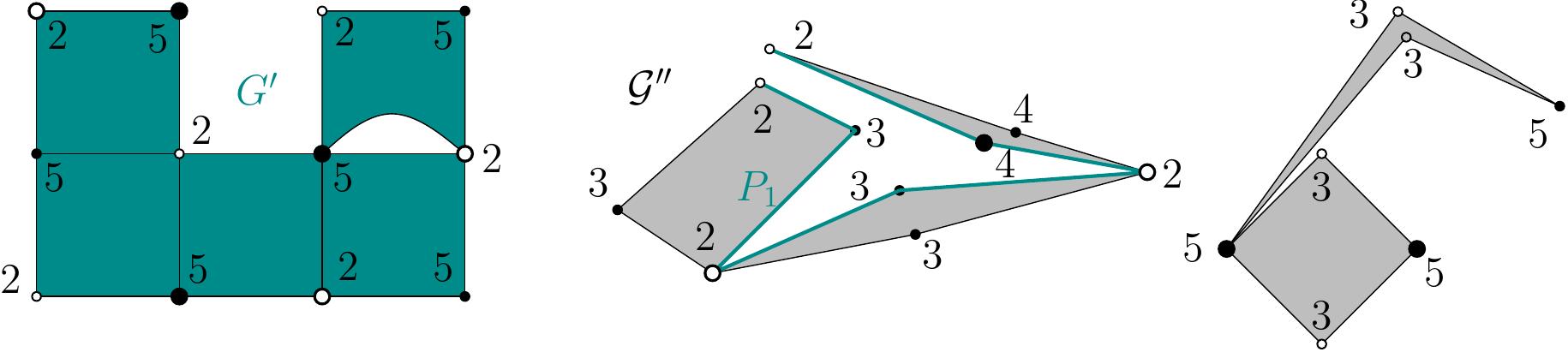}
\caption{The mapping of marked vertices to $G'$ and $G''$. The path $P_1$ is a 2-cap and every cup that forms a forbidden pair with $P_1$ has to end in a vertex with label four.}
\label{fig:facesProof2}
\end{figure}

\begin{claim}
\label{claim:nonViolate}
By performing the separations at all non-marked joint vertices we cannot violate property~(\ref{it:paths})
if none of~(\ref{it:interior}) and~(\ref{it:paths})
was violated before the separations.
\end{claim}
\begin{proof}
The violation of~(\ref{it:paths}) is ruled out as follows if a pair of 
paths $P_1$ and $P_2$ violating it were paths before the separation. 
Due to properties~(\ref{it:interior}) and~(\ref{it:cycle}) good wedges  at the end vertices of $P_1$ and $P_2$ after the separation yield good wedges at the end vertices of $P_1$ and $P_2$ before the separations.
Indeed, by~(\ref{it:cycle}) the only way for an end vertex $v_1$ of $P_1$ not to have an incident good wedge (before the separation) is to be contained in the interior of $C$ or on  $C$, where $C$ is a cycle in $G$ such that $\min(C)=\gamma(v)$, and not 
to be incident to a face in $F'$ in the exterior of $C$,
and hence, by~(\ref{it:interior}) to be marked (contradiction).
The same argument works for $P_2$ except $\max(C)=\gamma(v)$.
Now, a pair of  paths $P_1$ and $P_2$ violating~(\ref{it:paths})
is excluded by observing that none of the end vertices of such $P_1$ and $P_2$
is a joint by the separation(s) performed, and thus, if an end vertex $v_1$ of $P_1$ is
contained in $P_2$ all the good wedges at $v_1$ are on the same side of $P_2$
in the rotation at $v_1$, and vice-versa.

To rule out a violation of~(\ref{it:paths}) when one of $P_1$ and $P_2$ was not a path before the separation
we proceed as follows.
For the sake of contradiction consider  a forbidden pair $P_1$ and $P_2$ violating~(\ref{it:paths}).
 Let $W_1$ denote the walk in $G$
that was turned by  separations  into $P_1$.  
 Let $W_2$ denote the walk in $G$
that was turned by  separations  into $P_2$.  

First,  suppose that the end vertices of both $W_1$ and $W_2$ are different.
Let $P_1'$ and $P_2'$ denote the path contained in $W_1$ and $W_2$, respectively, connecting its end vertices. 
By Lemma~\ref{lemma:symmDif}, we have $i_A(e(P_1'),e(P_2'))=i_A(e(W_1),e(W_2))$, where $e(.)$ is defined for walks in the same way as for paths.
Moreover, by the definition of $i_A$ we have  $i_A(e(P_1),e(P_2))=i_A(e(W_1),e(W_2))$ (contradiction with~(\ref{it:paths}) before the separations).

\begin{figure}[htp]
\centering
\includegraphics[scale=0.7]{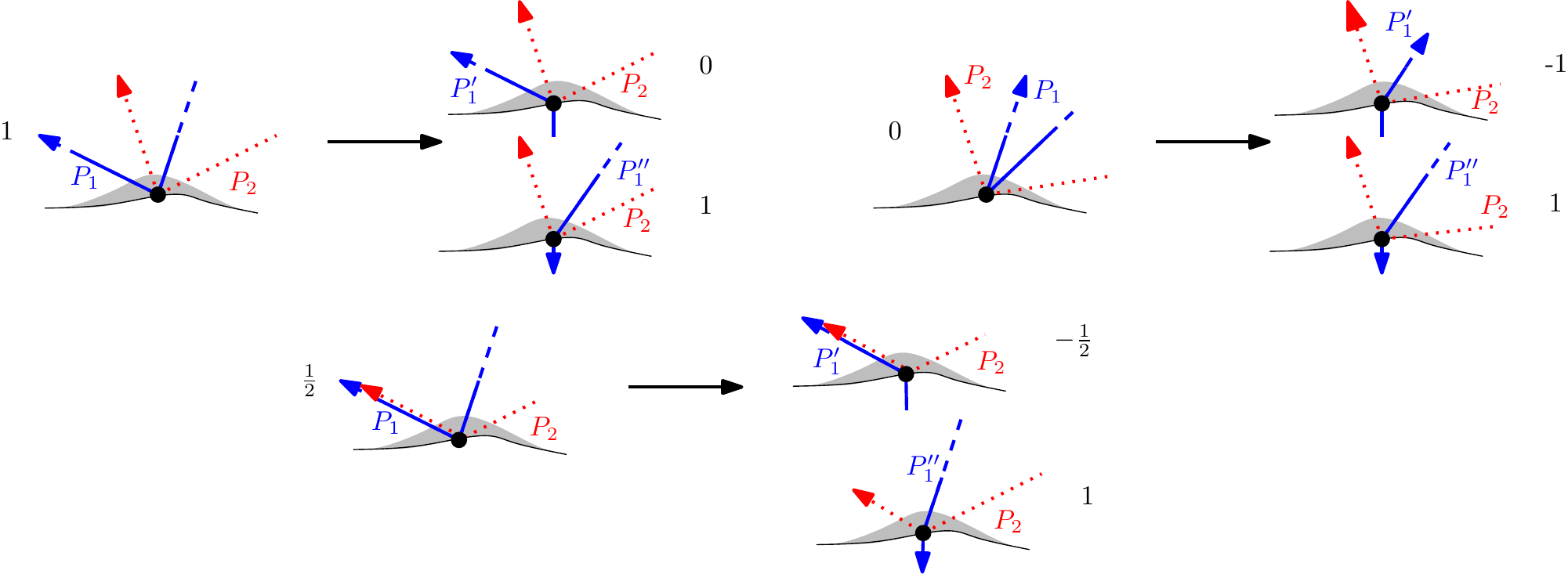}
\caption{The path $P_1$ after we divide it into two parts $P_1'$ and $P_1''$ at a former joint, and the corresponding contributions towards algebraic intersection numbers with $P_2$. Note the short additional edges that come from the definition of the forbidden pair of paths.}
\label{fig:patch}
\end{figure}

 Hence, we assume that the end vertices of $P_1$ were identified in $W_1$. 
 Let us choose $P_1$ so that its length is smallest possible.
We claim that $W_1$ is, in fact, a cycle passing only through one joint of $G$ that underwent a separation splitting vertices of $W_1$.
Indeed, otherwise we can divide $W_1$ into two parts ending in joints, one of which either violates the choice of $P_1$, or violates property~(\ref{it:paths}) (together with $P_2$)  in $G$.  The last fact follows, since if we divide $W_1$ into $W_1'$ and $W_1''$, we have $i_A(e(W_1),e(W_2))=i_A(e(W_1'),e(W_2))+i_A(e(W_1''),e(W_2))$, for the corresponding pairs, see Figure~\ref{fig:patch},
and $i_A(e(W_1'),e(W_2))=i_A(e(P_1'),e(P_2))$ and
$i_A(e(W_1''),e(W_2))=i_A(e(P_1''),e(P_2))$, where $P_1'$ and $P_1''$ is obtained from $W_1'$ and $W_1''$, respectively, by separations.

Hence, we assume that $W_1$ is a cycle.
Recall that $P_1$ joins a pair of vertices in $V_1$.  
Since we eliminated the joints by separations, each end vertex $v_e$ of $P_2$ is incident to only one face in $F'$ 
after the separations. Moreover, the faces in $F'$, that $v_e$'s are incident to before the separations, are all in the exterior of $W_1$ due to~(\ref{it:interior}).
Hence, if $P_2$ was a path before separations we have $i_A(W_1,e(P_2))=0$,
and a consideration as in Figure~\ref{fig:patch} gives $i_A(e(P_1),e(P_2))=0$ after the separation
(contradiction).
If $P_2$ was a cycle before the separations we obtain  again $i_A(e(P_1),e(P_2))=0$ by the same token.
\end{proof}

After the separation $G$ could split into connected components. Note that we can treat the connected components separately.
Thus, we just assume that $G$ is connected.
We have also one more condition to satisfy in the induction that follows.
\begin{enumerate}
\item[(c)]
\label{it:jointsInM}
All the joints are in $M$.
\end{enumerate}

Let $G'=G[F_0]$ denote a sub-graph of $G$ induced by $F_0\subseteq F$ in an $\approx$-class for which
the label of the vertices in $V(G)\cap V_2$ (which is the same for all the vertices in such a sub-graph)
is maximized and under that condition the label of the vertices in $V(G)\cap V_1$ is minimized. We choose $G'$
so that no other sub-graph in $G$ that can play the role of $G'$ contains $G'$ in the closure of the interior of its inner face.

Let $G''$ denote the sub-graph of $G$ induced by $E-E(G')$. Note that $G''$ might be an empty graph.
Let $J_1$ and $J_2$, respectively, denote the set of joints of $G$ belonging to the intersection of $V(G')$ and $V(G'')$ in $V_1$ and $V_2$.
Note that $J_1$ and $J_2$ are all marked due to condition~(c).

\paragraph{Partition of $M$ into $M'$ and $M''$.}
Refer to Figures~\ref{fig:facesProof} and~\ref{fig:facesProof2}.
In what follows we map each joint in $J_1 \cup J_2$ either to $G'$ or $G''$. A vertex in $M$ not belonging to $J_1 \cup J_2$
is mapped to $G'$ if it belongs to $V(G')$
and to $G''$ if it belongs to $V(G'')$.
Let $G_0$ denote a sub-graph of $G''$ corresponding to the union of the faces in an $\approx_1$-class sharing a vertex of $V_1$ with $G'$.
If all the vertices of $G_0$ in $V_1$  are in $M$, we map an arbitrary joint between $G_0$ and $G'$ to $G'$,
and all the other joints to $G''$. Otherwise, if not all the vertices  of $G_0$ in $V_1$ are in $M$, we map all the joints between $G_0$ and $G'$ to $G''$.
Similarly, we handle $\approx_2$-classes of $G''$ sharing a joint in $V_2$ with $G'$.


We perform separations in $G''$ at the joints in $M$  that were mapped to $G'$.
We denote by $\mathcal{G}''$ the resulting graph.
We show that in the resulting partition of $M$ into $M'$ and $M''$ corresponding to the mapping of $M$ to $\{G',G''\}$, where $M'\subseteq V(G')$ and $M''\subseteq V(\mathcal{G}'')$,
the size of $M'$ is bounded from below by the number of faces in $F$ belonging to $G'$; and that the hypothesis of the lemma is satisfied for each connected component of $\mathcal{G}''$,
where the marked vertices are those belonging to $M''$.
Clearly, once we establish that the lemma follows.


\paragraph{Proof of the properties (a)--(c) for $M''$ in $\mathcal{G}''$.}
Note that all the joints are again in $M''$, since we performed separations at all the joints that were not in $M''$.
Thus, condition~(c) is satisfied.
By property~(\ref{it:cycle}) of the labeling $\gamma$ and the choice of $G'$, the graph $G'$ is not contained in the closure of the interior of an internal face of a sub-graph of  $\mathcal{G}''$ induced by faces in  an $\approx_1$ or $\approx_2$ equivalence class.
Thus, property~(\ref{it:interior}) is satisfied by $M''$ in $\mathcal{G}''$.
Property~(\ref{it:paths}) follows immediately, if all the joints from $J_1 \cup J_2$ in $G_0$ were mapped to ${G}''$.
Otherwise, either we have only one non-marked vertex in $V_1$ or $V_2$ in $G_0$, where a path of a forbidden pair could end
which is impossible, or there exist walks in $G_0$ passing through a vertex in $M$ mapped to $G'$ that are turned into  a forbidden pair of paths by separations.
However, this is ruled out by Claim~\ref{claim:nonViolate}.

\paragraph{Base case for $G'$.}
We check that the number of vertices in $M'$ is bounded from below by the number of faces in $F$ for each connected
component of $G'$ obtained after separations,
or equivalently that
\begin{equation}
\label{eqn:euler2}
|U'|\le \sum_{f'\in F' \ \mathrm{of} \  \mathcal{G}'}(|f'|/2 -1)+2,
\end{equation}
where $U'=V(G')\setminus M'$.

Note that $G'$ is connected.
We will prove~(\ref{eqn:euler2}) by showing that on the outer face $f_e$ of $G'$ we have at most $|f_e|/2 +1$ vertices in $U'$,
and on an internal face $f_i$ of $G'$ not belonging to $F$ we have at most $|f_i|/2-1$ vertices in $U'$. Summing that over all
faces of $G'$ not belonging to $F$ gives~(\ref{eqn:euler2}).
Indeed, by~(\ref{it:interior}) a vertex not incident to a  face in $F'$
must be in $M$.

For the outer face $f_e$ of $G'$, let $G_0$ denote (as above)
a sub-graph of $G''$ corresponding to the union of the faces in an $\approx_1$-class and $\approx_2$-class, respectively, sharing  with $G'$
vertices of $J_1$ and $J_2$ on the outer face of $f_e$. Let us assume the former. The other case is analogous.
Let $v_1,\ldots, v_k$ denote the joints  between $G_0$ and $G'$.
All the vertices $v_1,\ldots, v_k$ are in $V_1$.
Let $W_0$ denote the walk between $v_1$ and $v_k$ contained in $f_e$,
and internally disjoint from the outer face of $G_0\cup G'$.

\bigskip

\begin{figure}[htp]
\centering
\includegraphics[scale=0.7]{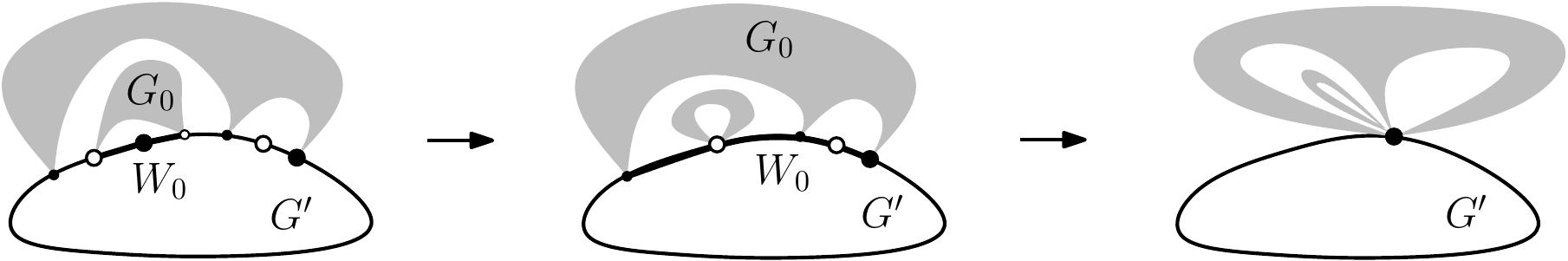}
\caption{Two consecutive contractions of minimal walks $W_0$. Enlarged vertices in the interiors of $W_0$'s are necessarily in $M$, and hence also in $M'$, due to  property~(\ref{it:interior}).}
\label{fig:contractWalk}
\end{figure}

Refer to Figure~\ref{fig:contractWalk}.
Let us consider $W_0$ as a multiset of vertices where each vertex appears as many times as it is visited by the walk $W_0$.
Assume that $G_0$ was chosen so that the length of $W_0$ is minimal.
Note that the cardinality $|W_0|-1$ is even.
Exactly $(|W_0|-1)/2$
vertices (counted with multiplicities) of $W_0$ belong to $V_2$  all of which are in $M'$.
 Indeed, by the choice of $W_0$ none of these vertices is incident to $G''$. By property~(\ref{it:interior})
  and the fact that these vertices have the maximal label at an internal face of $G_0\cup G'$,
and hence also of $G_0 \cup {G}'$, they are marked and belong to $M'$.
The walk $W_0$ corresponds to a portion of $f_e$ of length $|W_0|-1$, in which every  vertex of $V_2$ is mapped to $G'$.

Thus, the corresponding portion of $f_e$ contributes at most $|W_0|-(|W_0|-1)/2=(|W_0|-1)/2+1$ vertices to $U'$ if all the vertices
of $J_1$  on $W_0$ were mapped to $G''$, and at most $(|W_0|-1)/2$ if all the vertices of $J_1$  in $W_0$ but
one were mapped to $G''$.
 Hence, in the light of what we want to prove, $W_0$ can play a role of a vertex in $V_1$  that belongs to $M''$, if all the vertices of $J_1$ on $W_0$
 were mapped to $G''$, and that belongs to $M'$, if all the vertices of $J_1$  on $W_0$ but one
 were mapped to $G''$.

 Note that if the end vertices of the walk $W_0$ are not incident to the outer face of $G'\cup G_0$ then
  by property (\ref{it:interior}) and the mapping of marked vertices to $G'$ and $G''$, $W_0$  necessarily corresponds to a vertex that belongs to $M'$.
  Thus, if $W_0$ is not maximal, the contracted vertex is mapped to $G'$.

It follows that until there does not exist a minimal walk $W_0$, we can keep contracting minimal walks $W_0$ of $f_e$
into the vertices that either correspond to a vertex in $M'$ or in $M''$.
A new vertex belongs to $V_1$ and $V_2$,  respectively, if the first and last vertex of its corresponding walk belongs to $V_1$ and $V_2$.
Let $G_{aux}$ denote the resulting auxiliary graph.
Let $f_e'$ denote the outer face of $G_{aux}$.
We summarize the discussion from the above in the next claim.

\begin{claim}
\label{claim:contract}
If $G_{aux}$ has at most $|f_e'|/2+1$ vertices in $U'$
incident to $f_e'$ then $G'$ has at most
$|f_e|/2+1$ vertices in $U'$
incident to $f_e$.
\end{claim}

 We show that $G_{aux}$ has at most $|f_e'|/2+1$ vertices in $U'$ incident to $f_e'$. 
 It will follow by Claim~\ref{claim:contract} that in $G'$ we have
 at most $|f_e|/2+1$ vertices in $U'$ incident to $f_e$
 as desired.

\bigskip

\begin{figure}[htp]
\centering
\includegraphics[scale=0.7]{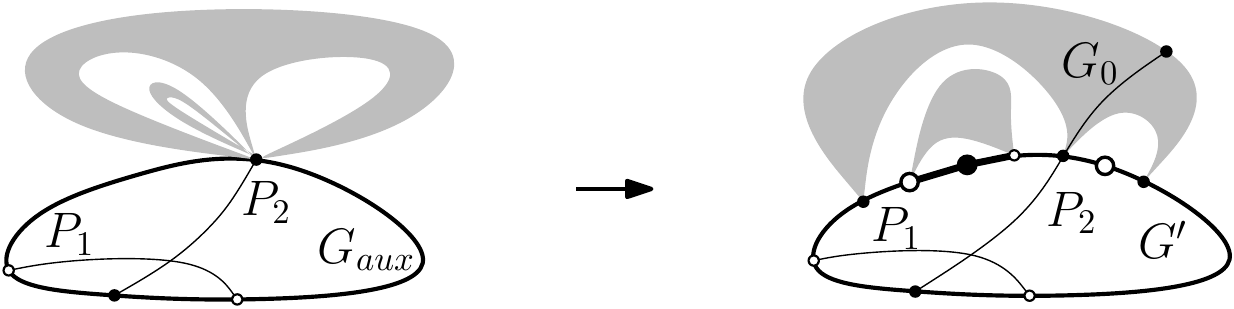}
\caption{An unfeasible interleaving pair $P_1$ and $P_2$ of an $i$-cap and $j$-cup in $G_{aux}$ and its corresponding forbidden pair of paths in $G$.}
\label{fig:extend}
\end{figure}

\begin{claim}
 $G_{aux}$ has at most $|f_e'|/2+1$ vertices in $U'$ incident to $f_e'$. 
\end{claim}
\begin{proof}
We proceed by induction on the number of cut-vertices
of $f_e'$ treated as a sub-graph of $G_{aux}$. 
If $f_e'$ is a cycle
 we show that $G_{aux}$ cannot contain four vertices $v_1,v_2,v_3,v_4\in U'$ in this order on $f_e'$ such that
$v_1,v_3\in V_1$ and $v_2,v_4\in V_2$.
Such a four-tuple gives rise to a forbidden pair $P_1$ and $P_2$ of
$i$-cap and $j$-cup thereby violating property~(\ref{it:paths}) in $G_{aux}$.
Nevertheless, a non-marked vertex in $G_{aux}$ represents either a vertex that was previously non-marked,
or a set of joints between $G'$ and $G''$. Moreover, in the latter the corresponding sub-graph $G_0$ of $G''$
attached at those joints contains an unmarked vertex with the same label. Thus, $P_1$ and $P_2$ can be 
extended to a forbidden pair of paths also in $G$ (see Figure~\ref{fig:extend} for an illustration).
Hence, by Observation~\ref{obs:simple} at most $|f_e'|/2+1$ vertices incident to $f_e'$ belongs to $U'$ if $f_e'$ is a cycle which concludes the base case.

In the inductive case we turn $f_e'$ into a walk $W_e'$ in which we represent each marked vertex 
as many times as it appears on $W_e$ by different vertices. 
An unmarked joint $v$ of $G'$ in $W_e'$ splits
$W_e'$ into two walks $W_1$ and $W_2$.
By the induction hypothesis $W_1$ has at most 
$|W_1|/2+1$ non-marked vertices, and $W_2$ has at most
$|W_2|/2+1$ non-marked vertices. 
Since $v$ is shared by $W_1$ and $W_2$ we have
at most $|f_e'|/2+1=|W_1|/2+1+|W_2|/2+1-1$ non-marked
vertices incident to $f_e'$,  and the desired upper bound on $U'$ follows.
\end{proof}

For an internal face $f_i$ of $G'$ not belonging to $F$ we show, in fact, that after contracting all walks $W_0$ defined
 similarly as above no vertex incident to $f_i$ is in $U'$.
Observe that, by the choice of $G'$, at an internal face $f_i$ of $G'$, all the vertices in $V_1$ have the label $\min(f_i)$,
and all the vertices in $V_2$ have the label $\max(f_i)$.
By property~(\ref{it:interior}), $M$ contains all the vertices with label $\min(f_i)$ incident to or in the interior of $f_i$.
Hence, $G_0$ (defined as above) living in the interior of an internal face of $G'$ has all the vertices of $V_1$ in $M$.
Thus, a joint between $G_0$ and $G'$ was mapped to $G'$.
Similarly, we argue in the case, when the joints between $G_0$
and $G'$ are in $V_2$, and $G_0$ is contained in an internal face of $G'$.
It follows that we can assume that all the vertices of $f_i$ are in $M'$. Since $f_i$ is always incident to at least $2\le f_i/2+1$ vertices,
the lemma follows.
\end{proof}

\section{Characterization of normalized  embedded strip planar graphs}

\label{sec:normalized}

In this section we prove our characterization of strip planar clustered graphs
in a special case.
In the next section we establish the results in general by reducing it to the
case considered in the present section.

An embedded strip clustered graph $(G,T)$ with a given outer face is \emph{normalized}
if (i) every cluster induces an independent set; and (ii) every
internal face is either simple or semi-simple and the outer face is simple.

\begin{lemma}
\label{lemma:characterizationSpecial}
Let $(G,T)$ denote a normalized strip clustered graph.
$(G,T)$ is c-planar if and only if $(G,T)$ does not contain
an unfeasible pair of paths, or a trapped vertex.
\end{lemma}

\begin{proof}
The ``if'' part follows directly by Observation~\ref{obs:crossing} and~\ref{obs:trap}.
For the ``only if'' part we first show that we can assume that $G$ is vertex two-connected.

\paragraph{Cut vertices.}
Refer to Figure~\ref{fig:cutvertex}.
Suppose that a cut vertex $v$ is incident to a semi-simple face $f$, whose facial walk $W=vW_1vW_2v$ contains two occurrences of $v$. It holds that either $vW_2v$ is
contained in the closure of the interior of the cycle $vW_1v$ in our embedding of $G$ or vice-versa. Suppose that $vW_1v$ is
contained in the closure of the interior of the cycle $vW_2v$. We divide $G$
into two components $G_1$ and $G_2$ both of which contains $v$. The graph $G_1$ has $vW_1v$ as the outer face, thus, $G_1$ is obtained from $G$ by deleting the exterior
of the cycle $vW_1v$, and $G_2$ is obtained from $G$ by deleting the interior of the cycle $vW_2v$.
We denote by $(G_1,T_1)$ and $(G_2, T_2)$ the corresponding strip clustered graphs inherited from $(G,T)$.
Clearly, the hypothesis of the lemma is satisfied for both $(G_1,T_1)$ and $(G_2, T_2)$ since it is satisfied for $(G,T)$.
Moreover,  strip clustered embeddings of $(G_1, T_1)$ and $(G_2,T_2)$ can be combined thereby obtaining a strip clustered embedding of $(G,T)$.
Thus, we assume that $G$ is two-connected.

\paragraph{Towards an application of Lemma~\ref{marriage:lemma}.}
As explained in Section~\ref{sec:labelings} we combine the  marriage condition of Lemma~\ref{marriage:lemma} with the characterization of upward planar graphs to prove c-planarity of $(G,T)$. To this end we first alter our graph so that Lemma~\ref{marriage:lemma} is applicable.
Let $F$ denote the set of internal semi-simple faces $F$.
Label the vertices of $G$ by $\gamma$.
We would like to modify $(G,T)$ without introducing a trapped vertex or an unfeasible pair so that in the obtained modification $(G',T')$ after suppressing
each vertex of degree two that is neither minimum nor maximum of a face in $F$, the faces of $F$ are fancy in $G'[F]$.
Moreover, we require that in the modification $(G',T')$
the incidence relation between sources and sinks of $\overrightarrow{G'[F]}$ on one side, and internal semi-simple faces of $G'[F]$ on the other side
is isomorphic to the same relation for $G[F]$. We remark that the obtained modification $(G',T')$ does not necessarily have vertex sets corresponding to clusters
independent.

\begin{figure}[htp]
\centering
\subfigure[]{\includegraphics[scale=0.7]{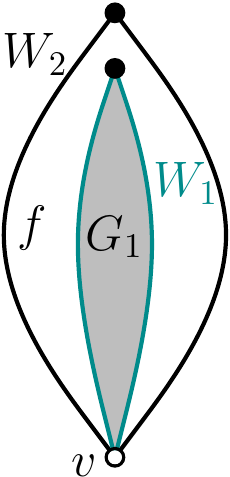}
	\label{fig:cutvertex}
	} \hspace{30px}
\subfigure[]{
\includegraphics[scale=0.7]{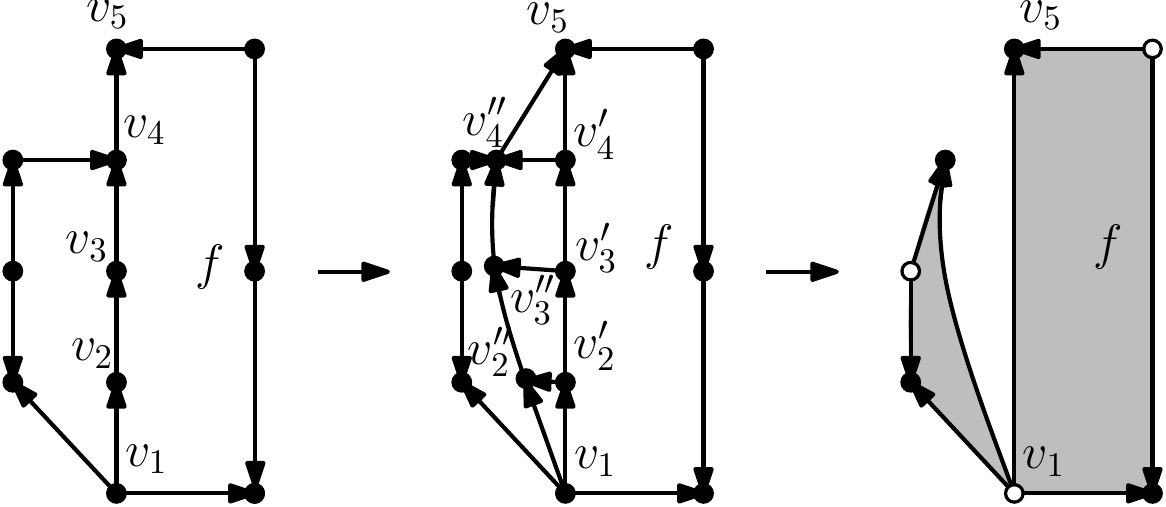}
\label{fig:doubling}}
\hspace{30px}
\subfigure[]{
\includegraphics[scale=0.7]{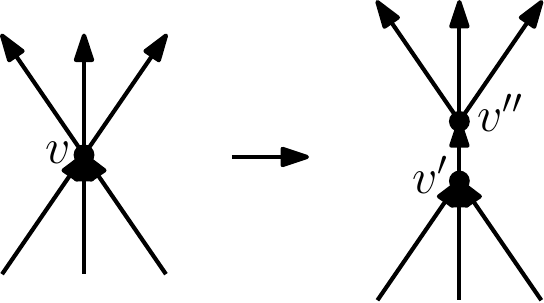}
\label{fig:bipartite}}

\caption{(a) Splitting $G$ into $G_1$ and $G_2$ at the cut-vertex $v$; (b) Doubling of the path $v_1\ldots v_5$; (c) Splitting a vertex that is neither sink nor source of $G'[F]$.}
\end{figure}

\paragraph{Making the semi-simple faces fancy.}
Refer to Figure~\ref{fig:doubling}.
Let $P$ denote a strictly monotone path with respect to $\gamma$
joining a minimum with a maximum along the boundary of a semi-simple face $f$. Suppose that an internal vertex of a path $P=v_1v_2\ldots v_l$
is a (local) minimum or maximum of another face in $F$. In other words, the path $P$ is preventing semi-simple faces of $G[F]$ to
form a set of fancy faces as defined in Section~\ref{sec:labeling}.
 We ``double'' the path $P=vPv_ju$ as follows. We apply the operation of vertex split (as defined in Section~\ref{sec:prelim}) to each internal vertex $v_i$ of $P$
thereby splitting it into two vertices $v_i'$ and $v_i''$ joined by an edge such that in the resulting graph $v_i'$ has degree three and is still incident to $f$.
We put $\gamma(v_i')=\gamma(v_i'')=\gamma(v_i)$.
The vertex $v_i''$ is drawn outside of $f$ and is adjacent to all the neighbors of $v_i$. The splitting is performed for each internal vertex $v_i$ of $P$ without introducing any pair of crossing edges  and while preserving the order in which edges leave newly created vertices.

Suppose that an unfeasible interleaving pair of paths $P_1$ and $P_2$ was introduced by the previous modification.
Note that we can assume that both $P_1$ and $P_2$ do not end in a vertex $v_i''$. Indeed, if that is the case, we shortcut or prolong them
so that they end in $v_i'$. This could not turn $P_1$ or $P_2$ into a cycle, since we would have $1=i_A(P_1,P_2)=i_A(v_i''P_1v_i'v_i'',P_2)=0$.
We turn $P_1$ and $P_2$, respectively, into an unfeasible interleaving pair $P_1'$ and $P_2'$ such that they are both internally disjoint from, let's say,
$v_1''v_2''\ldots v_l''$ (contradiction). We prove this by induction on the size of the edgewise intersection
of $P_1\cup P_2$ with $v_1'',\ldots,v_l''$. Suppose that $P_1$ passes through $v_i''v_{i+1}''$.
 We have $i_A(P_1v_i''v_{i+1}''v_i''v_i'v_{i+1}'v_{i+1}''P_1, P_2)=i_A(P_1,P_2)$ by Lemma~\ref{lemma:symmDif}.
Moreover, the walk $P_1v_i''v_{i+1}''v_i''v_i'v_{i+1}'v_{i+1}''P_1$  contains a path $P_1''$ joining the same pair of vertices as $P_1$ not passing through
 $v_i''v_{i+1}''$, and having edgewise a strictly smaller intersection with $v_1'',\ldots, v_l''$.
 By Lemma~\ref{lemma:symmDif} we have $i_A(P_1'',P_2)=i_A(P_1v_i''v_{i+1}''v_i''v_i'v_{i+1}'v_{i+1}''P_1,P_2)=i_A(P_1,P_2)$. By repeating the argument with $P_2$ and Lemma~\ref{lemma:conjChar} we obtain a desired interleaving pair $P_1'$ and $P_2'$.

No trapped vertex $v$ was introduced as well. To this end note that the trapped vertex $v$ or a cycle
$C$ witnessing this can be assumed to be disjoint from $v_1'',\ldots, v_l''$.
Thus, no unfeasible interleaving pair of paths or a trapped vertex was introduced by our modifications in $(G',T')$.

Finally, we pick an arbitrary orientation for each edge $v_i'v_i''$ in $\overrightarrow{G'}$, $i\not=1,l$. Note that the newly created
faces are simple, and hence, they do not belong to $F$.
Let $(G',T')$ denote a strip clustered graph obtained from $G$ after doubling all paths $P$  joining a minimum with a maximum along a semi-simple face $f$.
Thus, after splitting all problematic paths $P$ we have the same incidence relation between sources and sinks of $\overrightarrow{G'}$ and semi-simple faces of $G'$ as in $G$.

However, $F$ still does not have to form a set of fancy faces after suppressing
each vertex of degree two that is neither minimum nor maximum of a face in $F$,
since a vertex can be simultaneously a minimum and maximum of a face in $F$.
This would violate condition~(\ref{it:v1v2}) of fancy faces.
By our conditions, in $\overrightarrow{G[F]}$ the incoming and outgoing edges do not alternate at any vertex $v$.
Refer to Figure~\ref{fig:bipartite}.
Thus, we can apply a vertex split to each vertex that is simultaneously a minimum and maximum of a face in $F$.
 We turn a vertex $v$ into two vertices $v'$ and $v''$ contained in the same cluster,
where $v'$ has only one outgoing edge and $v''$ has only one incoming edge, namely $v'v''$.
Hence, this operation does not introduce semi-simple faces in $G'$ and
does not affect the incidence relation between sources and sinks of $G'$, and faces in $F$.
Using the notation of Section~\ref{sec:labeling}, the vertices of $V_1$ and $V_2$,
 respectively, correspond to sources and sinks of $\overrightarrow{G'[F]}$.


The last adjustment of $(G,T)$ we need is to make sure that the source and sink incident to the outer face
is not incident to a semi-simple face. Thus, if a source $v$ (a sink is treated analogously)
incident to the outer face is incident to a semi-simple
face, we introduce an additional vertex $u$ joined with $v$ by an edge that belongs to a new cluster so
that $v$ is not a source anymore and $u$ is the new source on the outer face.
Afterwards we add a strictly monotone path joining $u$ with the sink on the outer face so that the
resulting clustered graph is still strip clustered.
Clearly, $u$ is not incident to any semi-simple face, and the last modification does not introduce
an unfeasible pair or a trapped vertex.

%

\paragraph{Upward digraphs.}
In order to simplify the notation we let $(G,T)$ denote the modification $(G',T')$ of $(G,T)$ obtained previously.
By our assumption and Observation~\ref{obs:alternate} in $\overrightarrow{G}$ the incoming and outgoing edges do not alternate at any vertex $v$ .
Thus, the embedding of $\overrightarrow{G}$ is a candidate embedding of $\overrightarrow{G}$.
 \cite[Theorem 3]{BBLM94} implies that our embedding of  $\overrightarrow{G}$
 admits an upward-planar embedding if there exists a mapping of sources and sinks to the faces of $\overrightarrow{G}$
such that

\begin{enumerate}[(i)]
\item
every sink or source of $\overrightarrow{G}$ is mapped to an internal semi-simple face, or to the outer face it is incident to;
\item
each internal face has exactly one source or sink mapped to it, if it is semi-simple, and zero otherwise; and
\item
the outer face has exactly a source and a sink mapped to it.
\end{enumerate}


By Hall's theorem there is
a mapping of sources and sinks of $\overrightarrow{G}$ to the faces of $\overrightarrow{G}$
satisfying (ii), if every subset of $s$ internal semi-simple faces is incident to at least
 $s$ sources or sinks.
Let $F$ denote a connected subset of internal semi-simple faces in $\overrightarrow{G}$.
Let $G'=G[F]$ be the sub-graph of $G$ induced by faces in $F$.

We would like to apply Lemma~\ref{marriage:lemma} to $F$ so that the set of marked vertices $M$ in $G'$ contains all the sinks
and sources of $\overrightarrow{G}$ in $G'$.
The embedding of $G'$ is inherited from our given embedding.
In order to apply the lemma we first
suppress the vertices that are  neither minimum nor maximum of a face of $F$
In what follows we show that the hypothesis of Lemma~\ref{marriage:lemma} is satisfied.
Note that by our modification the two unmarked vertices on the outer-face of $G$ cannot be incident to a face in $F$.


To prove the property~(\ref{it:interior}) consider a pair of a cycle $C$ and a vertex $v$ violating it.
Assume that $\gamma(v)=\max (C)$. The other case is treated analogously.
We have a vertex $u$ joined with $u$ by an edge $\gamma(u)=\gamma(v)+1$
belonging to the interior of $C$. The vertex $u$ is trapped in the interior of $C$ (contradiction).

It remains to show that property~(\ref{it:paths}) is satisfied.
Furthermore, we claim that $\overrightarrow{G'}$ does not contain a forbidden pair of an  $i$-cap $P_1$ and $j$-cup $P_2$, where $i<j$.
 Indeed, let $P_1'$ denote
a path obtained from $P_1$ by appending to both its ends an edge of $G$ joining its end vertex with a vertex in the $(i-1)$-st cluster.
This is possible, since $P_1$ ends in a non-marked vertex with a good
wedge in $G'$.
In an analogous manner we construct $P_2'$, where the end vertices of $P_2'$ belong to $(j+1)$-st cluster.
The paths $P_1'$ and $P_2'$ form an unfeasible pair of paths in our given embedding contradicting our assumption due to properties of a forbidden pair.

Thus, it is left to show that our mapping can be extended to a mapping satisfying (i) and (iii).
The condition (iii) is easy as we have one source and one sink left for the outer face.
In order to show (i) it is enough to prove that we do not have  more sources and sinks than
required by all the faces.
However, \cite[Lemma~5]{BBLM94} directly implies that this is exactly the case.
Since the conditions (i)--(iii) hold for our modified $(G,T)$, they have to hold also for the graph
we started with. Indeed, the incidence relations between semi-simple faces, and sinks and sources before and after the modification are isomorphic,
except possibly for a sink or source that were previously on the outer face.
In the modified $(G,T)$ such sink or source is not mapped to any face, and hence, it can substitute the
newly introduced sink or source on the outer face.

Finally, we show how to turn the planar upward drawing of $\overrightarrow{G}$ into a clustered drawing of $(G,T)$.
Consider an upward straight-line embedding of $\overrightarrow{G}$, whose existence is guaranteed by Theorem~3 from~\cite{BBLM94}.
(Of course, we do not need the embedding to be straight-line, but rather we just stick to the formulation of~\cite{BBLM94}.)
Similarly as in~\cite{ADDF13+} we augment further the embedding of $\overrightarrow{G}$ by adding an edge
inside every semi-simple face $f$ connecting two minima, if a minimum has a non-convex angle inside $f$, and connecting two maxima of $f$,
if a maximum has a non-convex angle inside $f$. We orient the added edges so that they point upward.
The obtained directed graph, let us denote it by $\overrightarrow{G_0}$, has exactly one source $s$ and one sink $t$ that are both incident to the outer face.

We start constructing a clustered drawing of $(G_0,T)$  by drawing an arbitrary directed $s-t$ path such that the resulting drawing is clustered.
In each subsequent step the left-to-right order of the incoming and outgoing edges at each vertex is the same as in our upward straight-line embedding of $\overrightarrow{G_0}$.
In a single step we draw a directed path $P$ in the exterior of the already drawn part joining a pair of vertices on its outer face
such that the number of inner faces is increased by one. Here, we require the new inner face to be also the face in the final embedding. This can be performed while preserving the properties of strip clustered embeddings which concludes the proof.
\end{proof}

\section{Characterization of embedded strip planar graphs}

\label{sec:char}

In this section we prove our characterization of strip planar clustered graphs
by reducing a general instance of strip clustered planarity to a normalized one.
We remark that the normalized instances which are (vertex) two-connected are \emph{jagged instances} from~\cite{ADDF13+}.  Therein
it is also proved that for every instance of strip clustered planarity $(G,T)$, where $G$ is given by an embedding, there exists a finite set of instances $\mathcal{I}$ all of which are jagged such that $(G,T)$ is strip planar if and only if every instance in $\mathcal{I}$ is strip planar.
As a byproduct of our work we obtain an alternative proof of this fact.

\begin{theorem}
\label{thm:characterization}
Let $(G,T)$ denote an embedded strip clustered graph.
$(G,T)$ is strip planar if and only if $(G,T)$ does not contain
an unfeasible interleaving pair of paths, or a trapped vertex.
\end{theorem}

Before we turn to the proof of Theorem~\ref{thm:characterization} we discuss its
relation to  Conjecture~\ref{conj:conj1}.
We note that the theorem is, in fact, stronger than the conjecture for $n=2$
due to a more restricted condition on pairs of paths we consider.
However, the strengthening is not significant, since it is rather a simple
exercise to show that forbidding an unfeasible interleaving pair of paths
and trapped vertices renders the hypothesis of the conjecture satisfied.

\begin{wrapfigure}{r}{.3\textwidth}
\centering
\includegraphics[scale=0.7]{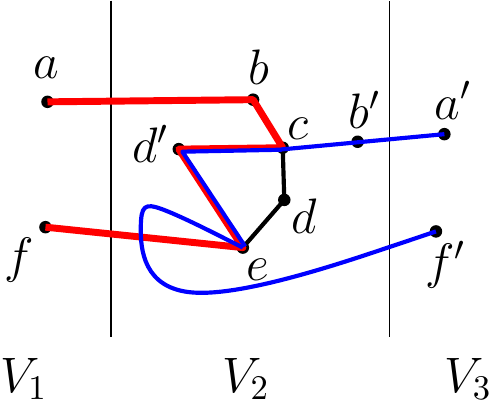}
\caption{Replacing the portion of $P_1'=abcdef$ with the portion of $P_2'=a'b'cd'ef'$ on the cycle $C=cded'$.}
\label{fig:ffEx}
\end{wrapfigure}

Indeed, if a pair of intersecting paths $P_1$ and $P_2$ satisfies $\gamma(P_1) \cap \gamma(\partial P_2) = \emptyset$ and $\gamma(\partial P_1) \cap \gamma(P_2) = \emptyset$,
there exist sub-paths $P_1'$ of $P_1$ and $P_2'$ of $P_2$ such that $i_A(P_1,P_2)=i_A(P_1',P_2')$ that either form an interleaving pair
or do not form an interleaving pair only because they do not intersect in a path.
In the latter, no end vertex of $P_1$ or $P_2$ is contained in the interior
of a cycle in $P_1' \cup P_2'$ due to the non-existence of trapped vertices.
Let $W_2$ be a walk obtained from $P_2'$ by replacing its portion on a cycle $C$
contained in $P_1' \cup P_2'$, such  that $P_1' \cap C$ is a path, with the portion of $P_1'$ for every such cycle (see Figure~\ref{fig:ffEx}). Let $P_2''$ denote the path in $W_2$ connecting
its end vertices.
We have $i_A(P_1,P_2)=i_A(P_1',P_2'')$, and $P_1'$ and $P_2''$ form an interleaving pair.
Hence, we just proved the following.

\begin{lemma}
\label{lemma:conjChar}
Given that $(G,T)$ is free of trapped vertices,
if a pair of intersecting paths $P_1$ and $P_2$ in a strip embedded clustered graph satisfies $\gamma(P_1) \cap \gamma(\partial P_2) = \emptyset$ and $\gamma(\partial P_1) \cap \gamma(P_2) = \emptyset$\footnote{In the case of paths the boundary operator $\partial$ returns the end vertices.},
there exist sub-paths $P_1'$ of $P_1$ and $P_2'$ of $P_2$ such that $i_A(P_1,P_2)=i_A(P_1',P_2'')$, where $P_2''\subset P_1' \cup P_2'$ is constructed as above, forming an interleaving pair.
\end{lemma}

The application of Theorem~\ref{thm:characterization} in Section~\ref{sec:tree} and~\ref{sec:theta} reveals that a much stronger version of Theorem~\ref{thm:characterization} holds at least in the case of trees and theta graphs.

\begin{proof}
As advertised in Section~\ref{sec:normalized} we proceed by normalizing $(G,T)$
so that Lemma~\ref{lemma:characterizationSpecial} is applicable.
First, we turn every cluster of $(G,T)$ into an independent set,
and augment $G$ so that all internal faces are either simple or semi-simple (as defined in Section~\ref{sec:even}) and the outer face is simple.
Second, we show that during the ``normalization period'' we cannot introduce an unfeasible pair or a trapped vertex.

\paragraph{Turning clusters into independent sets.}
We proceed by the induction on $\sum_{V_i}\sum_{C_j}(|V(C_j)|-1)$, the inner sum is over the connected components $C_j$ induced by $V_i$.
The base case is treated in the next paragraph.
If we have an edge $e$ in $E(G)$ between two vertices $u,v$ in the same cluster
we contract $e$ in the given embedding of $G$. The resulting drawing is still an embedding.
We apply the induction hypothesis on the resulting drawing thereby obtaining a strip clustered embedding of the corresponding embedded strip clustered graph.
The induction works since by a contraction we cannot introduce a trapped vertex or an unfeasible pair of paths as we will see in the paragraph after
the next one. We can ignore loops created by contractions, since they do not influence the algebraic intersection number
between walks. Thus, we can effectively delete them. However, we re-introduce them at the same position in the rotation
when restoring the edge $e$. Since we do not change the rotation at any step, the induction goes through.
Finally, we restore the edge $e$ by splitting the vertex $u$ into two vertices, which can be done while
keeping the embedding strip clustered.

\paragraph{Turning faces into simple or semi-simple.}
 Let $\gamma:V \rightarrow \mathbb{N}$ be a labeling of the vertices of $G$ such that $\gamma(v)=i$ if $v\in V_i$.
 Let $f$ denote an internal face of $f$ which is neither simple nor semi-simple with respect to $\gamma$. Let $W$ denote the
 (closed) facial walk of $f$. Suppose that a sub-walk $W'$ of $W$ between a global minimum $u$ and maximum $v$ of $f$ contains a local
 minimum and a local maximum both different from $u$ and $v$. We augment our drawing with a path $P$
  that subdivides $f$ and yields two new faces, both having smaller number of local minima and maxima than $f$.

Refer to Figure~\ref{fig:augmenting}.
 Let $u'\not=u$ and $v'\not=v$, respectively, denote two consecutive local minima and maxima appearing along $W'$ and assume that $u,v',u'$ and $v$ appear
 in this order along $W'$. Let $v''$ denote the vertex of $W'$ that is the closest vertex to $u'$ on $W'$ with the following properties.
 The vertex $v''$ appears on $W'$ between $u'$ and $v$ and  $\gamma(v'')=\gamma(v')$. Similarly,
 let $u''$ denote the vertex of $W'$ that is the closest vertex to $v'$ on $W'$ such that it appears on $W$ between $u$ and $v'$ and $\gamma(u'')=\gamma(u')$.
We add the path $P$ joining $u''$ and $v''$ in two steps.
 First, we add an edge $e_P$ joining $u''$ and $v''$ inside $f$, and then we subdivide $e_P$ so that the resulting embedded clustered graph is strip clustered. Let $P$ denote the resulting path.
Note that we split $f$ into a semi-simple face and a face that has a smaller number of local minima and maxima than $f$.
Thus, by subdividing faces repeatedly we eventually end up with all faces being either simple or semi-simple and the outer face simple.
The sub-walk of $W'$ between $u$ and $v$ with exactly one local minimum and maximum in its (relative) interior is \emph{covered} by $P$.

\paragraph{No unfeasible pairs of paths or trapped vertices.}
It remains to show that by contracting edges and subdividing faces we do not introduce an unfeasible pair of paths or
a trapped vertex. Clearly, by contracting an edge whose both end vertices belong to the same cluster we cannot introduce
an unfeasible pair of paths or a trapped vertex. Indeed, if that were  the case, the inverse operation of such contraction
would certainly destroy neither a trapped vertex nor an unfeasible pair (contradiction). This follows since for a pair of paths
 $i_A(P_1,P_2)=i_A(P_1',P_2')$, where $P_1'$ and $P_2'$ is obtained from $P_1$ and $P_2$ by
replacing their edges and vertices by their pre-images w.r.t. the contraction.
 Thus, suppose for the sake of contradiction that by subdividing a face
as in the previous paragraph we introduce an unfeasible pair of paths or a trapped vertex.

Refer to Figure~\ref{fig:separ2}.
First, suppose that a vertex $v$ trapped  in the interior of $C$ was introduced by  subdividing a face $f$ by a path $P$.
If $v$ lies in the interior of $P$ then $C$ is edge disjoint from $P$, and $C$ has to separate $v$ from a vertex $u$ in the same
cluster as $v$ incident to $f$ which is impossible.
Otherwise, $u$ is also trapped in $C$ (contradiction).
Hence, $P$ is contained in $C$. We perform the operation of the symmetric difference edgewise to $C$ with the cycle obtained by concatenating $P$ with
the sub-walk of the facial walk of $f$ covered by $P$. Thus, we turned $C$ into a set of closed walks one of which contains $v$ in its interior.
Hence, we obtain a cycle in the original graph such that $v$ is trapped in its interior (contradiction).

Let  $P''=vP''u$ denote the sub-walk of the facial walk $W$ of $f$ covered by $P$.

Refer to Figure~\ref{fig:separ1}.
Second, let  $P_1$ and $P_2$ be an unfeasible pair of paths obtained after we subdivided a face $f$ by a path $P$. Let $P_1'$ and $P_2'$, respectively,
 be obtained from $P_1$ and $P_2$ by replacing its sub-path $P'=vP'$ contained in $vPu$ with a shortest sub-path of $P''$ ending in the same cluster as $P'$.
  Then we show that the replacement does not change $i_A(P_1,P_2)$ which will lead to contradiction.

We deal only with $P_1$. The path $P_2$ is taken care of in the same way.
If $P'=P$, let $W'$ be the concatenation of the reverse of $vPu$, denoted by $\overline{vPu}$, with $vP''u$.
We have $i_A(P_1,P_2)=i_A(P_1uW'uP_1,P_2)$ by Lemma~\ref{lemma:symmDif}. Moreover,
$i_A(P_1uW'uP_1,P_2)=i_A(P_1'vP\overline{P}vP_1',P_2)=i_A(P_1',P_2)$ (by Lemma~\ref{lemma:symmDif}).
Note that $P_1'$ is a walk. However, by  Lemma~\ref{lemma:symmDif} the walk $P_1'$ contains a path joining the same pair of vertices
having the same algebraic intersection number with $P_2$ as $P_1'$. By repeating the argument with $P_2$ 
and Lemma~\ref{lemma:conjChar} we obtain a desired interleaving pair,
since no trapped vertex was introduced.

Otherwise, $P'\not=P$, and an end vertex of, let's say $P_1$, is contained in $P$. We join the other end vertex of $P'$ than $v$ with
the end vertex of $P''$ different from $v$ by a crossing-less edge $e$ drawn inside $f$, and contract $e$ into a vertex $u$.
Let $P_1''$ denote the path $P_1$ after we perform
the previous operation. Note that $i_A(P_1'',P_2)=i_A(P_1, P_2)$ and that $P_1''$ cannot be a cycle, since the intersection number of $P_2$
with such cycle would not be zero.
Now, we proceed as above with $P_1''$ playing the role of $P_1$, and find $P_1'$. Finally, we split $v$ into $e$ and shortened $P_1'$ by $e$.

\begin{figure}
\centering
\subfigure[]{
\includegraphics[scale=0.7]{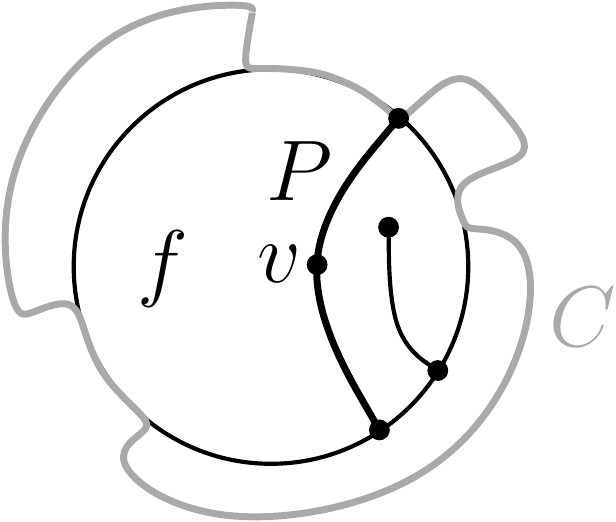}
\label{fig:separ1}}
\hspace{30px}
\subfigure[]{
\includegraphics[scale=0.7]{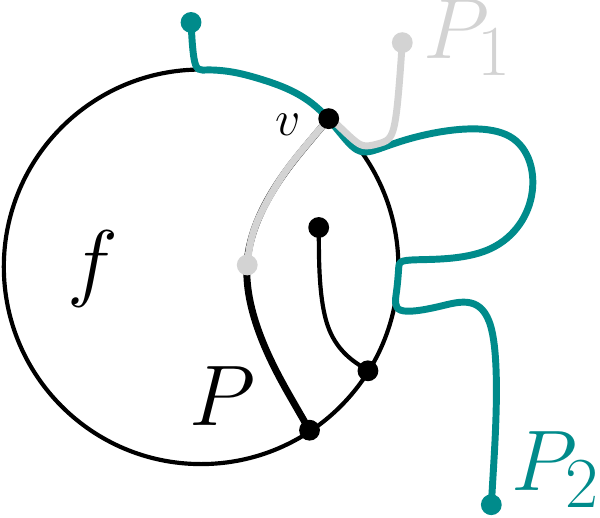}
\label{fig:separ2}}
\caption{(a) A vertex $v$  trapped inside $C$ after we subdivide $f$ by $P$; and
(b) An unfeasible pair of paths $P_1$ and $P_2$ after we subdivide $f$ with $P$. }
\end{figure}

\end{proof}

\begin{figure}[htp]
\centering
\subfigure[]{\includegraphics[scale=0.7]{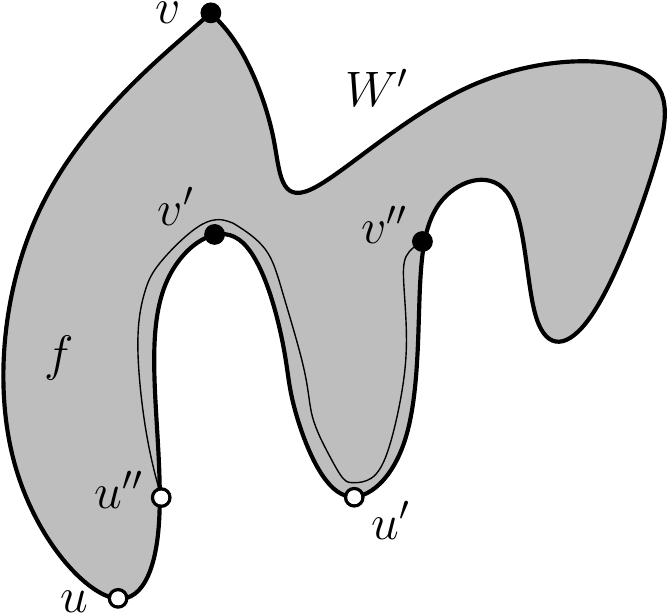}
    \label{fig:augmenting}
	} \hspace{30px}
\subfigure[]{
\includegraphics[scale=0.7]{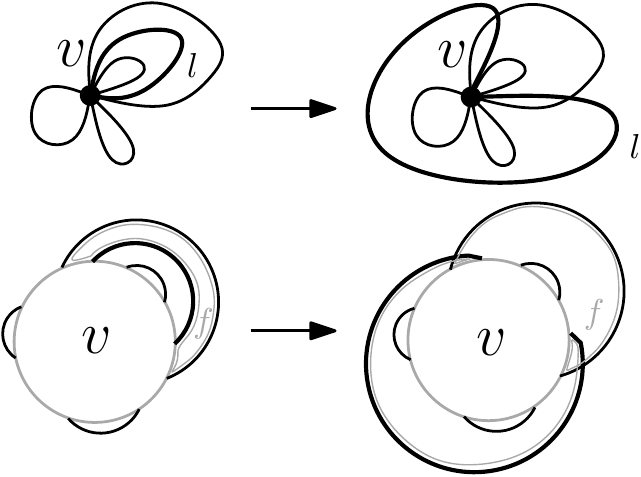}
\label{fig:obser5}}
\caption{(a) The augmenting path $u''$ between $ v''$ drawn along the contour of face $f$;
(b)  A bouquet of loops at $v$ and its modification after we pull a loop  $l$ over $v$ (top). The cyclic intervals
corresponding to the loops in the rotation at $v$ before and after we pull $l$ over $v$ (bottom). }
\end{figure}

\section{The variant of the weak Hanani-Tutte theorem for strip clustered graphs}
\label{sec:linearly}

In this section we prove the weak Hanani-Tutte theorem for strip clustered graphs, Theorem~\ref{thm:linearly}.

Given a drawing of a graph $G$ where every pair of edges cross an even number of times, by the weak Hanani-Tutte theorem~\cite{CN00,PT00,PSS06},
we can obtain an embedding of $G$ with the same rotation system, and hence, the facial structure of an embedding of $G$ is already present in an even drawing. This allows us to speak about faces in an even drawing of $G$. Hence, a face in an
even drawing of $G$ is the walk bounding the corresponding face in the embedding of $G$ with the same rotation system.

 A face $f$ in an even drawing corresponds to a closed (possibly self-crossing)
curve $C_f$ traversing the edges of the defining walk of $f$ in a close vicinity of its edges without crossing an edge that is being traversed, i.e,
$C_f$ never switches to the other side of an edge it follows.
An \emph{inner face} in an even drawing of $G$ is a face for which all the vertices of $G$ except those incident to $f$ are outside of $C_f$.
Similarly, an \emph{outer face} in an even drawing of $G$ is a face such that all the vertices of $G$ except those incident to $f$ are inside of $C_f$.
Note that by the weak Hanani--Tutte theorem every face is either an inner face or an outer face.
Unlike in the case of an embedding (in the plane), in an even drawing the outer face might not be unique. Nevertheless, an outer face always exists in an even drawing of a graph in the plane.

\begin{lemma}
\label{obs:outer-face}
Every even drawing of a connected graph $G$ in the plane has an odd number of outer faces.
\end{lemma}

\begin{proof}
Refer to Figure~\ref{fig:obser5}.
By successively contracting every edge in an even drawing of $G$ we obtain a vertex $v$ with a bouquet of loops, see e.g., the proof of ~\cite[Theorem 1.1]{PSS06}.
Let  $\mathcal{D}$ be the obtained drawing whose underlying abstract graph is not  simple
unless it is edgeless.
Let us treat $\mathcal{D}$ as an even drawing. Thus, we obtain the facial structure in $\mathcal{D}$
by traversing walks consisting of loops at $v$.
Every loop $l$ at $v$ corresponds to a cyclic interval in the rotation at $v$ containing
the end pieces of edges that are in a close neighborhood at $v$ contained inside
$l$. By treating every walk in $\mathcal{D}$ as a walk along cyclic intervals of the loops it traverses we define the \emph{winding number} of a face in $\mathcal{D}$ as
the number of times we walk around $v$ when traversing the intervals of its walk.
The winding number can be positive or negative depending on the sense of the traversal.

Note that the outer faces in $\mathcal{D}$ are those faces $f$ whose corresponding walks wind around $v$ an odd number of times. This follows because whenever we visit $v$ during a walk of $f$ winding an odd number of times around $v$ the corresponding position in the rotation at $v$ is contained inside of an even number of loops of the walk, and hence outside of $C_f$.

By pulling a loop $l$ over $v$ we flip the cyclic interval in its rotation
that corresponds to the inside of $l$. It follows that we change the winding number of both facial walks that $l$ participates in by one.
Hence, we do not change the parity of the total number of outer faces in $\mathcal{D}$.
Since a crossing free drawing of $G$ has an odd number of outer faces the lemma follows.
\end{proof}

Thus, given an even strip clustered drawing of $(G,T)$ we can associate  it with an embedding  $\mathcal{D}$ having an outer face $f$. Note that by the connectivity of $G$ the vertices
incident to $f$ span all the clusters of $(G,T)$.
By Theorem~\ref{thm:characterization} it is enough to prove that $\mathcal{D}$ does not contain an unfeasible pair of paths
or a vertex trapped in the interior of a cycle. However, due to evenness of the given drawing of $(G,T)$ both of these
forbidden substructures would introduce a pair of cycles crossing an odd number of times (contradiction).
In order to rule out the existence of a trapped vertex we use the fact that
the boundary of the outer face $f$ spans all the clusters. If a vertex $v$ is trapped in the interior of  a cycle $C$ then by the connectedness of $G$
we can join $C$ with $v$ by a path $P$ of $G$. By the evenness of the drawing it follows that the end piece of $P$ at $C$
in our drawing start outside of $C$. On the other hand, a path connecting $C$ with any vertex on the outer face $f$ must
also start at $C$ outside of $C$, since the boundary of $f$ spans all the clusters.
 Thus, $v$ cannot be trapped, since the rotation system from the even drawing is preserved
in the embedding.

\section{The variant of the Hanani-Tutte theorem for strip clustered 3-connected graphs}

\label{sec:linearlyStrong}

In this section we prove the Hanani-Tutte theorem for strip clustered graphs if the underlying abstract
graph is three connected, Theorem~\ref{thm:linearlyStrong}.

First, we prove a lemma that allows us to get rid of odd crossing pairs by doing only local redrawings and vertex splits.
A drawing of a graph $G$ is obtained from the given drawing of $G$ by \emph{redrawing edges locally at vertices}
if the resulting drawing of $G$ differs from the given one only in small pairwise disjoint neighborhoods of  vertices not containing any other
vertex. The proof of the following lemma is inspired by the proof of~\cite[Theorem 3.1]{PSS06}.

\begin{lemma}
\label{lemma:removeOdd}
Let $G$ denote a subdivision of a vertex three-connected graph drawn in the plane so that every pair of non-adjacent edges cross an even number of times.
We can turn the drawing of $G$ into an even drawing by a finite sequence of local redrawings of edges at vertices and vertex splits.
\end{lemma}

\begin{proof}
We process cycles in $G$ containing an edge crossed by one of its adjacent edges an
odd number of times one by one until no such cycle exists.
Let $C$ denote a cycle of $G$. By local redrawings at the vertices of $C$ we obtain a drawing of $G$, where every edge of $C$ crosses
every other edge an even number of times. Let $v$ denote a vertex of $C$.

First, suppose that every edge incident to $v$ and starting inside of $C$ crosses every edge incident to $v$ and starting outside of $C$ an even number of times.
In this case we perform at most two subsequent vertex splits.
If there exists at least two edges starting at $v$ inside (outside) of $C$, we split $v$ into two vertices $v'$ and $v''$ joined by a very short crossing free edge so that $v'$ is incident to the neighbors of $v$ formerly joined with
$v$ by edges starting inside (outside) of $C$, and $v''$ is incident to the rest of the neighbors of $v$.
Thus, $v''$ replaces $v$ on $C$.
Notice that by splitting we maintain the property of the drawing to be independently even,
and the property of our graph to be three-connected. Moreover, all the edges incident to the resulting vertex $v''$ of degree three or four cross one another an even number of times. Hence, no edge of $C$ will ever be crossed by another edge an odd number of times, after
we apply appropriate vertex splits at every vertex of $C$.

Second, we show that there does not exist a vertex $v$ incident to $C$ so that an edge $vu$ starting inside of $C$ crosses an edge $vw$ starting outside of $C$
an odd number of times. Since $G$ is a subdivision of a vertex three-connected graph, there exist two distinct vertices $u'$ and $w'$ of $C$ different from $v$ such that $u'$ and $w'$, respectively, is connected with $u$ and $w$ by a path internally disjoint from $C$. Let $uP_1u'$ and $wP_2w'$, respectively, denote this path. Note that $u$ can coincide with $u'$ and $w$ can coincide with $w'$.
Let $vP_3u'$ denote the path contained in $C$ no passing through $w'$. Let $C'$ denote the cycle obtained by concatenation of $P_1$, $P_3$, and $vu$.
Let $C''$ denote the cycle obtained by concatenating $P_2$ and the portion of $C$ between $w'$ and $v$ not containing $u'$.
Since $vw$ and $vu$ cross an odd number of times and all the other pairs of edges $e\in E(C')$ and $f\in E(C'')$ cross an even number of times,
the edges of $C'$ and $C''$ cross an odd number of times. It follows that their corresponding
curves cross an odd number of times (contradiction).

Notice that by vertex splits we decrease the value of the function $\sum_{v\in V(G)}deg^3(v)$ whose value is always non-negative.
Hence, after a finite number of vertex splits we turn $G$ into an even drawing of a new graph $G'$.
\end{proof}

We turn to the actual proof of Theorem~\ref{thm:linearlyStrong}.

We apply Lemma~\ref{lemma:removeOdd} to the graph $G$ thereby obtaining a clustered graph $(G',T')$, where each vertex obtained by a vertex split,
belongs to the cluster of its parental vertex and the membership of other vertices to clusters is unchanged.
By applying Theorem~\ref{thm:linearly} to $(G',T')$ we obtain a clustered embedding of $(G',T')$. Finally, we contract the pairs of vertices
obtained by vertex splits in order to obtain a clustered embedding of $(G,T)$.

\section{Monotone variant of the weak Hanani--Tutte theorem}

\label{sec:mono}

In this section we obtain the weak Hanani-Tutte theorem for monotone graphs, Theorem~\ref{thm:mono},
as a corollary of Theorem~\ref{thm:linearly}.

Given a graph $G$  with a fixed order of vertices let $\mathcal{D}$ denote its drawing such that $x$-coordinates of the vertices of $G$ respect
their order, edges are drawn as $x$-monotone curves and every pair of edges cross an even number of times.
We turn our drawing $\mathcal{D}$ of $G$ into a clustered drawing $\mathcal{D}'$ of a strip clustered graph $(G',T')$ which is still even.

We divide the plane by vertical lines such that each resulting strip contains exactly
one vertex of $G$ in its interior. Let $(G,T)$ denote the clustered graph, in which
every cluster consists of a single vertex, such that the clusters are ordered according to $x$-coordinates of the vertices.
Thus, every vertical strip corresponds to a cluster of $(G,T)$.
Note that all the edges in the drawing of  $(G,T)$ are bounded, and hence, by
Lemma~\ref{lemma:bounded} can be turned into paths so that the resulting clustered graph is strip-clustered,
and the even drawing clustered.
We denote the resulting strip-clustered graph by $(G',T')$ and drawing by $\mathcal{D}'$.

By applying Theorem~\ref{thm:linearly} to $\mathcal{D}'$, we obtain an embedding of $(G',T')$ that can be turned into an embedding of $(G,T)$ by converting the subdivided edges in $G'$ back to the edges of $G$. The obtained embedding is turned into an $x$-monotone embedding by replacing each edge $e$ with a polygonal path whose bends are intersections of $e$ with vertical lines
separating clusters in $(G,T)$.

\section{Strip trees}

\label{sec:tree}

In this section we prove the Hanani-Tutte theorem for strip clustered graphs if the underlying abstract
graph is a tree, Theorem~\ref{thm:tree}. We also give an algorithm proving Theorem~\ref{thm:ahahah}.

In order to make the present section easier do digest, as a warm-up we first prove Theorem~\ref{thm:tree} in the case, when $G$ is a subdivided star.
Once we establish Theorem~\ref{thm:tree} for stars, we show that a slightly more involved argument
based on the same idea also works for general trees.

\subsection{Subdivided stars}

\label{sec:star}

In the sequel $G=(V,E)$ is a subdivided star. Thus, $G$ is a connected graph that  contains a special vertex $v$, \emph{the center of the star},
of an arbitrary degree and all the other vertices in $G$ are either of degree one or two.
Let $(G,T)$ denote a strip clustered graph.

Recall that $\max (G')$ and $\min (G')$, respectively, denote the maximal and minimal value of $\gamma(v)$, $v\in V(G')$, where
$\gamma(v)$ returns the index of the cluster containing $v$, and that
a path $P$ in $G$ is an \emph{$i$-cap} and \emph{$i$-cup}, respectively, if for the end vertices $u$ and $v$ of $P$ and all $w\not=u,v$ of $P$
we have $\min (P) = \gamma (u) =\gamma(v)=i\not=\gamma(w)$ and $\max (P) = \gamma (u) =\gamma(v)=i\not=\gamma(w)$.

\subsubsection{Algorithm}

\label{sec:alg}

The following lemma is a direct consequence of our characterization stated in Theorem~\ref{thm:characterization}.

\begin{lemma}
\label{lemma:cap-cupG}
Let us fix a rotation at $v$, and thus, an embedding of $G$.
Suppose that every interleaving pair of an $i$-cap $P_1$ and $j$-cup $P_2$ in $G$ containing $v$ in their interiors is feasible in the fixed embedding of $G$.
 Then $(G,T)$ is strip planar and in a corresponding clustered embedding of $(G,T)$ the rotation
at $v$ is preserved.
\end{lemma}

In what follows we show how to use Lemma~\ref{lemma:cap-cupG} for a polynomial-time strip planarity testing if the underlying abstract graph
is a subdivided star. The algorithm is based on testing in  polynomial time  whether the columns  of a  0--1 matrix can be ordered so that, in every row, either the ones or the zeros are
consecutive.
We, in fact, consider matrices containing $0,1$ and also an ambiguous symbol $*$. A matrix containing 0,1 and $*$ as its elements has the \emph{circular-ones property}
if there exists a permutation of its columns such that in every row, either the ones or the zeros are
consecutive among undeleted symbols after we delete all $*$.  Then each row in the matrix corresponds to a constraint imposed on the rotation at $v$ by Lemma~\ref{lemma:cap-cupG}
simultaneously for many pairs of paths.

By Lemma~\ref{lemma:cap-cupG} it is enough to decide if there exists a rotation at $v$
so that every interleaving pair of an $s$-cap $P_1$ and $b$-cup $P_2$ meeting at $v$ is feasible. Note that if either $P_1$ or $P_2$ does not contain
$v$ in its interior the corresponding pair is feasible.
%
An interleaving pair $P_1$ and $P_2$ passing through $v$ restricts the set of all rotations at $v$ in a strip clustered embedding of $G$.
Namely, if $e_i$ and $f_i$ are edges incident to $P_i$ at $v$ then in a strip clustered  embedding of $(G,T)$ in the rotation
at $v$ the edges $e_1,f_1$ do not alternate with the edges $e_2,f_2$, i.e.,
 $e_1$ and $f_1$ are consecutive when we  restrict the rotation to
  $e_1,f_1,e_2,f_2$. We denote such a restriction by $\{e_1f_1\}\{e_2f_2\}$.
Given a cyclic order $\mathcal{O}$ of edges incident to $v$, we can interpret $\{e_1f_1\}\{e_2f_2\}$ as a Boolean predicate
which is ``true'' if and only if $e_1,f_1$ follow the edges $e_2,f_2$ in $\mathcal{O}$.
Of course, for a given cyclic order we have $\{ab\}\{cd\}$ if and only if $\{cd\}\{ab\}$,
and $\{ab\}\{cd\}$ if and only if $\{ba\}\{cd\}$.
Then our task is to decide in  polynomial time if the rotation at $v$ can be chosen so that
 the predicates $\{e_1f_1\}\{e_2f_2\}$ of all the interleaving pairs $P_1$ and $P_2$ are ``true''.
 The problem of finding a cyclic ordering satisfying a given set
  of Boolean predicates of the form $\{e_1f_1\}\{e_2f_2\}$  is $\cNP$-complete in general, since the problem of total ordering~\cite{O79} can be easily  reduced to it in  polynomial time.
  However, in our case the instances satisfy the following structural properties making the problem tractable
  (as we see later).


\begin{observation}
\label{obs:alegebera}
If $\{ab\}\{cd\}$ is false and $\{ab\}\{de\}$ (is true) then $\{ab\}\{ce\}$ is false.
\end{observation}

The restriction on  rotations at $v$ by the pair of an $s$-cap $P_1$ and $b$-cup $P_2$ is \emph{witnessed} by
an ordered pair  $(s,b)$, where $s<b$. We treat such pair as an interval in $\mathbb{N}$. \\
Let $I=\{(s,b)| \ (s,b) \mathrm{ \ witness \ a \ restriction \  on \ rotations \ at} \ v \mathrm{\ by  \ an  \ interleaving \ pair \
of \ paths} \}$.

\begin{observation}
\label{obs:growth}
If an $s$-cap $P$ contains $v$ then $P$ contains an $s'$-cap $P'$ containing $v$ as a sub-path for every $s'$ such that $s<s'<\gamma(v)$.
Similarly, if a $b$-cup $P$ contains $v$ then $P$ contains a $b'$-cup $P'$ containing $v$ as a sub-path for every $b'$ such that $\gamma(v)<b'<b$.
\end{observation}

\begin{observation}
\label{obs:interleave}
Let $s<s'<b<b'$, $s,s',b,b'\in \mathbb{N}$. If the set $I$ contains both $(s,b)$ and $(s',b')$,
it also contains $(s,b')$ and $(s',b)$.
\end{observation}

We would like to reduce the question of determining if we can choose a rotation at $v$ making all the interleaving pairs feasible
to the following problem.
Let $S=\{e_1,\ldots, e_n\}$ of $n$ elements (corresponding to the edges incident to $v$).
Let $\mathcal{S}'=\{L_i',R_i'|\ i=1,\ldots \}$ of  polynomial size in $n$ such that $R_i',L_i'\subseteq S$ and  
$|L_i'|,|R_i'|\ge 2, \ L_{i+1}' \cup R_{i+1}' \subseteq L_i' \cup R_i'$.
  Can we cyclically order $S$ so that
both $R_i'$ and $L_i'$, for every  $R_i', L_i' \in \mathcal{S}'$, appear consecutively, when restricting the order to $R_i' \cup L_i'$?
Once we accomplish the reduction, we end up with the problem of testing the circular-ones property on matrices containing $0,1$ and $*$ as elements,
where each $*$ has only  $*$ symbol underneath. This problem is solvable in  polynomial time as we will see later.
We construct an instance for this problem which is a matrix $M=(m_{ij})$ as follows. The $i^\mathrm{th}$ row of $M$ corresponds to the pair $L_i'$ and $R_i'$ and each
column corresponds to an element of $S$.
For each pair $L_i',R_i'$ we have $m_{ij}=0$ if $j\in L_i'$, $m_{ij}=1$ if $j\in R_i'$, and $m_{ij}=*$ otherwise.
Note that our desired condition on $\mathcal{S}'$ implies that in $M$ each $*$ has only $*$ symbols underneath.
The equivalence of both problems is obvious.


In order to reduce our problem of deciding if a ``good'' rotation at $v$ exists, we first linearly
order intervals in $I$.
Let $(s_0,b_0)\in I$ be inclusion-wise minimal interval such that $s_0$ is the biggest and similarly let $(s_1,b_1)\in I$ be inclusion-wise minimal
 such that $b_1$ is the smallest one.
By  Observation~\ref{obs:interleave} we have $s_0=s_1$ and $b_0=b_1$.
Thus, let $(s_0,b_0)\in I$ be such that $s_0$ is the biggest and $b_0$ is the smallest one.
Inductively we relabel elements in $I$ as follows.
Let $(s_{i+1},b_{i+1})\in I$ be such that $s_{i+1}<s_i<b_i<b_{i+1}$ and subject to that condition $s_{i+1}$ is the biggest and $b_{i+1}$ is the smallest one. By  Observation~\ref{obs:interleave}
all the elements in $I$ can be ordered as follows \begin{equation}\label{eqn:order}{\bf (s_0,b_0)},
(s_{0,1},b_0),\ldots, (s_{0,i_0},b_0), (s_0,b_{0,1}),\ldots (s_0,b_{0,j_0}), {\bf (s_1,b_1)}, (s_{1,1},b_{1}) \ldots, (s_{1},b_{1,j_1}),  {\bf (s_2,b_2)}, \ldots\end{equation}
where $s_{k,i+1}<s_{k,i}<s_k$ and $b_{k,i+1}>b_{k,i}>b_k$.
For example, the ordering corresponding to the strip clustered graph in Figure~\ref{fig:goingUp} is
$(4,6),(3,6),(2,6),(4,7),(3,7),(2,7)$.
Let $E(s,b)$ and $E'(s,b)$, respectively, denote the set of all the edges incident to $v$
contained in an $s$-cap and $b$-cup, where $(s,b)\in I$. Thus, $E(s,b) \cup E'(s,b)$ contain edges incident to $v$
 contained in an interleaving pair that yields a restriction on rotations at $v$ witnessed by $(s,b)$.
Note that $E(s,b) \cap E'(s,b) = \emptyset$.
By Observation~\ref{obs:growth}, $E(s_{k,j+1},b_k)\subseteq E(s_{k,j},b_k)$ and $E'(s_k,b_{k,j+1})\subseteq E'(s_k,b_{k,j})$.
 The restrictions witnessed by $(s,b)$
correspond to the  following condition. In the rotation at $v$ the edges in $E(s,b)$ follow the edges in $E'(s,b)$.
Indeed, otherwise we have a four-tuple of edges $e_1,e_2,f_1$ and $f_2$ incident to $v$,
 such that $e_1,f_1\in P_1$ and $e_2,f_2\in P_2$, where $P_1$ and $P_2$ form an interleaving pair of an $s_i$-cap and $b_i$-cup,
violating the restriction $\{e_1f_1\}\{e_2f_2\}$ on the rotation at $v$.
However, such a four-tuple is not possible in an embedding by Theorem~\ref{thm:characterization}.

Let $L_i=E(s,b)$ and $R_i=E'(s,b)$, respectively, for $(s,b)\in I$, where $i$ is the
index of the position of $(s,b)$ in~$(\ref{eqn:order})$.
Note that $E(s_{i+1},b_{i+1})\cup E'(s_{i+1},b_{i+1}) \subseteq E(s_{i},b_{i})\cup E'(s_{i},b_{i})$.
Our intermediate goal of reducing our problem to the circular-ones property testing would be easy to accomplish if $I$ consisted only of intervals of the form $(s_i,b_i)$ defined above.
However, in $I$ there might be intervals of the form $(s_i,b)$, $b\not=b_i$, or $(s,b_i)$, $s\not=s_i$.
Hence, we cannot just put $L_i':=L_i$ and $R_i':=R_i$ for all $i$, since
we do necessarily have the condition $L_{i+1} \cup R_{i+1} \subseteq L_i \cup R_i$ satisfied for all $i$.


\paragraph{Definition of $\mathcal{S}'$.}
Let $\mathcal{S}=\{L_i,R_i|\ i=1,\ldots \}$. We obtain $\mathcal{S}'$ from $\mathcal{S}$ by deleting the least
number of elements from $L_i$'s and $R_i$'s so that
$L_{i+1}' \cup R_{i+1}' \subseteq L_i' \cup R_i'$ for every $i$.
More formally, $\mathcal{S}'$ is defined recursively as $\mathcal{S}_m'$, where
$\mathcal{S}_1'=\{L_1',R_1'| \ L_1'=L_1, \ R_1'=R_1 \}$ and $\mathcal{S}_{j}'=\mathcal{S}_{j-1}\cup \{L_j',R_j'| \ L_j'=L_j\cap (L_{j-1}'\cup R_{j-1}'),
 R_j'=R_j\cap (L_{j-1}'\cup R_{j-1}')\} $.
 Luckily, the following lemma lying at the heart of the proof of our result shows that information
contained in $\mathcal{S}'$ is all we need.

\bigskip
\begin{figure}[htp]

\centering
\includegraphics[scale=0.7]{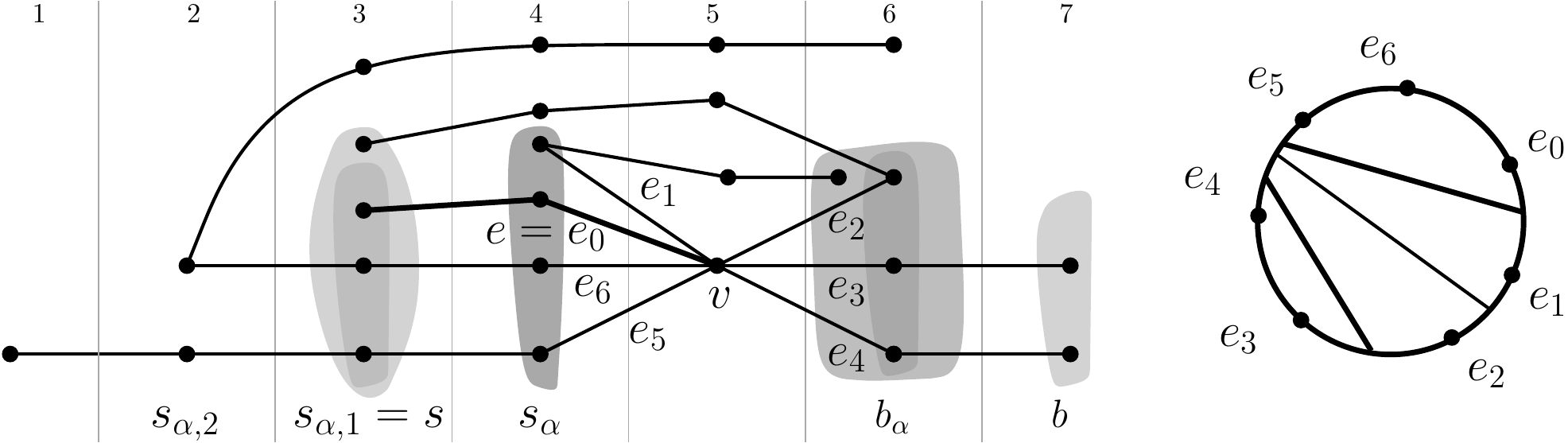}
\caption{A strip clustered subdivided star (on the left) with the center $v$, and some restrictions
  on the set of rotation at $v$ (on the right) corresponding to the intervals $(s_\alpha, b_\alpha), (s_{\alpha,1},b_\alpha)$ and $(s,b)$.
  We have $\{e_0,e_5,e_6\}=E(s_{\alpha,1},b_\alpha)\subseteq E(s,b)=\{e_0,e_2,e_5,e_6\}$ and $\{e_3,e_4\}=E'(s,b)\subseteq E'(s_{\alpha,1},b_\alpha)=\{e_3,e_4,e_1,e_2\}$.
Thus, by removing $e_0$ from $E(s,b)$ we obtain the same restrictions on the rotation at $v$.}
\label{fig:goingUp}
\end{figure}

\begin{lemma}
\label{lemma:star}
We can cyclically order the elements in $S$ so that every pair $L_i',R_i'$ in $\mathcal{S}'$ gives rise to two disjoint cyclic intervals
if and only if $(G,T)$ is strip planar.
\end{lemma}

\begin{proof}
The proof of the lemma is by a double-induction. In the ``outer--loop'' we induct over $|\mathcal{S}'|/2$ while respecting the order of pairs
 $L_{i},R_{i}$ given by~(\ref{eqn:order}). In the ``inner--loop'' we induct over the size of $S$, where in the base
 case of the $j^\mathrm{th}$ step of the ``outer--loop''
  we have $S_{j,0}=L_{j}' \cup R_{j}'$. In each ${k}^\mathrm{th}$ step of the ``inner--loop''
 we assume by induction hypothesis that a cyclic ordering $\mathcal{O}$ of $S$
 satisfies all the restrictions imposed by $\{L_{i}, R_{i}| i=1,\ldots, j-1\}$ and
 $L_j \cap S_{j,k-1}, R_j \cap S_{j,k-1}$. Clearly, once we show that $\mathcal{O}$ satisfies restrictions imposed by $L_j \cap S_{j,k}, R_j \cap S_{j,k}$,
 where $S_{j,k}=S_{j,k} \cup \{e\}$ and  $e\in (L_{j} \cup R_{j})\setminus S_{j,k-1}$ we are done.

Refer to Figure~\ref{fig:goingUp}.
Roughly speaking, by~(\ref{eqn:order}) a  ``problematic'' edge $e$ is an initial edge on a path starting at $v$ that never visits a cluster $b_\alpha$ after passing through the cluster $s_{\alpha}$ such that $e\in E(s_\alpha, b_\alpha)$ (or vice-versa with $E'(s_\alpha, b_\alpha)$).
 The edge $e$ is an \emph{$(\alpha,\beta)$-lower trim} (or $(\alpha,\beta)$-\ding{33}) if the lowest index $i$ for which $e\not\in L_i' \cup R_i'$ corresponds to $E(s_{\alpha,\beta},b_\alpha)\cup E'(s_{\alpha,\beta},b_\alpha)$, where $\beta >0$.
 Analogously, the edge $e$ is an \emph{$(\alpha,\beta)$-upper trim} (or $(\alpha,\beta)$-\ding{35}) if the lowest index $i$ for which $e\not\in L_i' \cup R_i'$ corresponds to $E(s_\alpha, b_{\alpha,\beta})\cup E'(s_\alpha, b_{\alpha,\beta})$, where $\beta >0$.
By~(\ref{eqn:order}) and symmetry (reversing the order of clusters) we can assume that $e$ is an $(\alpha, \beta)$-\ding{33}, and $e\in E(s_{\alpha,\beta-\beta'},b_\alpha)$, for some $\beta'>0$,
 where $s_{\alpha,0}=s_\alpha$, and $e\in E(s,b)=L_j$, where $s=s_{\alpha,\beta-\beta'}$ and $b>b_\alpha$, following $E(s_{\alpha,\beta},b_\alpha)$ in our order.
Moreover, we pick $e$ so that $e$ maximizes $i$ for which $e\in L_i' \cup R_i'$. We say that $e$ was ``trimmed'' at the $(i+1)^{\mathrm{th}}$ step.

\begin{figure}[htp]

\centering
\includegraphics[scale=0.7]{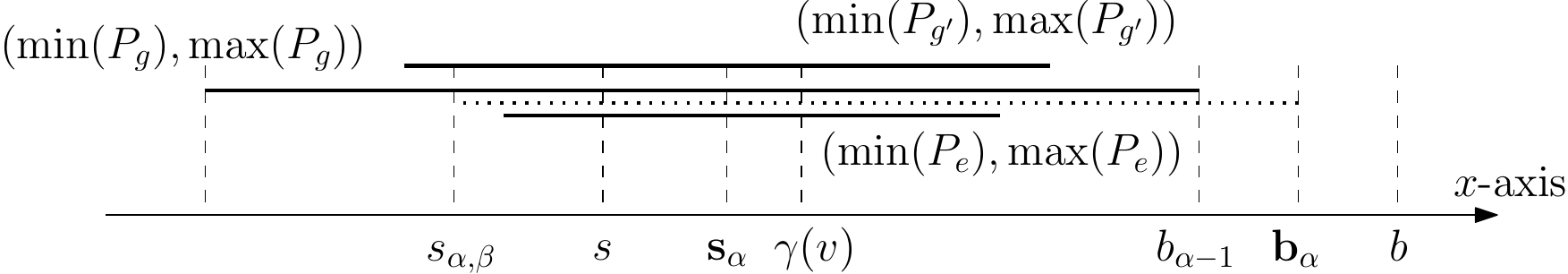}
\caption{Three intervals of clusters corresponding to three paths that start at $v$: $P_e$ that  passes through $e$ and ends in the first vertex in
 the cluster  $s_{\alpha,\beta-1}$, $P_{g'}$ that passes through $g'$ and ends in a leaf, and $P_{g}$ that ends in the first
vertex of the cluster $s''$. (An alternative interval for $P_{g}$ is dotted.) Here, $g'$ was ``trimmed'' before $e$.}
\label{fig:xAxis1}
\end{figure}

 Thus, $e$ is contained in $E(s,b)$ for some $s,b$ such that $E(s,b), E'(s,b)$
follows \\ $E(s_{\alpha,\beta},b_\alpha), E'(s_{\alpha,\beta},b_\alpha)$ in our order. However, it must be that
\begin{equation}
\label{eqn:subset1}
E(s_{\alpha,\beta-\beta'}=s,b_\alpha)\subseteq E(s,b) \  \mathrm{and} \  E'(s,b)\subseteq E'(s_{\alpha,\beta-\beta'}=s,b_\alpha),
\end{equation}
where the first relation follows directly from the fact $b>b_\alpha$ and the second relation is a direct consequence of Observation~\ref{obs:growth}.
In what follows we show that~(\ref{eqn:subset1}) implies that $\mathcal{O}$ satisfies all the required restrictions involving $e$.
We consider an arbitrary four-tuples of edges $e_1',e_2',e_3' \in S_{j,k-1}$ that together with $e$
gives rise to a restriction $\{e_1'e_2'\}\{e_3'e\}$ on $\mathcal{O}$ witnessed by $(s,b)$.
 The incriminating four-tuple must also contain
an element from $E(s,b)\setminus E(s,b_\alpha)$, let us denote it by $f=e_3'$.
Indeed, otherwise by~(\ref{eqn:subset1}) the restriction is witnessed by~$(s,b_\alpha)$ and we are done by induction hypothesis.
 Then $e_1',e_2'\in E'(s,b)$.
For the sake of contradiction we suppose that the order $\mathcal{O}$  violates the restriction $\{e_1'e_2'\}\{ef\}$.
Let $g\in L_{i'}'\subseteq E(s_{\alpha,\beta},b_{\alpha})$, for some $i'$. Note that $g$ exists (see Figure~\ref{fig:xAxis1}) for
if an edge $g' \in  E(s_{\alpha,\beta},b_{\alpha})$ is not in $L_{i'}'$ it means that $g'$ was ``trimmed'' before
$e$ and we can put $g$ to be an arbitrary element from $E(s'',b'')$ minimizing $s''$ appearing before
$E(s_{\alpha},b_{\alpha})$ in our order.

Here, the reasoning goes as follows.
Let $P_{g'}$ denote the path from $v$ passing through $g'$ and ending in a leaf.
Recall that $s_i$'s are decreasing and $b_i$'s are increasing
as $i$ increases. Thus, if we ``trimmed'' $g'$ before $e$, it had to be a \ding{33} by $s_{\alpha,\beta}<s_{\alpha}$, but then
there exists a path starting at $v$ that reaches a cluster with a smaller index than is reached by $P_{g'}$
before reaching even the cluster $b_{\alpha-1}<b_{\alpha}$.
Note that the edge $g$ can be also chosen as an edge in $E(s_{\alpha,\beta},b_{\alpha})$ minimizing $i$ such that the path starting at $v$
passing through $g$ has a vertex in the $i^\mathrm{th}$ cluster. This choice of $g$ plays a crucial role in our proof of the extension of the lemma
for trees.

 Thus, $g\in S_{j,k-1}$ by the choice of $e$, since $e\not\in L_{i'}'$.
Note that $g\in E(s,b_\alpha)$, and hence, $g\in E(s,b)$ by~(\ref{eqn:subset1}).
By Observation~\ref{obs:alegebera}
a restriction $\{e_1'e_2'\}\{fg\}$ is violated as well due to the restriction
 $\{e_1'e_2'\}\{eg\}$, that $\mathcal{O}$ satisfies by induction hypothesis,  witnessed by $(s,b_\alpha)$.  However, by~(\ref{eqn:subset1})  $\{e_1'e_2'\}\{fg\}$
is  witnessed by $(s,b)$ and we reach a contradiction with induction hypothesis.
\end{proof}

By Lemma~\ref{lemma:star} we successfully reduced our question to the problem stated above.
The  problem slightly generalizes the algorithmic question considered by  Hsu and McConnell~\cite{HC03} about testing  0--1 matrices for circular ones property.
An almost identical problem
of testing 0--1 matrices for consecutive ones property was already considered
by Booth and Lueker~\cite{Booth1976335} in the context of interval and planar graphs' recognition.
A matrix has the \emph{consecutive ones}  property if it admits
a permutation of columns resulting in a matrix in which ones  are consecutive in every row.
Our generalization is a special case of the related problem of simultaneous PQ-ordering considered recently by  Bl{\"{a}}sius and Rutter~\cite{BR14}.
In our generalization we allow some elements in the matrices to be ambiguous, i.e., they are allowed to play the roles of both zero or one.
However, we have the property that an ambiguous symbol can have only ambiguous symbols underneath in the same column.

The original algorithm in~\cite{HC03} processes the rows of the 0--1 matrix in an arbitrary order one by one.
In each step the algorithm either outputs that the matrix does not have the circular ones property and stops,
 or produces a data structure called the \emph{PC-tree} that stores \emph{all the
permutations} of its columns witnessing the circular ones property for the matrix consisting of the processed rows.
(The notion of PC-tree is a slight modification of the well-known notion of PQ-tree.)
The columns of the matrix corresponding to the elements of $S$ are in a one-to-one correspondence with the leaves of the PC-tree,
 and a PC-tree produced at every step
is obtained by a modification of the PC-tree produced in the previous step.
Let $\mathcal{Q}_i$ denote the set of permutations captured by the PC-tree after we process the first $i$ rows of the matrix.
Note that $\mathcal{Q}_{i+1}\subseteq \mathcal{Q}_i$.
By deleting some leaves from a PC-tree $T$ along with its adjacent edges we get a PC-tree $T'$
such that $T'$ captures exactly the permutations captured by $T$ restricted to their undeleted leaves.

The original algorithm in~\cite{HC03} runs in a linear time (in the number of elements of the matrix)  The straightforward cubic running time of our algorithm can be improved to a quadratic one.

\paragraph{Running time analysis.}
Let $l$ denote the degree of the center $v$ of the star $G$. 
Let $l_1<\ldots < l_{k-1}$ denote the lengths of paths ending in leaves in $G$ starting at $v$.
Thus, for each $l_i$ there exists such a path of length $l_i$ in $G$.
Let $l_i'$ denote the number of such paths starting at $v$ of length $l_i$.
The number of vertices of $G$ is $n=1+\sum_{i=0}^{k-1}l_il_i'$.
Let $l=\sum_{i=0}^{k-1}l_i$. Let $l'=\sum_{i=0}^{k-1}l_i'$.

Note that a path of length at most $l_i$ cannot ``visit'' more than $l_i$ clusters.
Thus, the number of 0's and 1's in the matrix corresponding to $(G,T)$ is upper bounded
by  $O\left(\sum_{i=0}^{k-1}\left(l'-\sum_{j=0}^{i-1}l_i'\right)l_i^2\right)$.
Indeed, each row of the matrix correponds to a pair of clusters 
and we have  $l'-\sum_{j=0}^{i-1}l_i'$ paths of length at least $l_i$.

In order to obtain a quadratic (in $n$) running time we need to 
upper bound the previous expression by $(\sum_{i=0}^{k-1}l_il_i')^2<n^2$.

We have the following 
 
 $$\sum_{i=0}^{k-1}\left(l'-\sum_{j=0}^{i-1}l_i'\right)l_i^2 \leq l\sum_{i=0}^{k-1}l_il_i' \leq \left(\sum_{i=0}^{k-1}l_il_i'\right)^2$$
 
 where the second inequality is obvious.
To show the first one we proceed as follows.

Consider the region $R$ of the plane bounded by the part of $x$-axis
between $(0,0)$ and $(l',0)$; a vertical line segment
from $(l',0)$ to $(l',l_{k-1})$; and a ``staircase polygonal line'' 
from $(l',l_{k-1})$ to $(0,0)$ with horizontal segments of lengths
$l_{k-1}',l_{k-2}',\ldots, l_0'$ and vertical segments of lenghts
$l_{k-1}-l_{k-2}, l_{k-2}-l_{k-3}, \ldots ,l_1-l_0,l_0$.
Thus, the polygonal line has vertices  \\ $(0,0), (0,l_0),(l_0',l_0), (l_0',l_1), (l_0'+l_1',l_1), \ldots, (l'-l_{k-1}', l_{k-1}), (l', l_{k-1})$.
Let $\mathcal{V}$ denote a three-dimensional set living in the Euclidean three-space obtained 
as a product of $R$ with the interval $[0,l]$ of length $l$ so that
 $\mathcal{V}$ vertically projects to $R$, and we assume that the base $R\times 0$ 
is contained in the $xy$-plane, and the rest of $\mathcal{V}$ is above this plane.
Note that the volume of $\mathcal{V}$ is exactly $ l\sum_{i=0}^{k-1}l_il_i'$.

\begin{figure}[htp]
\centering
\includegraphics[scale=0.7]{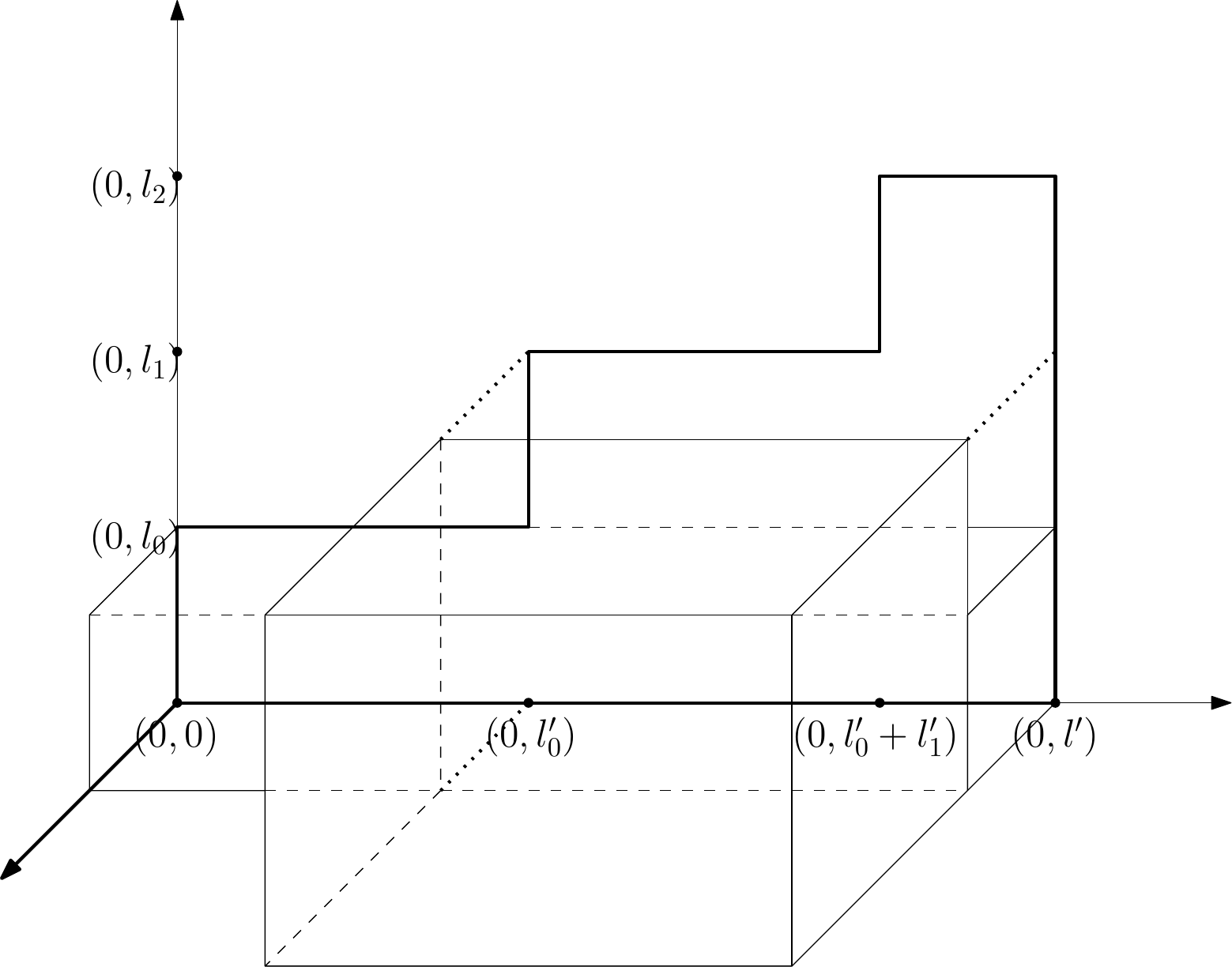}
\caption{The packing of boxes inside $\mathcal{V}$. Only two boxes out of three are shown.}
\label{fig:analysisByBoxes} 
\end{figure}

The expression $\sum_{i=0}^{k-1}\left(l'-\sum_{j=0}^{i-1}l_i'\right)l_i^2$
can be viewed as the sum of volumes of $k$ three-dimensional boxes with integer coordinates. 
Now, it is enough to pack the boxes inside $\mathcal{V}$.
We put the $i^\mathrm{th}$ box with dimensions $\left(l'-\sum_{j=0}^{i-1}l_i'\right)\times l_i \times l_i$ in an axis parallel fashion inside 
$\mathcal{V}$ such that its lexicographically smallest vertex has
coordinates $\left(\sum_{j=0}^{i-1}l_j', 0, \sum_{j=0}^{i-1}l_j\right)$.
It is a routine to check that the boxes are pairwise disjoint and contained in
$\mathcal{V}$ (see Figure~\ref{fig:analysisByBoxes} for an illustration). \\

Every permutation witnessing the circular ones property for the matrix consisting of the processed rows is  obtained as
the topological order of the leaves of the PC-tree in one of its \emph{allowed embeddings} in the plane.
The PC-tree has two types of internal leaves, namely $P$ and $C$. A node of type $P$ allows any order of its adjacent
edges in its rotation in an allowed embedding. A node of type $C$ allows only a prescribed order of its adjacent
edges in its rotation up to the choice of orientation in an allowed embedding.

\paragraph{Constraints of a PC-tree.}
A \emph{constraint} of a PC-tree $T$ is the set of leaves of
 $T$ that must appear consecutively in every allowable ordering. By splitting $T$ using a cut edge
  we obtain an \emph{edge constraint}, called \emph{edge module} in~\cite{HC03}. Analogously we define a \emph{constraint} of a 0-1 matrix $M$ as a subset of its columns
 corresponding to a constraint of its associated PC-tree.
 A 0-1 vector $\rho$ of dimension $n$ is a \emph{constraint vector} of an $n$ by $m$ matrix $M$ if the set of indices of the components of $\rho$ containing zeros (or ones) gives rise to a constraint of $M$. 

\paragraph{Matrix of constraints.}
Let $M=(m_{ij})$ denote an $n$ by $m$ matrix with ambiguous symbols 0,1 and $*$ such that each $*$ has only $*$'s underneath in the same column.
The \emph{depth} of the $j^{\mathrm{th}}$ column in $M$ is $\max \{ i| \ m_{ij}\in\{0,1\}\}$.
Let $C_i$ denote the set of columns of depth at least $i$.
Note that $C_m\subseteq C_{m-1} \subseteq \ldots \subseteq C_1$.
Let $M[i]$ denote the matrix consisting of the first $i$ rows of $M$.
Let $M_i'$ denote the matrix obtained from $M_i'$ by adding to $M[i]$ all constraint vectors of $M[1], \ldots ,M[i]$ (not already appearing in $M[i]$).
Let $M_i$ denote the restriction of $M_i'$ onto the columns of $C_i$.
In the light of Lemma~\ref{lemma:star} the following theorem yields Theorem~\ref{thm:ahahah} as a corollary if the tree is a subdivided star.
The first part of the theorem is needed for proving our variant of the Hanani--Tutte theorem.

\begin{theorem}
\label{thm:matrices}
$M$ has circular ones property if and only if all $M_i$'s have circular ones property.
There exists an algorithm to test if all $M_i$'s have circular ones property running in a quadratic time (in $|V(G)|$).
\end{theorem}

\begin{proof}
The ``only if'' direction is obvious.
We prove the ``if'' direction by induction on the number of columns in a considered matrix..
 In the base case we have that $M_m$
has circular ones property, and thus, $M$ restricted to the columns of $M_m$ has circular ones property.
Assuming that $M$ restricted to the columns of $M_i$ has the circular ones property we prove that $M$
restricted to the columns of $M_{i-1}$ has circular ones property.
Let $\mathcal{P}_i$ denote the set of all permutations of $C_i$ witnessing that $M$ restricted to the columns of $C_i$ has circular ones property.
We show that every permutation in $\mathcal{P}_i$ can be extended to a permutation of $C_{i-1}$ witnessing that
$M$ restricted to the columns of $C_{i-1}$ has circular ones property. Clearly, once we show this  induction goes through.

We proceed by analysing the PC-tree algorithm from~\cite{HC03}.
We slightly alter the algorithm so as to suit our purpose. We process the rows of $M$ in the order according to their indices.
 After each step we delete the leaves from the tree
corresponding to the columns containing ambiguous symbols in the rows processed in the subsequent steps.
Let $T_i'$ denote the PC-tree we obtain after processing the first $i$ rows of $M$ and before deleting the leaves corresponding to the columns of depth
$i$. Let $T_i$ the tree corresponding to $M_i$.
The key observations are that (i) $T_i$ is equivalent to $T_i'$.
and that (ii) every ordering of $C_i$ captured by $T_i'$ can be
extended to an ordering of $C_{i-1}$ captured by $T_{i-1}'$.
By the equivalence we mean that $T_i$ and $T_i'$ can differ only in nodes of degree two that could be suppressed or edges that could be contracted without
changing the set of permutations associated with the tree.

Regarding (ii) we just note that every ordering $o$ captured by $T_i'$ is obtained as a restriction of an ordering captured by  $T_{i-1}'$, since
there exists an allowed embedding of $T_{i-1}$, whose restriction to the leaves of $T_i'$ is $o$. Indeed,
an ordering captured by $T_i'$ is captured also by a tree obtained from $T_{i-1}'$ by deleting leaves not corresponding to the elements of $C_i$. \\

In order to obtain (i) we analyse the algorithm~\cite[Section 4]{HC03}. Also for omitted details we refer reader to~\cite{HC03}.
Let $T_{i,j}$ denote the PC-tree corresponding to the first $j$ rows of $M_i$ and rows of $M_i$
 containing constraints of $M[1], \ldots ,M[j-1]$ (restricted to the columns of $M_i$).
Note that $T_{i,i}=T_i$.
 We show by induction on the number of processed rows that $T_{i,j}$ is equivalent to $T_j'$ when restricted just to the leaves corresponding to the columns of $M_i$.

\begin{figure}[t]
\centering
\subfigure[]{
\includegraphics[scale=0.7]{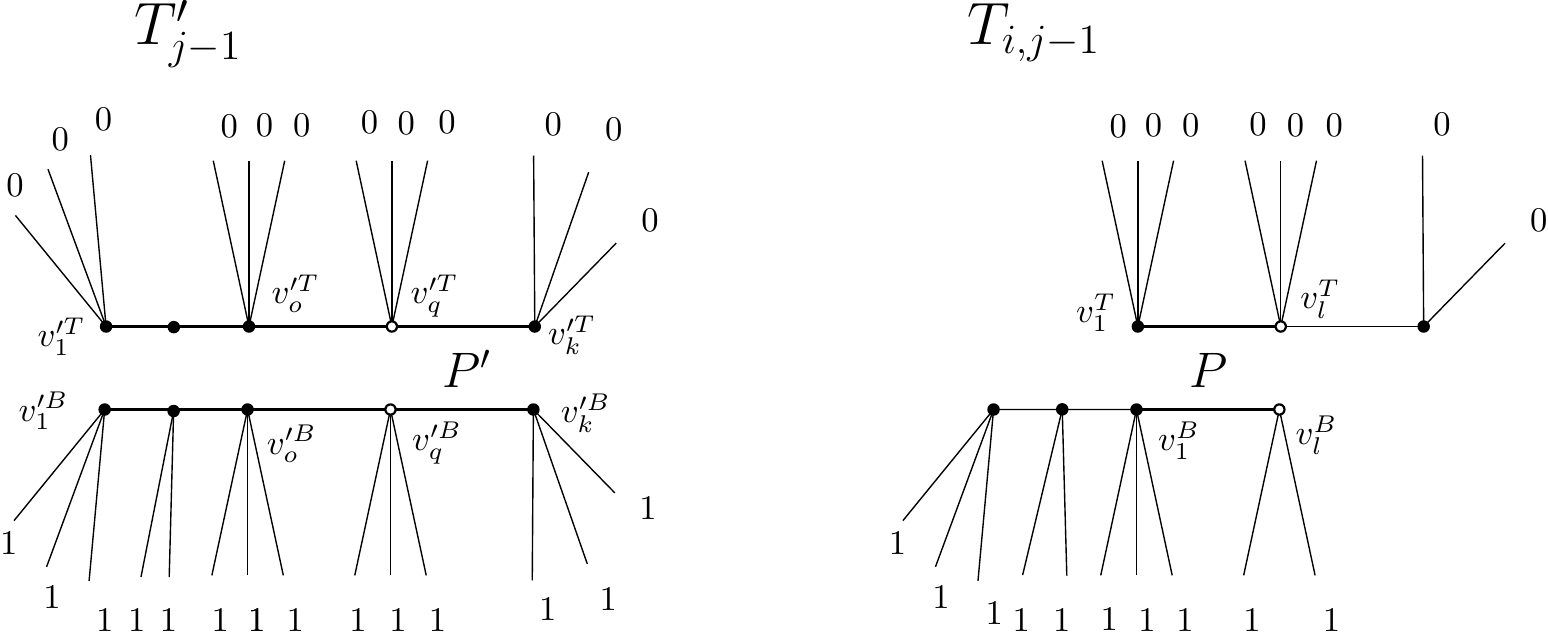}
\label{fig:pctree1}}
\subfigure[]{
\includegraphics[scale=0.7]{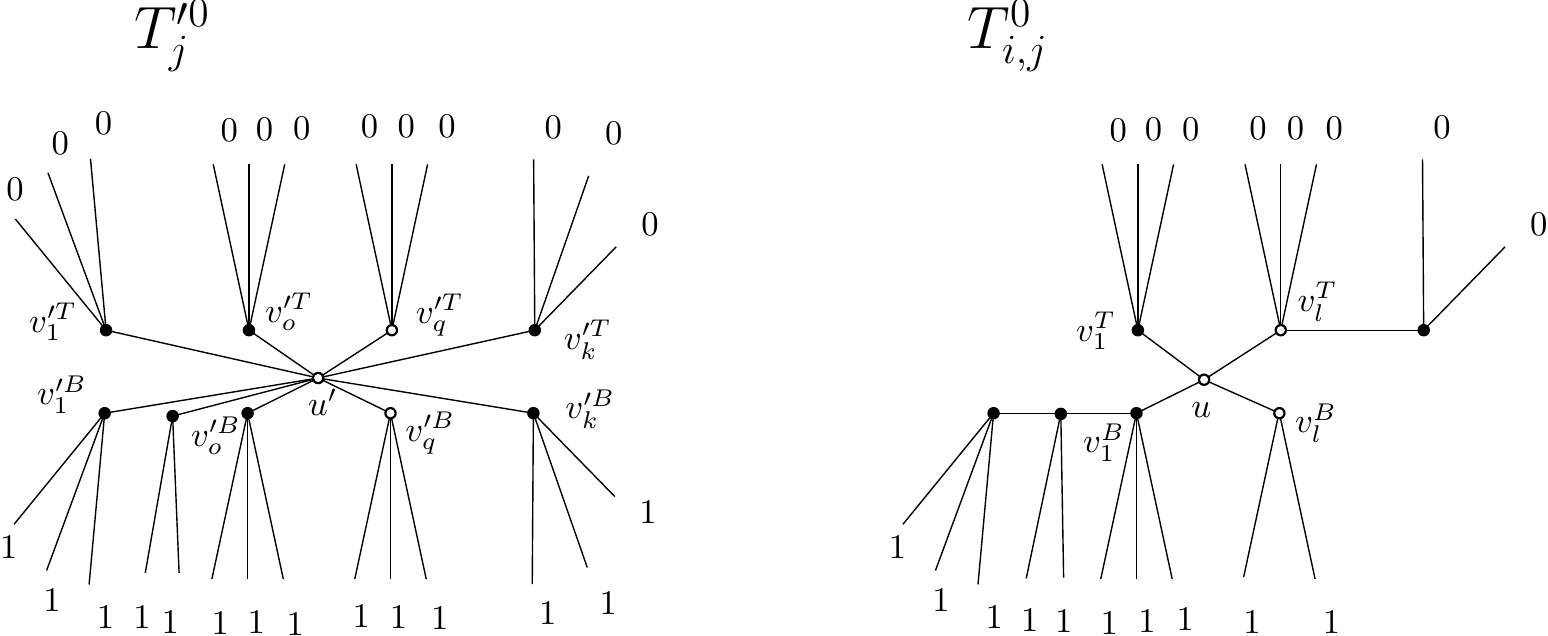}
\label{fig:pctree2}}
\caption{$C$ nodes are depicted by empty circles, $P$ nodes by full discs. The dotted line segments depict paths of defunct nodes.
(a) Splitting trees $T_{j-1}'$ and $T_{i,j-1}$, respectively, along $P'$ and $P$. By identifying bold paths $P'$ and $P$ we recover the original trees; and (b) Introduction of the $C$ node in $T_{j-1}'$ and $T_{i,j-1}$.}
\end{figure}

\begin{figure}[t]
\centering
\includegraphics[scale=0.7]{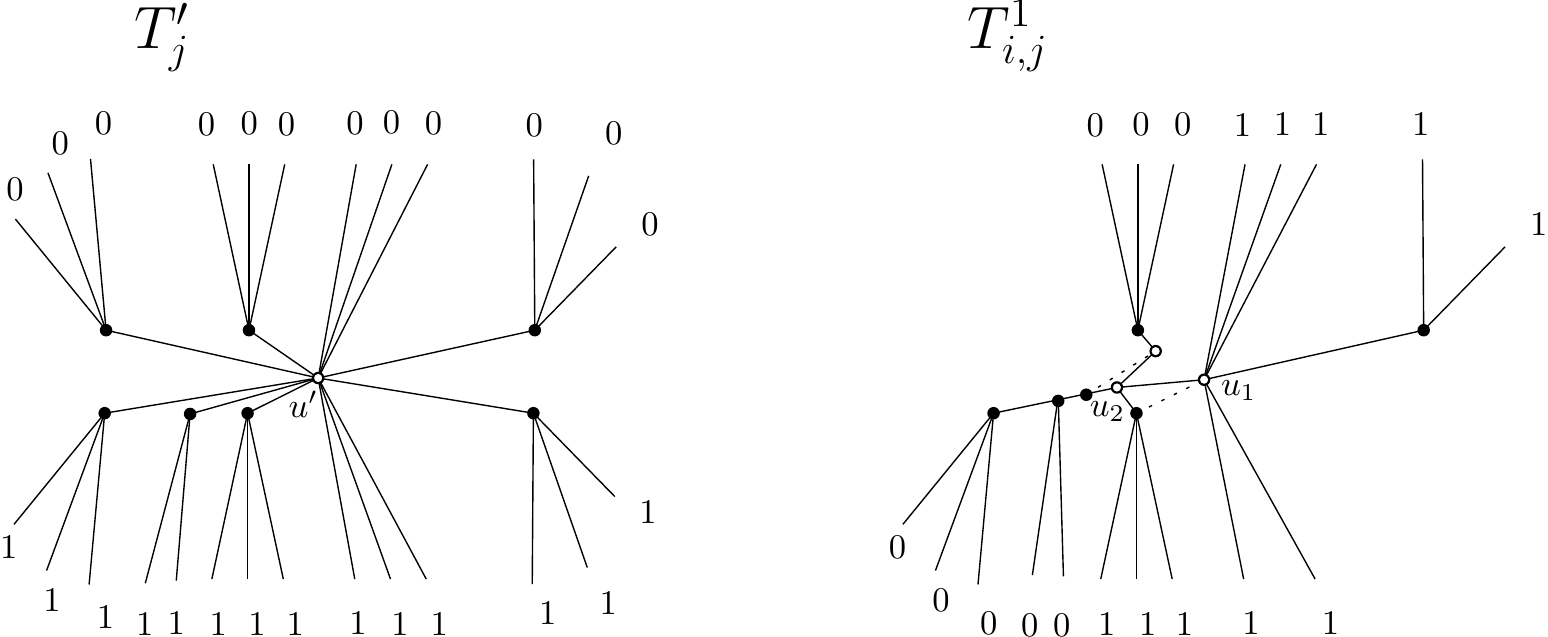}
\includegraphics[scale=0.7]{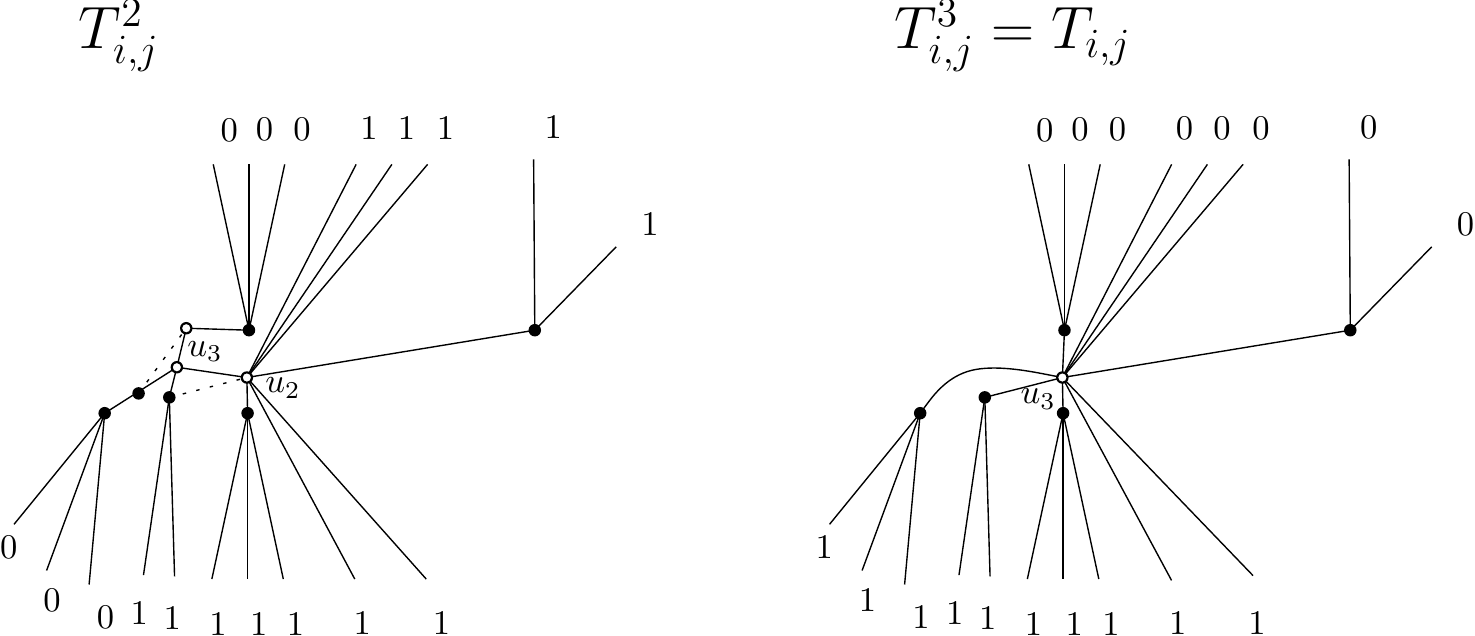}
\caption{$T_{j}'$ and $T_{i,j}$. We depict in each intermediate $T_{i,j}^{i'}$ also the defunct paths from $T_{i,j}^{i'-1}$.}
\label{fig:pctree4}
\end{figure}

 In the beginning of the algorithm the PC-tree $T_{i,0}$ and $T_0'$ are just stars whose center is a $P$ node. Thus, the claim holds.
 In the $j^\mathrm{th}$, $j\leq i$, step of the algorithm we embed $T_{i,j-1}$ and $T_{j-1}'$, respectively, such that its leaves corresponding to zeros and ones in the $j^\mathrm{th}$ row of $M$, respectively, are consecutive in the circular order given by the embedding. We also delete leaves of $T_{j-1}'$ corresponding to $*$ in the $j^\mathrm{th}$ row.
  By induction hypothesis $T_{i,j-1}$ can be obtained from $T_{j-1}'$ by deleting leaves of $T_{i,j-1}$ whose corresponding columns have $*$
  in the $i^\mathrm{th}$ row of $M$ and performing (order preserving) contractions of ``dummy'' edges in both trees.
 By following~\cite[Section 4.3]{HC03} we split $T_{i,j-1}$ and $T_{j-1}'$, respectively, along a path $P'$  and $P$ that consists of \emph{defunct} edges and nodes
 with respect to the $j^\mathrm{th}$ row of $M$. We depict the trees such that leaves representing zeros of the $j^\mathrm{th}$ row
   are above $P'$ ($P)$ and ones are below $P'$ ($P$).
 See Figure~\ref{fig:pctree1}.

 Then in the tree $T_{j-1}'$ and $T_{i,j-1}$, respectively, we introduce a $C$ node $u'$ and $u$, and connect it with the vertices of the path $P'$  and $P$
 by edges. We delete the edges of the paths $P'$ and $P$, and denote the resulting trees by $T_{j}'^0$ and $T_{i,j}^0$. See Figure~\ref{fig:pctree2}.
 The tree $T_{j}'$  is  obtained by from $T_{j}'^0$ by an (order preserving) contraction of its edges having both end vertices of type $C$.  See Figure~\ref{fig:pctree4} (left).
 The tree $T_{i,j}$ is obtained from $T_{i,j}^0$ as $T_{i,j}^{j'}$, for some $j'$,
 by an inductive process presently shown in which in each step we obtain   $T_{i,j}^{i'}$ from $T_{i,j}^{i'-1}$.
 In the $i'^{\mathrm{th}}$ step we apply a constraint of $M[j]$ to $T_{i,j}^{i'-1}$. Observe that we do have to apply all the constraints, since some of them
  do not change the tree $T_{i,j}^{i'-1}$. See Figure~\ref{fig:pctree4} (right).

 Let $P'=v_1'\ldots v_k'$ and $P=v_1 \ldots v_{l}$. By the induction hypothesis we obtain a bijection between the vertices appearing in both trees $T_{j-1}'$ and $T_{i,j-1}$.
Let $v_1$ be identified in the bijection with $v_o'$ and $v_l$ with $v_q'$
 for some $1\leq o\leq q \leq  k$. Let $v_r^T$ and $v_r^B$, for some $r$, denote the top and the bottom copy of $v_r$ obtained by splitting the path $P$.
 In what follows by \emph{children leaves} $CL(v_r^B)$  and $CL(v_r^T)$, respectively, we understand the leaves reachable by a path starting at the $C$ node $u_{i'}$, where $u_0=u$, and passing through $v_r^B$  and $v_r^T$. Similarly, we define $CL(v_r'^B)$  and $CL(v_r'^T)$.

Let the \emph{bottom part} of $T_{i,j}^{i'}$ be its subtree induced by all the vertices on the paths from $u_{i'}$ to a leaf of $CL(v_r^B)$ for some $r$.
Analogously we defined the \emph{top part} of  $T_{i,j}^{i'}$.
 Before we apply a row of $M_i$ to $T_{i,j}^{i'-1}$, we first contract all the edges having both end vertices of type $C$ in $T_{i,j}^{i'-1}$.
 The $C$ node $u_{i'}$ is the newly introduced $C$ node which, as we will see, is always identified with $u_{i'-1}$ by an (order preserving) contraction.
 Let $L(T_{i})$ denote the set of leaves of $T_{i}$.
 To obtain $T_{i,j}$ out of $T_{i,j}^0$ we successively apply constraints of the following type.
 Each time we only specify leaves corresponding to ones or zeros in the corresponding row of $M_i$.
 We use constraints of the form $C_s^B=L(T_{i}) \cap (\mathrm{CL}(v_{o-1-s}'^B) \cup \mathrm{CL}(v_{o-2-s}'^B) \cup\ldots \cup \mathrm{CL}(v_{1}'^B) \cup
 \mathrm{CL}(v_{1}'^T) \cup  \ldots \cup  \mathrm{CL}(v_{o}'^T))$,
 for  $0\le s\le o-2$ if  the bottom part of $T_{i,j}^0$ contains
 a vertex $w$ such that $w$ is identified in our bijection with $v_{r'}'\in P'$ for $r'<o$ in $T_{j-1}'$. In our figures this is the case. Analogously, we
 we use constraints $C_s^T=L(T_{i}) \cap (\mathrm{CL}(v_{o-1-s}'^T) \cup \mathrm{CL}(v_{o-2-s}'^T) \cup\ldots \cup \mathrm{CL}(v_{1}'^T) \cup
 \mathrm{CL}(v_{1}'^B) \cup  \ldots \cup  \mathrm{CL}(v_{o}'^B))$, for $0\le s\le o-2$, if  the top part of $T_{i,j}^0$ contains
 a vertex $w$ such that $w$ is identified in our bijection with $v_{r'}'\in P'$ for $r'<o$ in $T_{j-1}'$. Note that it cannot happen that both the top and bottom  part of
 $T_{i,j}^0$ contains vertices that can be identified with vertices of $T_{j-1}'$.
 We apply the constraints in the order with increasing $s$. For each $s$ we apply it only if $v_{o-s}'^B$, and hence, also $v_{o-s}'^T$,
  is a $P$ node, and only if a vertex that is identified with $v_{o-1-s}'$ exists in $T_{i,j-1}$.
  The chosen sets $C_s^B$ and $C_s^T$ are constraints, since $u'$ is a $C$ node.

  Afterwards, we analogously apply the constraints  $D_s^B=L(T_{i}) \cap (\mathrm{CL}(v_{q+1+s}'^B) \cup \mathrm{CL}(v_{q+2+s}'^B) \cup\ldots \cup \mathrm{CL}(v_{k}'^B) \cup \mathrm{CL}(v_{k}'^T) \cup  \ldots \cup  \mathrm{CL}(v_{q}'^T))$, $0\le s \le k-q-1$, or
 $D_s^T=L(T_{i}) \cap (\mathrm{CL}(v_{q+1+s}'^T) \cup \mathrm{CL}(v_{q+2+s}'^T) \cup\ldots \cup \mathrm{CL}(v_{k}'^T) \cup \mathrm{CL}(v_{k}'^B) \cup  \ldots \cup  \mathrm{CL}(v_{q}'^B))$, $0\le s \le k-q-1$.
 Note that each time we apply a constraint to $T_{i,j}^{i'-1}$ we attach one more $P$ node identified with $v_{r}'^B$ or $v_{r}'^T$, for some $r>q$,
 to the ``central'' $C$ node $u_i$. We do not need to attach a neighbor of a $C$ node joined by an edge with $u_i$ in this way, since such a node is attached by a contraction in the original
 algorithm. Thus, in the end all the relevant $P$ nodes of $T_{i,j}$ are attached to $u_{j'}$, and hence, the induction goes through. \\

Returning to the main argument,
by (i) every permutation $p$ in $\mathcal{P}_i$ captured by $T_i'$ is also captured by $T_i$ corresponding to $M_i$.
Since $M_i$ has circular ones property, $\mathcal{P}_i$ is non-empty.
By (ii) the permutation $p$ can be extended to a permutation witnessing the fact that $M$ restricted to the columns in $C_{i-1}$ has 
 circular ones property, and the first part of the claim follows.

Our algorithm runs in a quadratic time (in $|V(G)|$), since the algorithm in~\cite{HC03} is linear in the number of rows and columns of the matrix.
Indeed, by (i) to check whether all $M_i$'s have circular ones property it is enough to process rows of $M$ each time deleting leaves whose
corresponding columns have $*$ in the currently processed row.
 \end{proof}

\subsubsection{Hanani--Tutte}
In what follows we use the result and technique introduced in the previous section
to derive a variant of the Hanani--Tutte theorem for strip clustered subdivided stars.
Let us fix an independently even strip clustered drawing  $\mathcal{D}$ of a strip clustered subdivided star
 $(G,T)$.
Ultimately, we want to argue that matrices $M_i$'s from Theorem~\ref{thm:matrices} associated with $(G,T)$ 
admitting an independently even strip clustered drawing $\mathcal{D}$ do not contain Tucker's obstructions for circular ones property as sub-matrices.

 We recall some notations and their properties from the above. The set $\mathcal{S}=\{L_i', R_i'| \ L_i' \cap R_i'=\emptyset, \
|L_i'|,|R_i'|\ge 2, \ L_{i+1}' \cup R_{i+1}' \subseteq L_i' \cup R_i'   \}$ is such that
$L_i'\subseteq E(s,b)$ and $R_i' \subseteq E(s,b)$, where $(s,b)$ is the $i^\mathrm{th}$ element in~$\ref{eqn:order}$.
 Let $e_1,e_2,e_3$ and $e_4$ denote four edges in $G$ incident to $v$
 such that $\{e_1, e_2\}$  and $\{e_3,e_4\}$, respectively, is a subset
 of $L_i$ and $R_i$.
Suppose that in the rotation at $v$ the initial pieces
of $e_1$ and $e_2$ follow the initial pieces of $e_3$ and $e_4$.
Let $cr(e,f)$ denote the parity of the number of crossings between $e$ and $f$ in the given drawing of $G$.
In what follows arithmetic operations (including comparisons) are carried out in $\mathbb{Z}_2$.
The proof of the non-existence of Tucker's obstructions in our matrices boils down to the following observation.

\begin{observation}
\label{obs:separation}
$cr(e_1,e_3)+cr(e_1,e_4)+cr(e_2,e_3)+cr(e_2,e_4)=0$
\end{observation}

\begin{proof}
Let $P_1$ and $P_2$ denote the two paths meeting at $v$ witnessing
the fact that $\{e_1, e_2\}$  and $\{e_3,e_4\}$, respectively, is a subset
of $L_i$ and $R_i$.
The path $P_1$ is an $s$-cap and $P_2$ is a $b$-cup.
Let $C_1'$ and $C_2'$, respectively, denote a curve joining end vertices of $P_1$ and $P_2$
inside the cluster they belong to.
Let $C_1$ and $C_2$, respectively, denote the curve $P_1 \cup C_1'$ and $P_2 \cup C_2'$. Since $C_1$ crosses $C_2$ an even number of times,
and the initial pieces of $e_1$ and $e_2$ follow the initial pieces of $e_3$ and $e_4$
 in the rotation at $v$ we have $\sum_{e,f}cr(e,f)+cr(C_1',C_2')=0$, where we sum over the pairs $e\in C_1$ and $f\in C_2$.
Since $cr(C_1',C_2')=0$, we have $cr(e_1,e_3)+cr(e_1,e_4)+cr(e_2,e_3)+cr(e_2,e_4)+\sum_{e',f'}cr(e',f')=0$,
where we sum over the pairs $e'\in C_1$ and $f'\in C_2$ not incident to $v$. Since the drawing of $G$ is independently
even, $\sum_{e',f'}cr(e',f')=0$, and the claim follows.
\end{proof}

Observation~\ref{obs:separation} can be generalized as follows.
 Let $e_1,\ldots, e_k,f_l,\ldots, f_l$ denote edges in $G$ incident to $v$
 such that $\{e_1,\ldots, e_k\}$ and $\{f_l,\ldots, f_l\}$, respectively, is a subset
 of $L_i'$ and $R_i'$.
Suppose that in the rotation at $v$ the initial pieces of $e_1,\ldots, e_k$ follow
the initial pieces of $f_1,\ldots, f_l$.

The operation of \emph{pulling an edge $e$ over a vertex} is a continuous deformation of $e$ during which $e$ passes exactly
once over $v$, does not pass over any other vertex and does not change its position in the rotation at $v$.
Note that by pulling $e$ over $v$ we change the parity of the number of crossings of $e$ with every edge incident
to $v$. In the case when $e$ is incident to $v$ a self-crossing of $e$ created during this operation is eliminated
by cutting $e$ at every self-crossing and reconnecting the severed pieces thereby getting rid of the crossing without
affecting the parity of the number of crossings between $e$ and every other edge.

\begin{lemma}
\label{lemma:separation}
By a finite number of operations of pulling an edge incident to $v$ over  $v$, we can transform the given drawing $\mathcal{D}$ of $G$
into a drawing in which $cr(e_i,f_j)=0$ for every $i,j$. (The resulting drawing is necessarily independently odd.)
\end{lemma}

\begin{proof}
By Observation~\ref{obs:separation} for every pair of edges  $e_{i_1}$ and $e_{i_2}$ either $cr(e_{i_1},f_j)= cr(e_{i_2},f_j)$ for all $j$,
or $cr(e_{i_1},f_j)\not= cr(e_{i_2},f_j)$ for all $j$. Indeed, for $j_1$ and $j_2$ violating the claim we have
$$cr(e_{i_1},f_{j_1})+cr(e_{i_2},f_{j_1}) \not= cr(e_{i_1},f_{j_2})+ cr(e_{i_2},f_{j_2}).$$ Hence,
$$cr(e_{i_1},f_{j_1})+cr(e_{i_2},f_{j_1}) +cr(e_{i_1},f_{j_2})+ cr(e_{i_2},f_{j_2}) \not= 0.$$

This can be also seen if we consider a bipartite graph \\ $G_v=(\{e_1,\ldots, e_k,f_1,\ldots, f_k\}, \{\{e_i,f_j\}| \ cr(e_i,f_j)=0\})$. Then we can partition the set $\{f_1,\ldots , f_l\}$ into two
parts $F_1$ and $F_2$ such that the neighborhood of each vertex $e_i$ in $G_v$ is either $F_1$ or $F_2$.
This in turn implies that we can split the set $\{e_1,\ldots, e_k\}$ into two parts such that in both parts
we have $cr(e_{i_1},f_j)= cr(e_{i_2},f_j)$ for all $j$ and every pair of $e_{i_1}$ and $e_{i_2}$,
while we have $cr(e_{i_1},f_j)\not= cr(e_{i_2},f_j)$ for all $j$ and every pair of $e_{i_1}$ and $e_{i_2}$
coming from different parts.
By pulling every edge $e_i$ in one part over $v$ we obtain  $cr(e_{i_1},f_j)= cr(e_{i_2},f_j)$ for all $j$
and every pair of $e_{i_1}$ and $e_{i_2}$.
Thus, we obtained a drawing in which for all $j$ either  $cr(e_i,f_j) = 1$ for all $i$ or  $cr(e_i,f_j) = 0$ for all $i$.
By pulling all $f_j$, for which  $cr(e_i,f_j) = 1$ for all $i$, over $v$ in the obtained drawing, we obtain a desired
drawing of $G$ and that concludes the proof.
\end{proof}

By associating  edges incident to $v$ with columns of a 0-1 matrix $M$ as explained in the previous section
we extend Lemma~\ref{lemma:separation} to constraint vectors of $M$ corresponding to constraints used in the proof of Theorem~\ref{thm:matrices}.
This extension is necessary in order to apply Lemma~\ref{thm:matrices}, since matrices in its statement contain
constraint vectors as rows.
Let $e_i$'s and $f_i$'s, respectively, correspond to zeros and ones of a constraint vector of a 0-1 matrix $M$ corresponding to a sub-graph of $G$.
Here, we choose a sub-graph of $G$ that does not introduce ambiguous symbols in $M$. Note that matrices $M_i$'s in Lemma~\ref{thm:matrices}
correspond to such sub-graphs.
Suppose that in the rotation at $v$ the initial pieces of $e_1,\ldots, e_k$ follow
the initial pieces of $f_1,\ldots, f_l$.

The operation of \emph{switching} two consecutive edges $e$ and $f$ in the rotation at $v$
switches the position of $e$ and $f$ in the rotation at $v$, and thus, the parity of the number of crossings
between $e$ and $f$. In an actual drawing the operation corresponds to redrawing one of $e$ and $f$ in a close neighborhood of $v$.

\begin{lemma}
\label{lemma:modular}
By a finite number of operations of pulling an edge incident to $v$ over  $v$, we can transform the given drawing $\mathcal{D}$ of $G$
into a drawing in which $cr(e_i,f_j)=0$ for every $i,j$. (The resulting drawing is necessarily independently odd.)
\end{lemma}

\begin{proof}
First, we deal with the case of constraint that are induced by a cut edge in a PC-tree.
We proceed by induction on the number of steps (modifications of the PC-tree) of an execution of the algorithm in~\cite{HC03}.
In the base case, when the processed matrix has only one row, is established by Lemma~\ref{lemma:separation}.
In this case we have only one constraint vector, which is the row itself (and its inverse).
By induction hypothesis we assume that before the $i^\mathrm{th}$ step the lemma holds for all the edge constraints
of the PC-tree obtained after $(i-1)^\mathrm{th}$ step.

Refer to Figure~\ref{fig:edgemodules}.
It follows from the algorithm that if $e_1,\ldots, e_k$ and $f_1,\ldots, f_l$ represent a (new) edge  constraint vector of
the matrix corresponding to the PC-tree obtained after $i^{\mathrm{th}}$ step that the lemma holds by induction hypothesis and for $e_1,\ldots, e_k$ and $f_1,\ldots, f_{j}$, $j<l$, by Lemma~\ref{lemma:separation} we have the same property for $e_1,\ldots, e_k$ and $f_{j'},\ldots, f_{k'}$, $j'\leq j$,
and again by induction hypothesis the lemma holds for $e_1,\ldots, e_k$ and $f_{k''},\ldots, f_l$, $k''\leq k'$ (possibly already $k'=l$ and $k''=l$,
and hence, the third claim does not apply).
Thus, we apply the lemma to $e_1,\ldots, e_k$ and $f_1,\ldots, f_{j}$.
Now, since $j'\leq j$ we obtain $cr(e_i,f_{i'})=0$,  for $j'<i'\leq k'$, without pulling any of the $e_i$'s over $v$.
Thus, we will not destroy the property for $i'\leq j'$.
Similarly, we proceed with $f_{i'}$ for $i'>k'$.

Second, we treat the non-edge constraints. Refer to Figure~\ref{fig:nonedge}.
The claim follows by the fact that if the lemma holds for all four sets and their complements depicted by ellipses in the figure,  it also holds for the gray set. Let the ``top and bottom'' sets in the figure correspond to a row constraint of $M$.
By Lemma~\ref{lemma:separation} we can assume that the edges in the top set cross the edges in the bottom set an even number of times,
and that they are separated in the rotation at $v$. Since the ``left and right'' sets in the figure corresponds to edge constraints,
the same can be assumed about them by the first part of the proof. 
 Now, we need to argue that we can simultaneously achieve even intersection number between edges in the top and bottom set,
and between edges in the left (or right) set and its complement. Clearly, once we establish this we are done.
A simple case analysis reveals that if the  edges in the top set cross the edges in the bottom set an even number of times
(Figure~\ref{fig:nonedge} right) either all the pairs of edges from the top or the bottom set consisting of an edge in the left (or right) set and an edge in its complement  cross each other an odd number of times, or all such pairs cross an even number of times.
By switching all the pairs of edges in the top set and all the pairs of edges in the bottom set we turn all such pairs into pairs of edges crossing each other an even number of times while keeping a desired rotation at $v$.

Observation~\ref{obs:separation}~\footnote{By Lemma~\ref{lemma:modular} the observation applies to all constraint vectors.} also gives us the following lemma.

\begin{lemma}
\label{lemma:antiSeparation}
Suppose that in $\mathcal{D}$ in the rotation at $v$ the six edges $e_{i_1}, e_{i_2}, e_{i_3}, e_{i_4}, e_{i_5}$ and $e_{i_6}$
appear in the given order, and that $cr(e_{i_{j}},e_{i_{j'}})=0$ for every pair except for
$e_{i_1}$ and $e_{i_2}$, $e_{i_3}$ and  $e_{i_4}$, and  $e_{i_5}$ and $e_{i_6}$.

It cannot happen that $\{e_{i_{2j}}| \ j=1,2,3\} $ and $\{ e_{i_{2j-1}}|\ j=1,2,3\}$, respectively, is
a subset of $L_i$ and $R_i$.
\end{lemma}
\begin{proof}
Changing the rotation at $v$ by a sequence of switchings so that $e_{i_1}, e_{i_3}, e_{i_5}, e_{i_2}, e_{i_4}$ and $e_{i_6}$
appear in this order yields
\begin{equation}
\label{eqn:ha}
 cr(e_{i_1},e_{i_4})= cr(e_{i_1},e_{i_6})= cr(e_{i_3},e_{i_6}) \not=
cr(e_{i_2},e_{i_3})=cr(e_{i_2},e_{i_5})=cr(e_{i_4},e_{i_5})
\end{equation}
If the lemma does not hold, Observation~\ref{obs:separation} applies to the four-tuple of edges $e_{i_1}, e_{i_3}, e_{i_4}$
and $e_{i_6}$, and the four-tuple of edges $e_{i_3}, e_{i_5}, e_{i_2}$ and $e_{i_4}$,
and hence, by~(\ref{eqn:ha}) we have $cr(e_{i_3},e_{i_4}) \not= cr(e_{i_3},e_{i_4})$ (contradiction).
\end{proof}

\begin{figure}[htp]
\centering
\subfigure[]{
\includegraphics[scale=0.7]{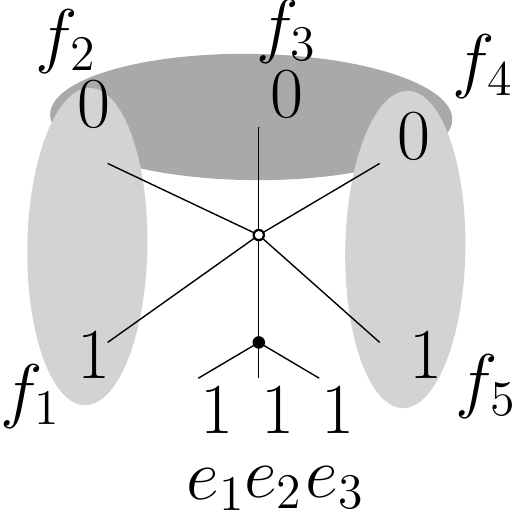}
\label{fig:edgemodules}}
\hspace{5px}
\subfigure[]{
\includegraphics[scale=0.7]{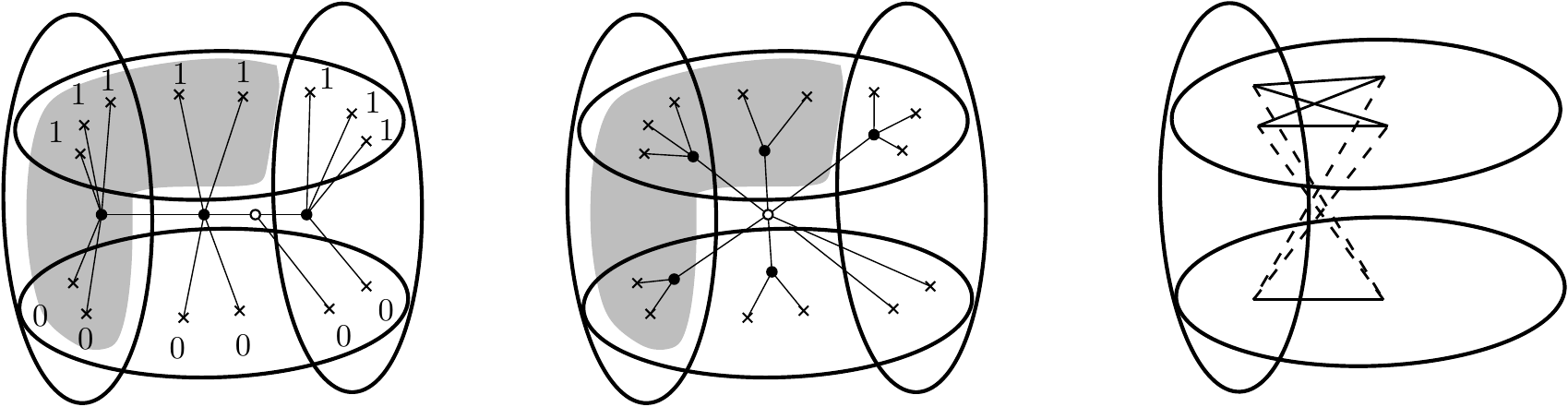}
\label{fig:nonedge}}
\caption{$C$ nodes of a PC-tree are depicted by empty circles, $P$ nodes by full discs.
(a) Sub-tree of a PC-tree corresponding to a new edge constraint $e_1,e_2,e_3$. Ones and zeros corresponds to the just processed row.
The sets $\{f_0, f_1\}$ and $\{f_4, f_5\}$ are (old) edge constraints.
(b) The four sets corresponding to two edge constraints, and a pair of complementary constraints corresponding to a processed row.
The grey set corresponds to a constraint used in the proof of Theorem~\ref{thm:matrices} (left); the corresponding PC-tree (middle),
and the (parity of) intersection number of the corresponding edges, dashed lines segment stand for even and full line segments stand for odd intersection number
(right).}
\end{figure}
\end{proof}

Let $e_1,\ldots ,e_k$ denote the edges incident to $v$.
As explained above our algorithm for testing whether a rotation at $v$ yielding an embedding exists generates a set of partitions of
$\mathcal{P}$ into two parts. The input instance $(G,T)$ is then positive, i.e., $G$ admits a strip clustered
embedding, if there exists a rotation at $v$ in which all partitions give rise to two disjoint cyclic intervals.
As we already mentioned this problem can be rephrased as a problem of testing 0--1 matrices for circular ones property.
A result of Tucker~\cite{Tucker72} says that a 0--1 matrix has the consecutive ones property if
it does not contain a matrix from a certain well-described family of  0--1 matrices as a sub-matrix.
Every matrix in the family is obtained by permuting rows or columns in one of the following  matrices.
\begin{displaymath}
\left(\begin{array}{ccccccc}
1 & 1 & 0 & 0 & \ldots & 0 & 0 \\
0 & 1 & 1 & 0 & \ldots & 0 & 0 \\
0 & 0 & 1 & 1 & \ldots & 0 & 0\\
\vdots & \vdots & \vdots & \vdots & \ddots & \vdots & \vdots\\
0 & 0 & 0 & 0 & \ldots & 1 & 1 \\
1 & 0 & 0 & 0 & \ldots & 0 & 1 \\
\end{array}\right) \hspace{10pt}
\left(\begin{array}{cccccccc}
1 & 1 & 0 & 0 & \ldots & 0 & 0 & 0 \\
0 & 1 & 1 & 0 & \ldots & 0  & 0 & 0 \\
0 & 0 & 1 & 1 & \ldots & 0 & 0 & 0\\
\vdots & \vdots & \vdots & \vdots & \ddots & \vdots & \vdots & \vdots\\
0 & 0 & 0 & 0 & \ldots & 1 & 1 & 0 \\
1 & 1 & 1 & 1 & \ldots & 1 & 0 & 1 \\
0 & 1 & 1 & 1 & \ldots & 1 & 1 & 1 \\
\end{array}\right) \hspace{10pt}
\left(\begin{array}{cccccccc}
1 & 1 & 0 & 0 & \ldots & 0 & 0 & 0 \\
0 & 1 & 1 & 0 & \ldots & 0  & 0 & 0 \\
0 & 0 & 1 & 1 & \ldots & 0 & 0 & 0\\
\vdots & \vdots & \vdots & \vdots & \ddots & \vdots & \vdots & \vdots\\
0 & 0 & 0 & 0 & \ldots & 1 & 1 & 0 \\
0 & 1 & 1 & 1 & \ldots & 1 & 0 & 1 \\
\end{array}\right)
\end{displaymath}

\begin{displaymath}
\left(\begin{array}{cccccc}
1 & 1 & 0 & 0 & 0 & 0 \\
0 & 0 & 1 & 1 & 0 & 0 \\
0 & 0 & 0 & 0 & 1 & 1 \\
0 & 1 & 0 & 1 & 0 & 1 \\
\end{array}\right)
\left(\begin{array}{ccccc}
1 & 1 & 0 & 0 & 0  \\
1 & 1 & 1 & 1 & 0  \\
0 & 0 & 1 & 1 & 0  \\
1 & 0 & 0 & 1 & 1  \\
\end{array}\right)
\end{displaymath}

\bigskip

\begin{observation}
\label{obs:consecutiveVsCircular}
A 0--1 matrix $M$ has circular ones property if and only if
the matrix $M'$ obtained from $M$ by inverting every row
having one in the first column and deleting the first column has the consecutive ones property.
\end{observation}

\begin{proof}
The ``only if'' direction.
We order the columns of $M$ so that ones or zeros are consecutive in every row.
By inverting every row with one in the first column, we turn $M$ into a matrix that has
ones consecutive in every row, and the first column contains only zeros. Thus, $M'$
has consecutive ones property.

The ``if'' direction.
We order the columns of $M'$ so that ones in every row are consecutive and add an all-zeros
column to $M'$. Clearly, the resulting matrix has  circular ones property.
This property is certainly not changed by inverting rows or permuting columns.
Thus, $M$ has also circular ones property.
\end{proof}

Observation~\ref{obs:consecutiveVsCircular} implies that similarly as matrices with consecutive ones property
we can characterize matrices with circular ones property by the family of matrices that can be obtained
by permuting rows or columns, or inverting rows in one of the following  matrices.
Indeed, by Observation~\ref{obs:consecutiveVsCircular} if a matrix $M$ does not have the circular ones property then by inverting rows so that the
first column is zero and deleting the first column we obtain a matrix $M'$ that does not have consecutive
ones property, and thus, contains one of the above forbidden matrices as a sub-matrix.
It remains to check that all such forbidden matrices with an additional all-zero column
can be transformed by permuting rows or columns, or inverting rows into a matrix containing
one of the two matrices below as a sub-matrix.

\begin{displaymath}
\mathcal{M}_1=\left(\begin{array}{cccccccc}
1 & 1 & 0 & 0 & \ldots & 0 & 0 & 0 \\
0 & 1 & 1 & 0 & \ldots & 0  & 0 & 0 \\
0 & 0 & 1 & 1 & \ldots & 0 & 0 & 0\\
\vdots & \vdots & \vdots & \vdots & \ddots & \vdots & \vdots & \vdots\\
0 & 0 & 0 & 0 & \ldots & 1 & 1 & 0 \\
1 & 0 & 0 & 0 & \ldots & 0 & 1 & 0 \\
\end{array}\right) \hspace{15pt}
M_2=\left(\begin{array}{cccccc}
1 & 1 & 0 & 0 & 0 & 0 \\
0 & 0 & 1 & 1 & 0 & 0 \\
0 & 0 & 0 & 0 & 1 & 1 \\
0 & 1 & 0 & 1 & 0 & 1 \\
\end{array}\right)
\end{displaymath}

For the sake of completeness we give the corresponding reductions.

\begin{displaymath}
\left(\begin{array}{cccccc}
1 & 1 & 0 & 0 & 0 & 0 \\
1 & 1 & 1 & 1 & 0 & 0 \\
0 & 0 & 1 & 1 & 0 & 0 \\
1 & 0 & 0 & 1 & 1 & 0 \\
\end{array}\right) \sim
\left(\begin{array}{cccccc}
1 & 1 & 0 & 0 & 0 & 0 \\
0 & 0 & 0 & 0 & 1 & 1 \\
0 & 0 & 1 & 1 & 0 & 0 \\
1 & 0 & 0 & 1 & 1 & 0 \\
\end{array}\right) \sim
\end{displaymath}

\begin{displaymath}
\left(\begin{array}{cccccc}
1 & 1 & 0 & 0 & 0 & 0 \\
0 & 0 & 0 & 0 & 1 & 1 \\
0 & 0 & 1 & 1 & 0 & 0 \\
1 & 0 & 1 & 0 & 1 & 0 \\
\end{array}\right) \sim
\left(\begin{array}{cccccc}
1 & 1 & 0 & 0 & 0 & 0 \\
0 & 0 & 1 & 1 & 0 & 0 \\
0 & 0 & 0 & 0 & 1 & 1 \\
0 & 1 & 0 & 1 & 0 & 1 \\
\end{array}\right) = M_2
\end{displaymath}

\begin{displaymath}
\left(\begin{array}{ccccccccc}
1 & 1 & 0 & 0 & \ldots & 0 & 0 & 0 & 0\\
0 & 1 & 1 & 0 & \ldots & 0  & 0 & 0 & 0 \\
0 & 0 & 1 & 1 & \ldots & 0 & 0 & 0 & 0\\
\vdots & \vdots & \vdots & \vdots & \ddots & \vdots & \vdots & \vdots & \vdots \\
0 & 0 & 0 & 0 & \ldots & 1 & 1 & 0 & 0 \\
1 & 1 & 1 & 1 & \ldots & 1 & 0 & 1 & 0 \\
0 & 1 & 1 & 1 & \ldots & 1 & 1 & 1 & 0 \\
\end{array}\right) \sim
\left(\begin{array}{ccccccccc}
1 & 1 & 0 & 0 & \ldots & 0 & 0 & 0 & 0\\
0 & 1 & 1 & 0 & \ldots & 0  & 0 & 0 & 0 \\
0 & 0 & 1 & 1 & \ldots & 0 & 0 & 0 & 0\\
\vdots & \vdots & \vdots & \vdots & \ddots & \vdots & \vdots & \vdots & \vdots \\
0 & 0 & 0 & 0 & \ldots & 1 & 1 & 0 & 0 \\
0 & 0 & 0 & 0 & \ldots & 0 & 1 & 0 & 1 \\
1 & 0 & 0 & 0 & \ldots & 0 & 0 & 0 & 1 \\
\end{array}\right) \sim
\end{displaymath}
\begin{displaymath}
\left(\begin{array}{ccccccccc}
1 & 1 & 0 & 0 & \ldots & 0 & 0 & 0 & 0\\
0 & 1 & 1 & 0 & \ldots & 0  & 0 & 0 & 0 \\
0 & 0 & 1 & 1 & \ldots & 0 & 0 & 0 & 0\\
\vdots & \vdots & \vdots & \vdots & \ddots & \vdots & \vdots & \vdots & \vdots \\
0 & 0 & 0 & 0 & \ldots & 1 & 1 & 0 & 0 \\
0 & 0 & 0 & 0 & \ldots & 0 & 1 & 1 & 0 \\
1 & 0 & 0 & 0 & \ldots & 0 & 0 & 1 & 0 \\
\end{array}\right) = \mathcal{M}_1
\end{displaymath}

\begin{displaymath}
\left(\begin{array}{ccccccccc}
1 & 1 & 0 & 0 & \ldots & 0 & 0 & 0 & 0 \\
0 & 1 & 1 & 0 & \ldots & 0  & 0 & 0 & 0 \\
0 & 0 & 1 & 1 & \ldots & 0 & 0 & 0 & 0\\
\vdots & \vdots & \vdots & \vdots & \ddots & \vdots & \vdots & \vdots & \vdots \\
0 & 0 & 0 & 0 & \ldots & 1 & 1 & 0 & 0\\
0 & 1 & 1 & 1 & \ldots & 1 & 0 & 1 & 0\\
\end{array}\right) \sim
\left(\begin{array}{ccccccccc}
1 & 1 & 0 & 0 & \ldots & 0 & 0 & 0 & 0 \\
0 & 1 & 1 & 0 & \ldots & 0  & 0 & 0 & 0 \\
0 & 0 & 1 & 1 & \ldots & 0 & 0 & 0 & 0\\
\vdots & \vdots & \vdots & \vdots & \ddots & \vdots & \vdots & \vdots & \vdots \\
0 & 0 & 0 & 0 & \ldots & 1 & 1 & 0 & 0\\
1 & 0 & 0 & 0 & \ldots & 0 & 1 & 0 & 1\\
\end{array}\right)\unlhd \mathcal{M}_1
\end{displaymath}

\bigskip
%
%

By combining this fact with Theorem~\ref{thm:matrices}, Lemma~\ref{lemma:star},\ref{lemma:separation},\ref{lemma:modular}, and~\ref{lemma:antiSeparation} we obtain a variant of the Hanani--Tutte theorem, i.e., Theorem~\ref{thm:tree},
for subdivided stars.

\begin{theorem}
\label{thm:starHanani}
If a strip clustered subdivided star admits an independently even clustered drawing
then it admits a strip clustered embedding.
\end{theorem}

\begin{proof}
By Lemma~\ref{lemma:star} and Theorem~\ref{thm:matrices}
we just need to prove that the 0--1 matrix corresponding
to our subdivided star does not contain forbidden matrices from Observation~\ref{obs:consecutiveVsCircular}.
In what follows we assume that the rotation at $v$ corresponds to the order of columns in the matrices
$\mathcal{M}_1$ and $M_2$ above.
Note that the property of not containing a given matrix as a sub-matrix is hereditary with respect
to taking (induced) sub-graphs in our context.
Thus, for the sake of contradiction we assume that a clustered subdivided star
 corresponding to $M_2$ admits an independently even clustered drawing.
 Consider the first three rows of $M_2$. Let $e_1,\ldots, e_6$ denote edges corresponding
 to the columns (in this order) of $M_2$. By Lemma~\ref{lemma:separation} and~\ref{lemma:modular} we can assume that $e_1$ and $e_2$ crosses
 every other edge $e_3,e_4,e_5$ and $e_6$ an even number of times, and $e_1$ and $e_2$ are consecutive
  in the rotation.
  Now, either the hypothesis of Lemma~\ref{lemma:antiSeparation} is satisfied,
   or a simple case analysis using Observation~\ref{obs:separation} reveals that by switching the pair $e_3,e_4$ for $e_5, e_6$ (this corresponds
to performing four operations of switching defined above) in the rotation and relabeling the edges the hypothesis of Lemma~\ref{lemma:antiSeparation}
   is satisfied.
     Then the last row of $M_2$ contradicts Lemma~\ref{lemma:antiSeparation}.

In order to rule out $\mathcal{M}_1$ similarly as $M_2$ we proceed by contradiction as well.
Let $e_1,\ldots e_k$ correspond to the columns (in this order) of an instance of $\mathcal{M}_1$ such that
the corresponding subdivided star admits an independently even clustered drawing, where in the rotation at
$v$ we have $e_1,\ldots e_k$ in this order.
Similarly as in the proof of Lemma~\ref{lemma:separation}
we can apply a sequence of operations of pulling an edge incident to $v$ over
 $v$ so that $e_1$  crosses an odd and even, respectively, number
of times exactly the same edges in the set $\{e_3,\ldots, e_k\}$ as $e_2$ does.
In particular, $cr(e_1,e_k)=cr(e_2,e_k)$.

Additionally, we assume that $cr(e_1,e_{k-1})= cr(e_1,e_{k})$, since if that is not the case
we make it happen by pulling either $e_{k-1}$ or $e_{k}$ over $v$.
Inductively, by the $i^{\mathrm{th}}$, $i>1$, row of the matrix and Observation~\ref{obs:separation}
we can guarantee for every pair $e_i$ and $e_{i+1}$, $i+1<k-1$, that
$cr(e_{i+1},e_{i+1+j})$, for $0<j\leq k-i-1$, is the same as $cr(e_{i},e_{i+1+j})$ without violating the same condition
for the previously considered pairs.
Indeed, for $j=1$ we possibly pull only $e_{i+1}$ over the center of the star if that is not the case.
Then Observation~\ref{obs:separation} applied to four-tuples $e_i,e_{i+1},e_{i+2}$ and $e_{i+1+j}$, for $j>1$,
gives us $cr(e_{i},e_{i+1+j})=cr(e_{i+1},e_{i+1+j})$, since $cr(e_{i},e_{i+2})=cr(e_{i+1},e_{i+2})$.
Thus, in particular, we have $cr(e_1,e_k)=cr(e_2,e_k)=\ldots=cr(e_{k-2},e_k)$ and
$cr(e_1,e_{k-1})=cr(e_2,e_{k-1})=\ldots=cr(e_{k-2},e_{k-1})$.

Hence, by taking the last but one row of our instance of $\mathcal{M}_1$  into account we obtain
$cr(e_1,e_{k-2})=cr(e_{k-1},e_k)$ since $cr(e_1,e_{k-1})=cr(e_1,e_{k})=cr(e_{k-2},e_k)$.
Moreover, $cr(e_1,e_{k})=cr(e_1,e_{k-1})=cr(e_{k-2},e_{k-1})$.
Thus, $cr(e_1,e_{k-2})+cr(e_1,e_k)+cr(e_{k-1},e_{k-2})+cr(e_{k-1},e_k)=0$.
 Then the last row of our matrix contradicts Observation~\ref{obs:separation}, since after switching $e_{k-1}$ with $e_k$
we obtain $cr(e_1, e_{k-2})+cr(e_1, e_{k})+ cr(e_{k-1}, e_{k-2}) + cr(e_{k-1}, e_{k})=1$.
Hence, the claim follows.
\end{proof}

\subsection{Trees}

In what follows we extend the argument from the previous section  to general trees.
Thus, for the remainder of this section we assume that $(G,T)$ is a strip clustered tree.
Let $v$ denote a vertex of $G$ of degree at least three. Refer to Figure~\ref{fig:gvstar}. Let $(G_v, T)$ denote a strip clustered star
centered at $v$ obtained as follows. For each path $P$ from $v$ to a leaf in $G$ we include to $G_v$ a path $P'$ of the same length,
whose vertex at distance $i$ from $v$ is in the same cluster as the vertex at distance $i$ from $v$ on $P$.

\bigskip
\begin{figure}[htp]
\centering
\subfigure[]{\includegraphics[scale=0.7]{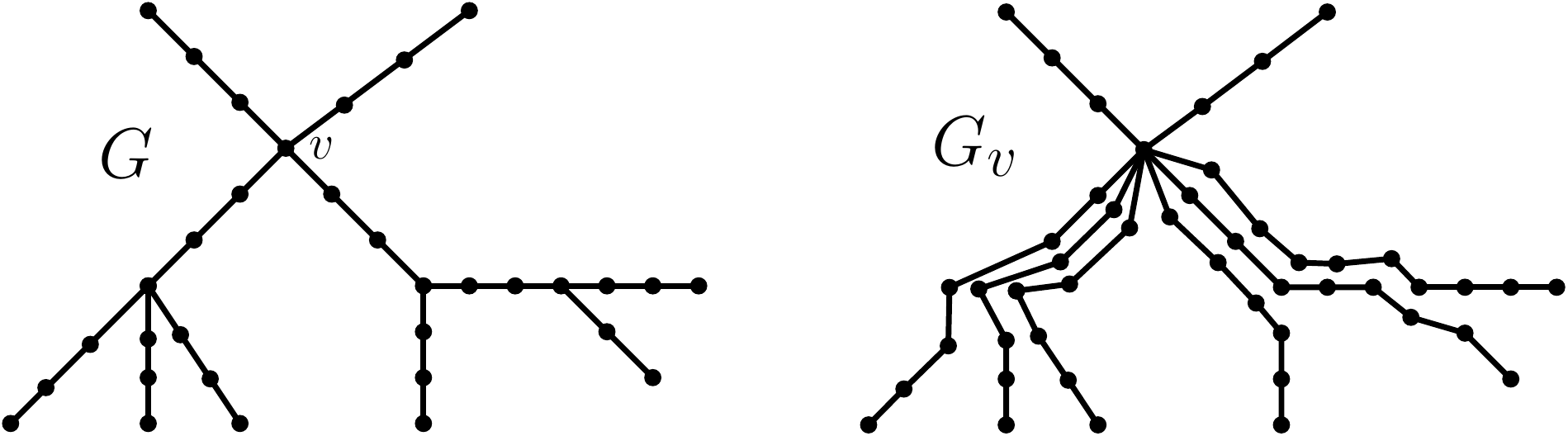}
	\label{fig:gvstar}
	} \hspace{10px}
\subfigure[]{
\includegraphics[scale=0.7]{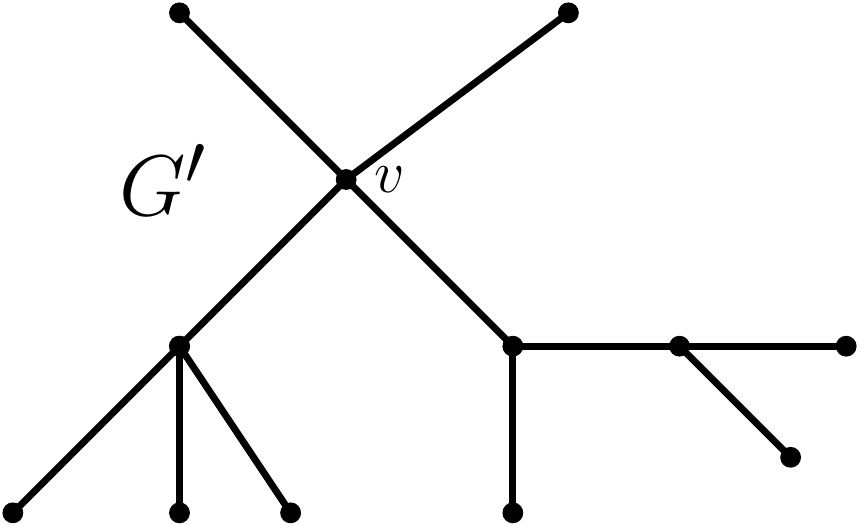}
\label{fig:suppress}}
\caption{(a) A tree $G$ (left) and its subdivided star $G_v$ centered at $v$ (right). (b) The tree $G'$ obtained after suppressing
vertices of degree two in $G$.}
\end{figure}

\subsubsection{Algorithm.}
In the light of the characterization from Section~\ref{sec:char} a na\"ive algorithm to test $(G,T)$ for strip planarity
could use the algorithm from  the previous section to check all the subdivided stars  $(G_v, T)$, $v\in V(G)$ with degree at least three,
for strip planarity. However, there are two problems with this approach.
First, we need to take the structure of the tree $G$ into account, since we pass only a limited amount of information about $(G,T)$ to the subdivided stars.
 Second, we need to somehow decide if the common intersection of the sets of possible cyclic orders of leaves of $G$
corresponding to the respective subdivided stars is empty or not.
This would be easy if we did not have ambiguous symbols in our 0--1 matrices corresponding to $(G_v,T)$.

To resolve the first problem is easy, since for each star we can simply
start the algorithm from~\cite{HC03} with the PC-tree isomorphic to $G$, whose all internal vertices are of type $P$ (see~\cite{HC03}
for a description of PC-trees).
This modification corresponds to adding rows into our 0--1 matrix, where each added rows represents the partition of the leaves
of $G$ by a cut edge, or in other words, by a bridge. Let $M_G$ denote the 0--1 matrix representing these rows.
Since we add $M_G$  at the top of the 0--1 matrix with ambiguous symbols corresponding to the given subdivided star $(G_v,T)$,
we maintain the property that an ambiguous symbol has only ambiguous symbol underneath.
Moreover, it is enough to modify the matrix only for one subdivided star.
To overcome the second problem we have a work a bit more.

First, we root the tree $G$ at an arbitrary vertex $r$ of degree at least three. Let us suppress all the non-root vertices in $G$ of degree two
and denote by $G'$ the resulting tree (see Figure~\ref{fig:suppress} for an illustration). Let us order the stars $(G_v,T)$, where
the degree of $v$ in $G$ is at least three, according to the distance of $v$ from $r$ in $G'$ in a non-increasing manner.
Thus, a subdivided star $(G_v,T)$ appears in the ordering after all the subdivided stars $(G_u,T)$ for the descendants $u$ of $v$.
For a non-root $v$ we denote by $P_v$ the path in $G$ from $v$ to its parent in $G'$.
Let $I_v=(\min(P_v),\max(P_v))$ denote the interval corresponding to $v$.
Let $I_r=(\min (G), \max (G))$.
Let $M_v$ denote a 0--1 matrix with ambiguous symbols defined by $(G_v,T)$ as in Section~\ref{sec:star}, where each row corresponds
to an interval $(s,b)$ and each column corresponds to an edge incident to $v$ in $G'$ or equivalently to a leaf of $G_v$, and hence, to a leaf of $G$.
In every  0--1 matrix $M_v$ with ambiguous symbols representing $(G_v,T)$ for $v\in G'$ with degree at least three
we delete rows that correspond to intervals $(s,b)$ strictly containing  $I_v$, i.e., $s<\min(P_v)<\max(P_v)<b$.
Let $M_v'$ denote the resulting matrix for every $v$. 

\paragraph{Running time analysis.}
We obtain a cubic running time due to the fact that there exists $O(|V(G)|)$ subdivided
stars $(G_v, T)$, $v\in V(G)$, each of which accounts for $O(|V(G)|^2)$ rows in $M$.

\paragraph{Definition of the matrix $M$.}
Refer to Figure~\ref{fig:matrica}.
Let us combine the obtained matrices $M_G$ and all $M_v'$ for $v\in V(G)$, in the given order so that the rows of $M_v'$ for some $v$ are
added at the bottom of already combined matrices. Let $M'$ denote the resulting matrix.
We replace in $M'$ the minimum number of 0--1 symbols by ambiguous symbols so that the resulting matrix has only
ambiguous symbols below every ambiguous symbol in the same column. Let $M$  denote the resulting matrix.
It remains to show the following lemma.

\begin{lemma}
\label{lemma:treelemma}
The matrix $M$ has circular ones property if and only if $(G,T)$ is strip planar.
\end{lemma}

\begin{proof}
We claim that $M$ has circular ones property if and only if $(G,T)$ is strip planar.
The ``if'' direction is easy. For the ``only if'' direction we proceed as follows.

A path $P$ in $G$ starting at $v$ is \emph{v-represented} by a column of $M$ or $M_v'$ if the column
corresponds to a leaf $u$ connected with $v$ by a path containing $P$. Note that a path $P$ can be $v$-represented
by more than one column.
A path $P$ in $G$ starting at $v$ is \emph{limited} by the interval $(s,b)$
if $s<\min(P)<\max(P)<b$. Let us assume that $P$ joins $v$ with a leaf.
Note that the column of $M_v'$ $v$-representing $P$ limited by $(s,b)$
 contains the ambiguous symbol in the row corresponding to $(s,b)$.
We consider an interleaving pair of an $s$-cap $P_1$ and $b$-cup $P_2$ that are not disjoint.
Since $G$ is a tree, $P_1$ and $P_2$ share a sub-path $P'$
(that could degenerate to a single vertex). Let $v'$ denote the vertex of $P'$ closest to the root $r$.
If the interval $(s,b)$  does not strictly contain $I_{v'}$ we let $v:=v'$.
If the interval $(s,b)$  strictly contains $I_v$ we let $v$ be the closest ancestor of $v'$ in $G'$ for which
 $(s,b)$ does not strictly contain $I_{v}$. Note that at least the root would do.
 Note that it is possible that $v$ belongs to $P'=P_1 \cap P_2$, that it does not belong to $P_1\cup P_2$, and that it belongs
 to exactly one of $P_1$ and $P_2$. However, by the definition of $v$, if $v$ does not belong to $P_1 \cup P_2$ then none of the paths $P_1$ and $P_2$ can be extended into a path containing $v$.  See Figure~\ref{fig:cases}.

  This property is crucial, and it implies that a row of $M_v'$ corresponding to $(G_v,T)$ gives rise to the restriction on the order
 of leaves corresponding to the pair $P_1$ and $P_2$. This in turn implies that in $M'$ there
 exists a row having
ones in two columns $v$-representing $P_{11}$ and $P_{12}$,
where $P_{i1}$ and $P_{i2}$ denote the paths in $G$ joining $v$ with the end vertices of $P_i$ for $i=1,2$,
and zeros in two columns $v$-representing $P_{21}$ and $P_{22}$ (or vice versa),
However, we need to show that there exists such a row in $M$, or similarly as in the case of subdivided stars, that the corresponding restriction on the order of leaves of $G$ is implied by
other rows, if such a row does not exist.

 \bigskip
\begin{figure}[htp]
\centering
\includegraphics[scale=0.7]{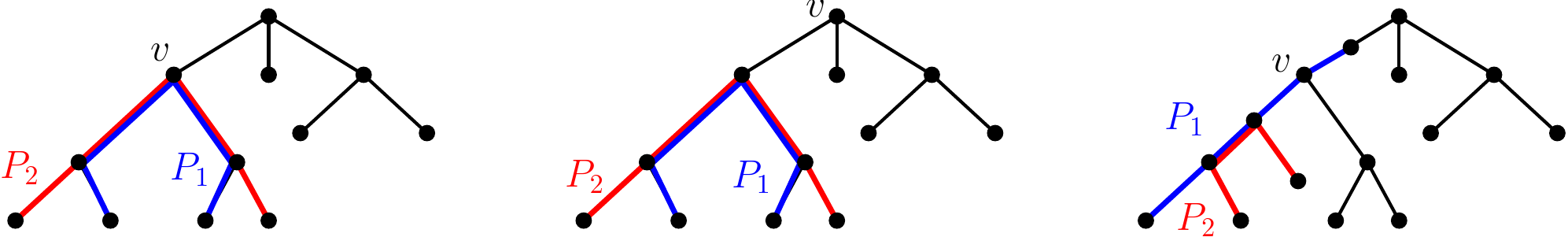}
\caption{Three distinct placements of $v$. On the left, $v$ belongs to $P_1 \cap P_2$, in the middle, $v$ does not belong to $P_1 \cup P_2$, and
on the right, $v$ belongs to exactly one of $P_1$ and $P_2$.}
\label{fig:cases}
\end{figure}

\bigskip
\begin{figure}[htp]
\centering
{
\includegraphics[scale=0.7]{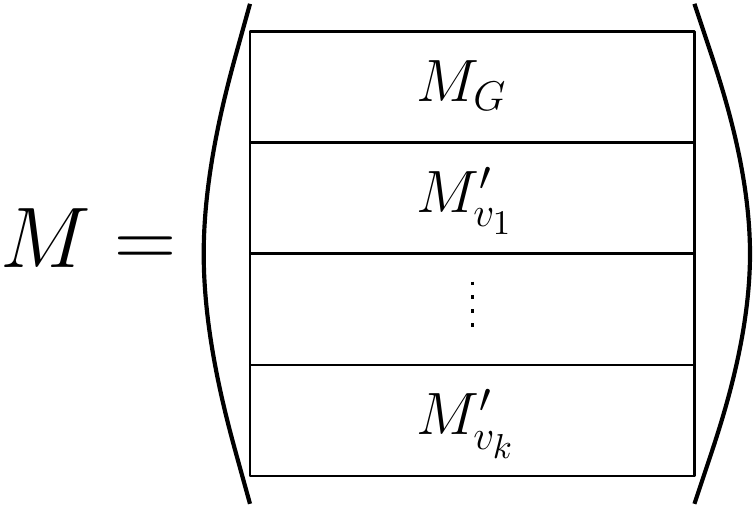}
}
\hspace{10px}
{
\includegraphics[scale=0.7]{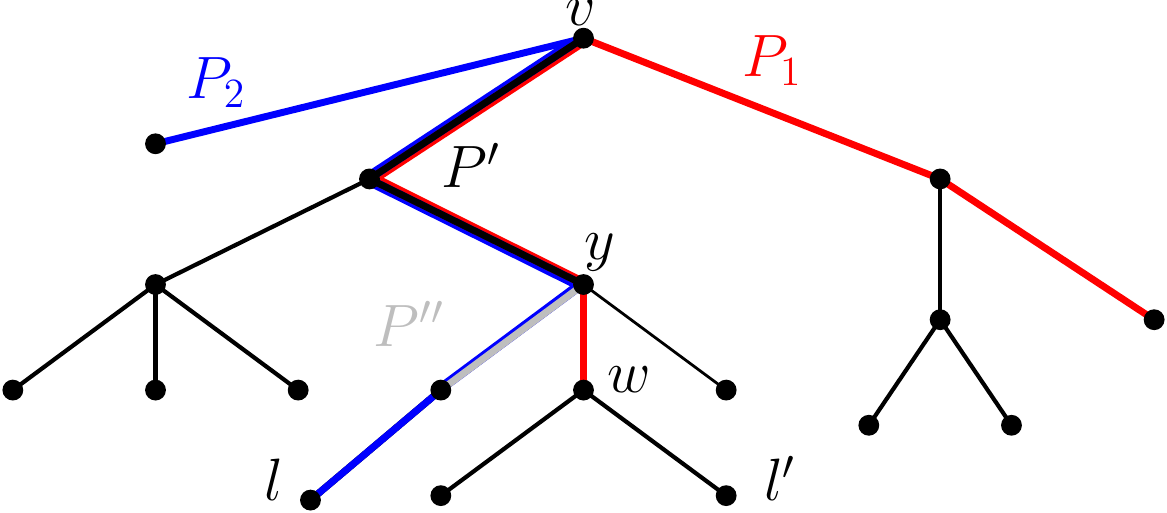}
}
\caption{(a) The composition of the matrix $M$. (b) An interleaving pair of an $s$-cap and $b$-cap $P_1$ and $P_2$ corresponding to
a restriction on a desired cyclic ordering of the leaves of $G$ involving a leaf $l'$ that was ``trimmed'' while processing $(G_y,T)$.}
\label{fig:matrica}
\end{figure}

The PC-tree algorithm processes $M$ from the top row by row.
Enforcing an ambiguous symbol below every ambiguous symbol in $M_y'$ corresponds to ``trimming'' $G$ by shortening every
path $P$ joining $y$ with a leaf $l$ in $(G_y,T)$
 limited by the interval corresponding to the currently processed row so that $P$ ends in the parent of $l$ in $G'$.
  Note that we never ``trim'' the path starting at $y$ going towards its parent in $G'$, when
processing $(G_y,T)$, since such a path can be limited only by an interval strictly containing $I_y$.
Also whenever we ``trim'' a path starting at $y$ we keep at least four paths starting at $y$, and thus, at least three paths
going from $y$ towards leaves.
Thus, if $w$, the other end vertex of $P_{\alpha\beta}$ than $v$, is not a descendant of $y$,
there exists at least one column $c$ of $M_y'$ such that $c$ $v$-represents the path $P_{\alpha\beta}$ and $c$ does not contain an ambiguous symbol in $M_y'$.
Unfortunately, if $w$ is a descendant of $y\not=v$, we could ``trim'' all the leaves that are descendants of $w$.
In this case we would like to argue similarly as in the case of a subdivided star that by introducing an ambiguous symbol in $M$  in the row corresponding
to $(s,b)$ of $M_v'$ in  all the columns $v$-representing $P_{\alpha\beta}$,
 we do not disregard required restrictions imposed on the order of leaves of $G$ by our characterization.

Let $y$ be a descendant of $v$ or equal to $v$.
Suppose that  we ``trimmed'' a descendant $l'$ of $w$ or $w$ (if $w$ is a leaf) that is a descendant of $y$, while processing $(G_y,T)$ that appears, of course, before $(G_v,T)$ in
our order or equals $(G_y,T)$. 
Let $L_z(s'',b'')$ and $L_z'(s'',b'')$, respectively, be defined analogously as $E(s'',b'')$ and $E'(s'',b'')$ in Section~\ref{sec:star}
for some $(G_z,T)$ with leaves of $G$ playing the role of the edges incident to the center of the star and some $s''$ and $b''$, $s''<b''$.
We also have an order corresponding to~(\ref{eqn:order}) for every $(G_z,T)$.
By reversing the order of clusters, without loss of generality we can assume that $w$ is an end vertex of $P_1$ which is an $s$-cap.
Refer to Figure~\ref{fig:xAxis23}(a).
If $\gamma(w)\in I_y$ then $w$ cannot be a descendant of $y$, since we have  $I_y\subseteq (\min (P_1\setminus w), \max(P_1 \setminus w))$ and $w\not\in (\min (P_1\setminus w), \max(P_1 \setminus w))$.
 Thus, $\gamma(w)\not\in I_y$, and hence, $s\not\in I_y$. Moreover, $s<\min(I_y)$, since $P_1$ is an $s$-cap.
Since the descendant $l'$ of $w$ was ``trimmed'' while processing $(G_y,T)$ by~(\ref{eqn:order}) there exists
 $L_y(s',b')$ for some $s',b'$ not containing $l'$.
Since the interval $(s',b')$ does not strictly contain $I_y$, contains $s$, and $s<\min(I_y)$ we have $s'\not\in I_y$, $s'<s$ and $b'\in I_y$.
  Moreover, if $y\not= v$ for some row of $M_y'$ we have the corresponding $L_y(s,b')$ containing $l'$
   since there exist at least two paths in $G_y$ from $y$ whose initial pieces correspond to the path from $y$ toward its parent as $v$ has degree at least three. Note that $L_y(s,b')$ can contain only descendants of $y$.
We claim the following (the proof is postponed until later)
\begin{equation}
{L_y(s,b')\subseteq L_v(s,b) \ \mathrm{and} \
L_y'(s,b')\supseteq L_v'(s,b)
\label{eqn:ahaha}}
\end{equation}
  Now, by using~(\ref{eqn:ahaha}) we can extend  the double-induction argument from Lemma~\ref{lemma:star}.

In the same manner as in the previous section $R_i$ and $L_i$ correspond to the $i^\mathrm{th}$ row of $M'$.
We define $\mathcal{S}'$ recursively as $\mathcal{S}_m'$, $m$ is the number of rows of $M'$, where
$\mathcal{S}_1'=\{L_1',R_1'| \ L_1'=L_1, \ R_1'=R_1 \}$ and $\mathcal{S}_{j}'=\mathcal{S}_{j-1}\cup \{L_j',R_j'| \ L_j'=L_j\cap (L_{j-1}'\cup R_{j-1}'),
 R_j'=R_j\cap (L_{j-1}'\cup R_{j-1}')\} $. Let $S_{j,0}=L_j' \cup R_j'$.
We need, in fact, to apply the condition~(\ref{eqn:ahaha}) only when a new leaf $l'$, such that $S_{j,k}=S_{j,k} \cup \{l'\}$,
 added to $S_{j,k-1}$ (playing the role of the edge $e$ from the proof of Lemma~\ref{lemma:star})
 is the only descendant of $w$ in ${S}_{j,k}$.
For the other descendants of $w$, $M_G$ forces the corresponding restriction.
  More formally, we need to show that a cyclic ordering of leaves $\mathcal{O}$ respecting restrictions imposed by the first $j-1$ rows, and
  the columns corresponding to $S_{j,k-1}$ in the $j^\mathrm{th}$ row, respects also restrictions imposed by
  the columns corresponding to $S_{j,k}$ in the $j^\mathrm{th}$ row. Let $\{l'l_1\}\{l_2l_3\}$ be such a restriction for a leaf $l'$ trimmed the most
  recently similarly as for $e$ in the case of  subdivided stars.
  Let the restriction  $\{l'l_1\}\{l_2l_3\}$, $l_1\in L_v(s,b)\setminus L_y(s,b')$, $l'\in L_v(s,b)$  and $l_2,l_3\in L_v'(s,b)$, induced by $S_{j,k}$ in the $j^\mathrm{th}$ row
 correspond to pair of an $s$-cap $P_1$ and $b$-cup $P_2$,
  where the leaf $l'$ $v$-represents a sub-path $P_{\alpha\beta}$ (for $v$ and $P_{\alpha\beta}$ defined as above)
  of $P_1$ ending in $w$ such that $\gamma(w)=s$.

\bigskip
\begin{figure}[htp]
\centering
{
\includegraphics[scale=0.7]{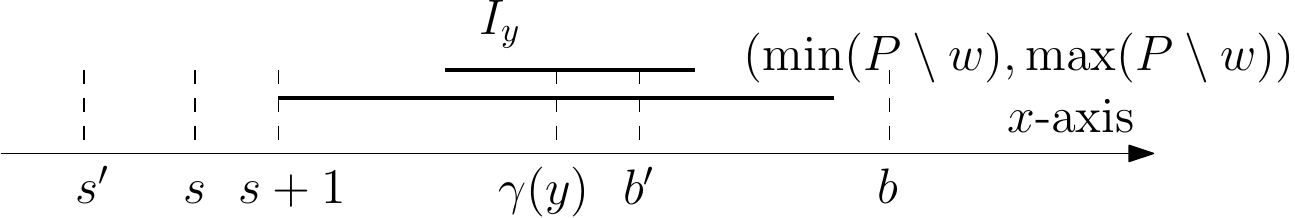}
}
\hspace{5px}
{
\includegraphics[scale=0.7]{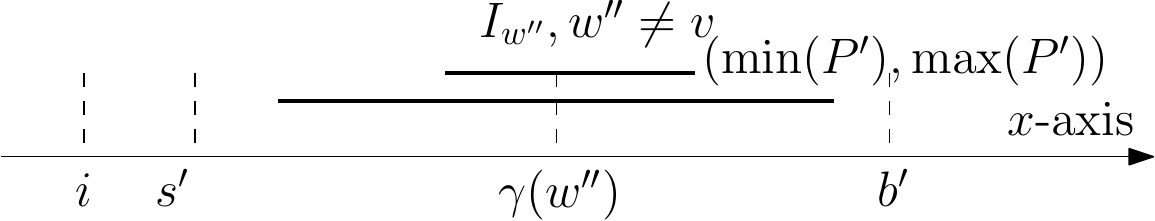}
}
\caption{(a) The interval corresponding to $P\setminus \{w\}$ containing $I_y$. (b) ) The interval corresponding to $P$ containing $I_{w'}$.}
\label{fig:xAxis23}
\end{figure}

  First, we assume that $y=v$. First, note that $l_2,l_3\in E_v'(s,b')$ by~(\ref{eqn:subset1}).
  We proceed by the same argument as in Lemma~\ref{lemma:star}, since we have~(\ref{eqn:order})
  for $G_v$. Here, $s,s',b$ and $b'$, respectively, plays the role of $s, s_{\alpha,\beta},b$ and $b_\alpha$.
After we find $l''\in  L_v(s,b')$ that was ``trimmed'' after $l'$ by induction hypothesis we have
   $\{l'l''\}\{l_2l_3\}$ and $\{l''l_1\}\{l_2l_3\}$. Thus, by Observation~\ref{obs:alegebera} we are done.
The only problem could be that we cannot find a leaf $l''\in L_v(s,b')$ (analogous to the edge $g$) in the proof of Lemma~\ref{lemma:star}
that was not ``trimmed'' before $l'$
   since all such leaves could be potentially ``trimmed'' while processing previous $(G_{y'},T)$. However, recall that we ``trim''
   only descendant of such $y'$ while processing $(G_{y'},T)$.  We show that we can pick $l''$ such that this does not happen.

Refer to Figure~\ref{fig:xAxis23}(b).
    Indeed, consider the path from $v$ to its descendant $w'$ that is an ancestor of $l''$
   such that $w'$ is an end vertex of an $s'$-cap witnessing presence of $l''$ in $L_v(s',b')$.
 Among all possible choices of $l''$ and $w'$, where, of course $l''$ is
   in the subtree rooted at $w'$, let us choose the one minimizing $i$ such that the subtree rooted at $w'$ contains a vertex in the $i^\mathrm{th}$  cluster.
   Moreover, we assume that $l''$ can be reached from $w'$ by following a path passing through the vertex in the $i^{\mathrm{th}}$ cluster.
   Let $P'$ denote the path between $v$ and $w'$.
   Let $w''$ be an ancestor of $w$  that belongs to $V(G')$ and is not an ancestor of $v$.
(For other choices of $w''$ we cannot trim $l''$ while processing $G_{w''}$.)
  If $v=w''$ we show that we cannot ``trim''  $l''$ while processing $G_{w''}$ by the argument from Lemma~\ref{lemma:star}.
Otherwise, suppose that $w''\not=v$.
   Since $I_{w''}\subseteq (\min(P'),\max(P'))$ and $b_\alpha\not\in (\min(P'),\max(P'))$, we cannot ``trim'' the descendant $l''$ of $w'$ while processing $(G_{w''},T)$. For otherwise  $l''$  is a \ding{33}
   and we would get into a contradiction with the choice of $w'$, since there exists a descendant $l'''$ of $w''$
   in a cluster with the index smaller than $i$ good for us.
    The choice of $w''$ and $l'''$ is good since the interval corresponding
     to the row of $M$ with the maximal index, in which a column of the leaf $l''$ has 0 or 1, does not strictly contain $I_{w''}$,
 and thus, the path from $v$ towards $l'''$ that ends in the cluster $s'$ never visits $b'$ cluster.

  Second, we assume that $y\not=v$. In this case we also have that $l_2,l_3\in E_y'(s,b')$ by~(\ref{eqn:ahaha}).
  We again repeat the argument from
Lemma~\ref{lemma:star} in the same manner as for the case  $y=v$.
  Here, again $s,s',b$ and $b'$, respectively, plays the role of $s, s_{\alpha,\beta},b$ and $b_\alpha$.
  Also a leaf $l''\in L_y(s,b')$ that was not ``trimmed'' before $l'$ playing the role of $g$ is found by the analysis in the previous paragraph,
where $y$ plays the role of $v$.

  If $l'$ is not a descendant of $v$, recall that
  there exists at least one ``untrimmed'' leaf $l''$ such that $l''$ $v$-represents the path $P_{\alpha\beta}$.
  Then by induction hypothesis a restriction $\{l''l_1\}\{l_2l_3\}$ witnessed by $(s,b)$ corresponding to the  $j^{\mathrm{th}}$ row gives
  the desired restriction $\{l'l_1\}\{l_2l_3\}$ on $\mathcal{O}$ by Observation~\ref{obs:alegebera}, since we have $\{l''l'\}\{l_2l_3\}$ by $M_G$.

It remains to prove~(\ref{eqn:ahaha}). Refer to Figure~\ref{fig:matrica}(b).
If $v=y$ we are done by the argument in Lemma~\ref{lemma:star}. Thus, we assume that $y\not= v$.
 We start by proving the first relation $L_y(s,b')\subseteq L_v(s,b)$.
If a leaf descendant $l$ of $y$ is in $L_y(s,b')$ then $l\in L_v(s,b)$, since $b>b'>\gamma(y)$ by $b'\in I_y,b\not\in I_y$,
and $s\not\in (\min(P'),\max(P'))$, where $P'$ is a path between $y$ and $v$.
Then the corresponding witnessing path from $y$ towards $l$ of the fact $l\in L_y(s,b)$ can be extended by $P'$ to the
path witnessing $l\in L_v(s,b')$.
In order to prove the second relation $L_y'(s,b')\supseteq L_v'(s,b)$ we first observe that
$L_y'(s,b')$ contains all the leaves that are not descendants of $y$, since $b'\in I_y$.
On the other hand, if a leaf descendant $l$ of $y$ is in $L_v'(s,b)$ then $l\in L_y'(s,b')$, since the corresponding witnessing path
from $v$ toward $l$ of the fact $l\in L_v'(s,b)$ contains a sub-path $P''$ starting at $y$ and ending in the cluster $b'$ due to $b>b'>\gamma(y)$
witnessing $l\in L_y'(s,b')$.

Thus, if the 0--1 matrix $M$ with ambiguous symbols has circular ones property then there exists an embedding of $G$ such that
every interleaving pair is feasible, and hence, by Theorem~\ref{thm:characterization} the clustered graph $(G,T)$ is strip planar.
\end{proof}

\subsubsection{Hanani--Tutte.}
Given an independently even strip clustered drawing $\mathcal{D}$ of $(G,T)$ where $G$ is a tree
we successively contract all the edges except those that are incident to the leaves of $G$ as follows.
We process the edges not incident to leaves in an arbitrary order.
By a sequence of operations of pulling an edge over a vertex we
turn the currently processed edge $e$ in $G$ into an edge that crosses every other edge an even number of times. In the resulting drawing
we contract $e$ while dragging their adjacent edges along. Note that after each step the resulting drawing is independently even.

Consider the resulting drawing of a star $G'$ with the vertex $v'$ at its center.
Hanani--Tutte variant for strip clustered graphs $(G,T)$ follows  from
Lemma~\ref{lemma:treelemma} in the same way as Theorem~\ref{thm:starHanani} from Lemma~\ref{lemma:star}
 once we extend Observation~\ref{obs:separation} to $G'$ so that the leaf end vertices of $e_1,e_2$ and $e_3,e_4$ represent
two pairs of leaves separated by a bridge in $G$, or two pairs of leaves corresponding to an interleaving pair of paths intersecting in their interiors.
For such an interleaving pair of paths $P_1$ and $P_2$ the corresponding leaves are the leaves reachable from the respective end vertices
of $P_1$ and $P_2$ by using edges not contained in their union, and hence, are not uniquely determined.

If the leaf end vertices of $e_1,e_2$ and $e_3,e_4$ represent
two pairs of leaves separated by a bridge in $G$, we clearly have that $cr(e_1,e_3)+cr(e_1,e_4)+cr(e_2,e_3)+cr(e_2,e_4)=0$,
since this property was true for the edges incident to these leaves in the original drawing, and it was not altered by contracting the non-leaf edges
or pulling an edge over a vertex.
Otherwise, if the leaf end vertices of $e_1,e_2$ and $e_3,e_4$ represent two pairs of leaves corresponding to an interleaving pair
 of paths $P_1$ and $P_2$ intersecting in their interiors,
we distinguish two cases. In the first case  $P_1\cap P_2$ is a vertex $v$. In the second case $P_1 \cap P_2$ is a non-trivial path $P$.
In the first case suppose that $e_1',e_2',e_3'$ and $e_4'$ are four edges contained in $P_1 \cup P_2$ incident to $v$.
We observe that after contracting, let's say $e_1'$,  $cr(e_1',e_3')+cr(e_1',e_4')+cr(e_2',e_3')+cr(e_2',e_4')=
cr(e_1'',e_3')+cr(e_1'',e_4')+cr(e_2',e_3')+cr(e_2',e_4')$, where $e_1''$ is an edge incident to the other end vertex $u\not=v$ of $e_1'$.
Since $cr(e_1',e_3')+cr(e_1',e_4')+cr(e_2',e_3')+cr(e_2',e_4')=0$, by Observation~\ref{obs:separation} applied to $P_1 \cup P_2$ in the original
drawing of $G$ and induction with the original drawing of $G$ as the base case,  we are done in this case.

\begin{figure}[htp]
\centering
\includegraphics[scale=0.7]{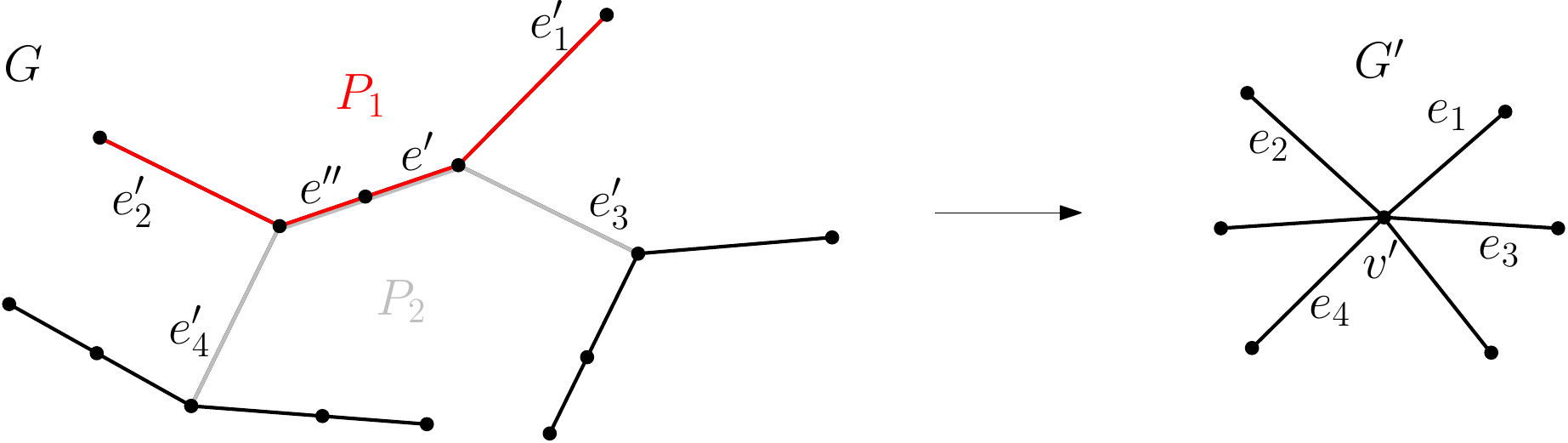}
\caption{The interleaving pair of paths $P_1$ and $P_2$ in $G$ and the edges $e_1,e_2,e_3$ and $e_4$ in $G'$ representing $e_1',e_2',e_3'$ and $e_4'$
adjacent to the path $P= P_1 \cap P_2$. Note that $e_3$ and $e_4$ are not uniquely determined by $e_3'$ and $e_4'$.}
\label{fig:last}
\end{figure}

Refer to Figure~\ref{fig:last}.
In the second case we observe that if in $G'$ the edges $e_1,e_2$ and $e_3,e_4$, respectively,
in the star $G'$ representing $P_1$ and $P_2$ appear in the rotation at $v'$
in the order $e_1,e_2,e_3,e_4$ then in the small neighborhoods of end vertices of $P=P_1 \cap P_2$ the end
pieces of the edges belonging to $E(P_1)\setminus E(P_2)$
start on the same side of $P_2$. This follows from the fact that by contracting an edge $e$ we do not alter rotations of the remaining edges and
the rotations of the two end vertices of $e$ are combined without interleaving.
Let $e_1',e_2'$ and $e_3',e_4'$, respectively, denote the corresponding edges of $P_1$ and $P_2$ so that $e_1'$ and $e_3'$ (and hence also $e_2'$ and $e_4'$)
 and incident to the same end vertex of $P$. We want to show that in the beginning before we contracted any edge the following holds
\begin{equation}
\label{eqn:last}
cr(e_1',e')+cr(e_3',e')+cr(e_1',e_3')=cr(e_2',e'')+cr(e_4',e'')+cr(e_2',e_4'),
\end{equation}
where $e'$ and $e''$, respectively, belongs to $P$ and is adjacent to $e_1'$ and $e_2'$. It can happen that $e'=e''$ if
$P$ consists of a single edge. Note that pulling an edge over a vertex does not change the validity of~(\ref{eqn:last}).
Thus, we can alter the drawing such that $cr(e_1',e')=cr(e_3',e')=cr(e_2',e'')=cr(e_4',e'')=0$. By joining two end vertices of $P_1$ by a path
completely contained in a single cluster we obtain a cycle $C$. The path $P_2$ intersects $C$ in $P$ and
is disjoint from $C$ otherwise. Note that, both, the end piece of $e_3'$ and the end piece of $e_4'$ at $C$ start either inside or outside of $C$.
Indeed, they start on the same
 side of $P_1$ and when traversing $P$ we encounter an even number of self-crossings of $C$, since the self-intersections of $P$ are counted twice
and $cr(e_1',e')=cr(e_2'',e'')=0$.  Thus, since both end vertices of $P_2$ are outside of $C$, we have
 $cr(e',e_3')+cr(e_1',e_3')=cr(e'',e_4')+cr(e_2',e_4')$, and since $cr(e',e_1')=cr(e'',e_2')$ the equation~(\ref{eqn:last}) follows.

If $P$ is formed by a single edge $e'=e''$ and $cr(e',e_i')=0$ for all $i=1,2,3,4$ by~(\ref{eqn:last})
we have $cr(e_1',e_3')=cr(e_2',e_4')$. Hence, after contracting $e'$ we have $cr(e_1',e_3')+cr(e_1',e_4')+cr(e_2',e_3')+cr(e_2',e_4')=0$ since
$cr(e_1',e_4')=cr(e_2',e_3')=0$.
 Otherwise, by contracting an edge incident to an end vertex of $P$ during our process
we preserve the parities of crossings between edges incident to the end vertices of the path that is the intersection of $P_1$ and $P_2$
after the contraction, and hence also the validity of~(\ref{eqn:last}).
This concludes the proof of Theorem~\ref{thm:tree}.

\section{Theta graphs}

\label{sec:theta}

In this section we extend result from the previous one to the class of strip clustered graphs
whose underlying abstract graph is a \emph{theta graph} defined as a union of internally vertex disjoint paths joining a pair of vertices that we call \emph{poles}. Hence, in the present section $(G,T)$ is a strip clustered graph, where $G$ is a theta graph. 

\subsection{Algorithm}

Our efficient algorithm for testing strip planarity of $(G,T)$ relies on the work of 
Bl\"asius and Rutter~\cite{BR14}. We refer the reader unfamiliar with this work
to the paper for necessary definitions. Thus, our goal is to reduce the problem to the 
problem of finding an ordering of a finite set that satisfies constraints 
given by a collection of PC-trees.\footnote{Despite the fact that~\cite{BR14} has the word
``PQ-ordering'' in the title, the authors work, in fact, with un-rooted  PQ-trees, which are our
PC-trees.}

The construction of the corresponding instance $\mathcal{I}$ of the simultaneous PC-ordering for the given $(G,T)$ 
is inspired by~\cite[Section 4.2]{BR14}.
Thus, the instance consists of a star $T_C$ having a P-node in the center  (\emph{consistency tree}),
and a collection of \emph{embedding trees} $T_0,\ldots, T_m$ constructed analogously
as in Section~\ref{sec:tree}.
The DAG (directed acyclic graph) representing $\mathcal{I}$ contains edges $(T_0,T_C,\varphi_1)$, $(T_0,T_C,\varphi_2)$, and $(T_i,T_{i+1},\phi_i)$ for $i=0,\ldots ,m-1$~\footnote{see~\cite[Section 3]{BR14} for
the definition of the DAG representing $\mathcal{I}$.} . Tree $T_0$ (see Figure~\ref{fig:thetaTree}) will consist, besides leaves and their incident edges, only of a pair of $P$-nodes joined by an edge.
It follows that the instance is solvable in a polynomial time by~\cite[Theorem 3.3 and Lemma 3.5]{BR14}\footnote{Using the terminology of~\cite{BR14} the reason is that the instance is 2-fixed.
Therein in the definition of $\mathrm{fixed}(\mu)$ it is assumed that 
every P-node $\mu$ in a  tree $T$ fixes in every parental tree $T_i$ of $T$ at most one 
P-node. This is not true in our instance due to the presence of multi-edges in the DAG.
However, multi-edges are otherwise allowed in the studied model.
 We think that the authors, in fact, meant to say that for every incoming edge of $(T_i,T)$
 the node  $\mu$ fixes in the corresponding projection of $T_i$ to the leaves of $T$ at most one P-node. Nevertheless, we can still fulfill the condition by getting rid of the multi-edge as follows. We introduce two additional copies of $T_0$, let's say $T'$ and $T''$,
 and instead of  $(T_0,T_C,\varphi_1)$ and $(T_0,T_C,\varphi_2)$ we
 put $(T',T_C,\varphi_1)$ and $(T'',T_C,\varphi_2)$. Finally, 
 we put $(T',T_0,\varphi)$ and $(T'',T_0,\varphi)$, where $\varphi$ is the identity.
 It is a routine to check that the resulting instance is 2-fixed.}
The running time follows by~\cite[Theorem 3.2]{BR14}.

\begin{figure}[htp]
\centering
\includegraphics[scale=0.7]{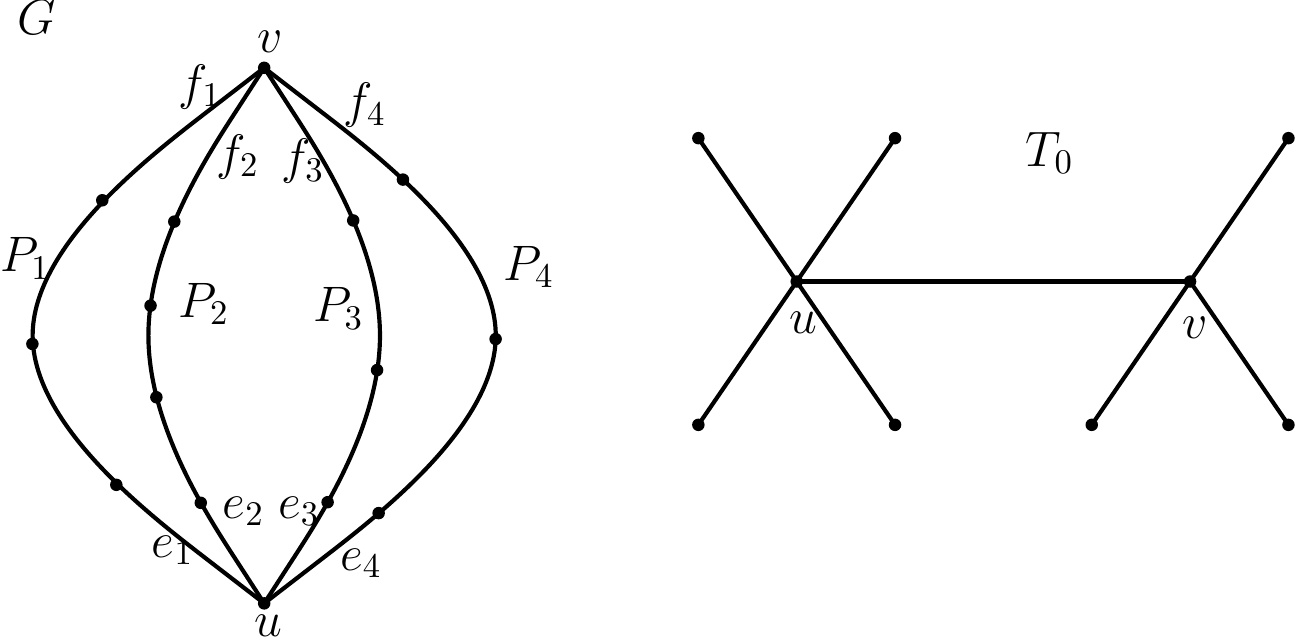}
\caption{The theta-graph $G$ (left) and the PC-tree $T_0$ in
its corresponding $\mathcal{I}$ instance.}
\label{fig:thetaTree}
\end{figure}

\paragraph{Description of $(T_0,T_C,\varphi_1)$, $(T_0,T_C,\varphi_2)$.}
Let $u$ and $v$ denote the poles of $G=(V,E)$.
Let $e_1,\ldots, e_n$ denote the edges incident to $u$.
Let $f_1,\ldots, f_n$ denote the edges incident to $v$.
We assume that $e_i$ and $f_i$ belong to the same  path $P_i$ between $u$ and $v$.
The non-leaf vertices of $T_0$ are $P$-nodes $u$ and $v$, corresponding 
to $u$ and $v$ of $G$, joined by an edge. The $P$-node $u$ is adjacent to $n$ leaves 
and $v$ is adjacent to $n-1$ leaves.
The tree $T_C$ is a PC-tree with a single $P$-node and $n$ leaves.
The map $\varphi_1$ maps injectively every leaf of $T_C$ except one to a leave adjacent to $u$
the remaining leaf is mapped to an arbitrary leave of $v$.
The map $\varphi_2$ maps injectively every leaf of $T_C$ except one to a leave adjacent to $v$
the remaining leaf is mapped arbitrarily such that the map is injective.\\

Let $I_{i}= (\min(P_i), \max(P_i))$ for every $i=1,\ldots,n$.
Let $I_j$ be chosen so that $I_i\subseteq I_j$ implies $i=j$.
Thus, $I_j$ is spanning a minimal number of clusters among $I_i$'s.
Let $I_{\alpha}:=I_j$ and $P_{\alpha}:=P_j$.
The leaves of $\varphi_1$ are mapped to the leaves of $u$ corresponding to edges 
incident to $u$ except for $e_{\alpha}$ and a single leave corresponding to the position of the 
the outer-face,
and we have the analogous compatible correspondence for $\varphi_2$ except that $v$ has 
no leave representing the outer-face that $\varphi_2$ avoids.
Since $T_C$ is the consistency PC-tree, the inherited cyclic orders of end pieces of
paths $P_i$ corresponding to $\varphi_1$ and $\varphi_2$ have opposite orientations. Thus,  one of the arcs  $(T_0,T_C,\varphi_1)$ and $(T_0,T_C,\varphi_2)$ is orientation preserving and the other one
orientation reversing.
\bigskip

\begin{figure}[htp]
\centering
\includegraphics[scale=0.7]{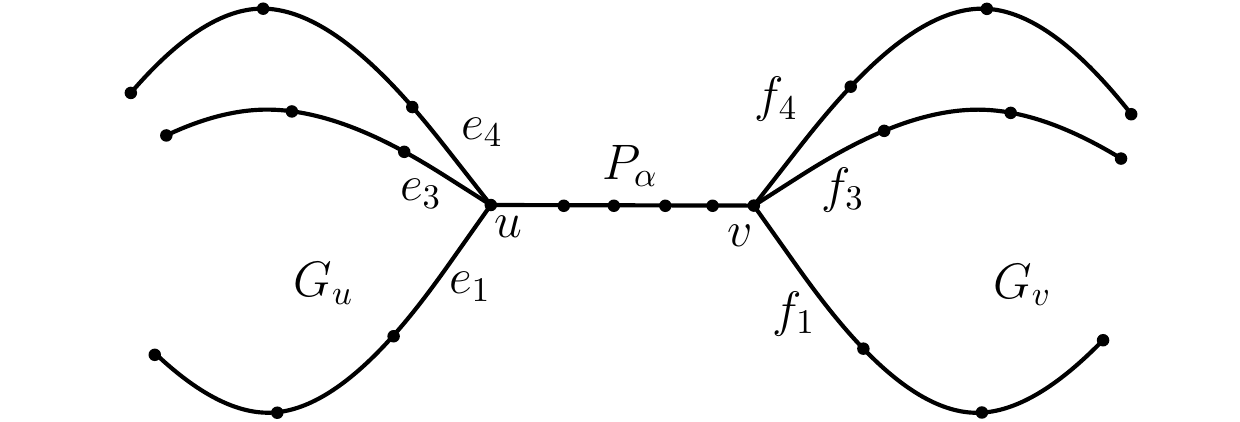}
\caption{The graph $G'$ consisting of $G_u$, $G_v$ and $P_\alpha$.}
\label{fig:thetaTree2}
\end{figure}

\paragraph{Description of $(T_i,T_{i+1},\phi_1)$.}
We will construct a strip clustered tree $(G',T')$ (see Figure~\ref{fig:thetaTree2}) yielding desired embedding trees $T_i$'s, for $i>0$, in
$\mathcal{I}$. The tree $G'$ is obtained as the union of a pair of vertex disjoint
subdivided star $G_u$ and $G_v$, and the path $P_\alpha$ (defined in the previous paragraph) joining the centers of $G_u$ and $G_v$.
The graph $G_u$ is isomorphic to  $G\setminus (\{v\}\cup E(P))$ 
and $G_v$ is isomorphic to $G\setminus (\{u\}\cup E(P))$.
 The assignment $T'$ of the vertices to clusters
is inherited from $(G,T)$.

The path $T_0T_1\ldots T_m$ in the  DAG of $\mathcal{I}$ corresponds to a variant $M'$ 
of the matrix $M$ from Section~\ref{sec:tree} constructed for $G'$ with additional rows between $M_{G'}'$ and $M_u'$.
 The leaves of $T_1$ corresponds to the columns of $M'$. \\

\noindent 
(i) The tree $T_0$ corresponds to $M_{G'}$. \\ (ii) $T_1$ takes care of the trapped vertices (that
we did not have to deal with in the case of trees).
 \\ (iii) $T_1,\ldots, T_l$, for some $l$, corresponds to $M_u'$, and  $T_{l+1},\ldots, T_m$, to $M_v'$. \\
 
The described trees naturally correspond to constraints on the rotations at $u$ and $v$.
Before we proceed with proving the correctness of the algorithm we describe the last missing
piece of the construction, the PC-tree $T_1$.

\paragraph{Description of $T_1$.}
The PC-tree $T_1$ corresponds to the set of constraints given by the following 0--1 matrix $M''$.
The leaves corresponding to $e_{\alpha}$ are all the leaves incident to $v$, and the leaves
corresponding to $f_{\alpha}$ are the leaves incident to $u$.
The correspondence of the remaining edges incident to $u$ and $v$ to the columns of $M''$ is
given by their correspondence with leaves of $T_1$ explained above.
The matrix $M''$ has only zeros in the column representing the outer-face in the rotation at $u$.
Let us take the maximal subset of edges incident to $u$, whose elements
$e_1,\ldots, e_{n'}$  are ordered (we relabel the edges appropriately) such that
$i<i'$ implies $\min (P_i)<\min(P_{i'})$. For every $j$, $1\leq j \leq n'$, we introduce
a row having zeros in the columns corresponding to $e_1,\ldots ,e_j$ and columns
corresponding to $e_i$ for which there exists $i'$, $1\leq i'\leq j$, $\min (P_i)=\min(P_{i'})$, and having ones
in the remaining columns except the one representing the outer-face.

Let us take the maximal subset of edges incident to $u$, whose elements
$e_1,\ldots, e_{n'}$ are ordered  such that
$i<i'$ implies $\max (P_i)>\max(P_{i'})$. For every $j$, $1\leq j \leq {n'}$, we introduce
a row having zeros in the columns corresponding to $e_1,\ldots ,e_j$ and columns
corresponding to $e_i$ for which there exists $i'$, $1\leq i'\leq j$, $\max (P_i)=\max(P_{i'})$, and having  ones
in the remaining columns except the one representing the outer-face.  

A trapped vertex $v'$ on $P_{i}$ in a cycle $C$ consisting of $P_{i'}$ 
and $P_{i''}$ would violate a constraint $\{e_i'e_i''\}\{e_ie_0\}$,
where $e_0$ is the dummy edge on the outer-face enforced by $T_1$. \\

It remains to prove that the instance $\mathcal{I}$ is a if and only if 
any corresponding witnessing order of the leaves yields an embedding of $(G,T)$ that
is strip planar by Theorem~\ref{thm:characterization}.
This might come as a surprise since some constraints on the rotation system enforced by the original instance $(G,T)$ 
might be missing in $\mathcal{I}$, and on the other hand some additional constraints might be introduced.

\begin{theorem}
\label{thm:theta_alg}
The instance $\mathcal{I}$ is a ``yes'' instance if and only if $(G,T)$ is strip planar.
\end{theorem}

\begin{proof}
By the discussion above it remains to prove that the order constraints corresponding to the
trees $T_2,\ldots, T_m$ are all the constraints given by the unfeasible
interleaving pairs of paths $P_1'$ and $P_2'$, and possibly additional
constraint enforced by trapped vertices.

First, we note that no additional constraints are introduced due to the fact 
that we might have two copies of a single vertex of $G$ in $G'$. Indeed, such a 
constraint corresponds to a pair of paths $P_1'$ and $P_2'$ intersecting in the copy of $P_{\alpha}$,
such that their union contains at least one whole additional $P_l$, for $l\not=\alpha$.
However, in this case it must be that the two copies of a single vertex in the union of $P_1'$ and $P_2'$ are the endpoints of, let's say $P_1'$. Thus, the constraint of $P_1'$ and $P_2'$ exactly prevents end vertices of $P_2'$ from being trapped
 in the cycle of $G$ obtained by identifying the end vertices of $P_1'$.

\begin{figure}[htp]
\centering
\includegraphics[scale=0.7]{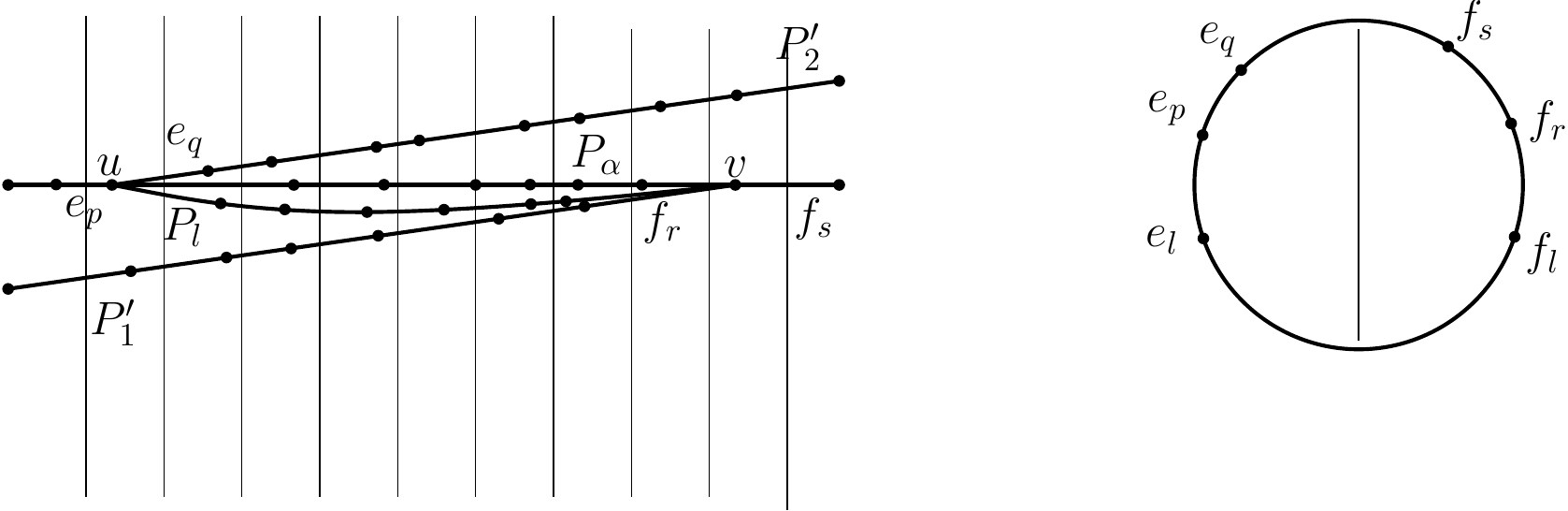}
\caption{The paths $P_1'$ and $P_2'$ meeting in $P_l$ such that $I_{l}= (\min(P_l), \max(P_l))$ 
contains $I_\alpha$ (left). The corresponding cyclic order of the leaves  corresponding to the edges incident to $u$ and $v$ captured by $T_0$ (right).}
\label{fig:thetaTree3}
\end{figure}

  If $P_1'$ and $P_2'$ intersect exactly in the vertex $u$ or $v$ the corresponding constraint is definitely captured. Similarly, if $P_1'$ and $P_2'$ intersects exactly in the path $P_\alpha$.
Otherwise, if $P_1'$ and $P_2'$ intersects in a path $P_l$ 
such that $I_{l}= (\min(P_l), \max(P_l))$ contains $I_\alpha$, we have the corresponding
constraint present implicitly. 

Refer to Figure~\ref{fig:thetaTree3}. Indeed, an ordering $\mathcal{O}$ of the columns of $M'$  witnessing that $\mathcal{I}$ is a ``yes'' instance satisfies 
$\{e_{\alpha}e_l\}\{e_pe_q\}$  by $T_1$, where
$e_p$ and $e_q$, respectively, is the edge incident to $u$ belonging to $P_1'$ and $P_2'$,
and  also $\mathcal{O}$ satisfies $\{f_{\alpha}f_l\}\{f_rf_s\}$, by the consistency tree $T_C$, where
$f_r$ and $f_s$, respectively, is the edge incident to $v$ belonging to $P_1'$ and $P_2'$,
Moreover, by the constraint obtained from the union of
paths  $P_1'$ and $P_2'$ by replacing $P_l$ with $P_{\alpha}$, we obtain
that $\mathcal{O}$ satisfies $\{e_pf_r\}\{e_qf_s\}$. By the constrains of $T_0$ it then 
follows that in the rotation at $u$ 
 the edges $e_l,e_p,e_q$ appear w.r.t. to this order 
 with the opposite orientation as $f_l,f_r,f_s$. Hence, the pair of $P_1'$ and $P_2'$ is feasible with respect to $\mathcal{O}$.

\begin{figure}[htp]
\centering
\includegraphics[scale=0.7]{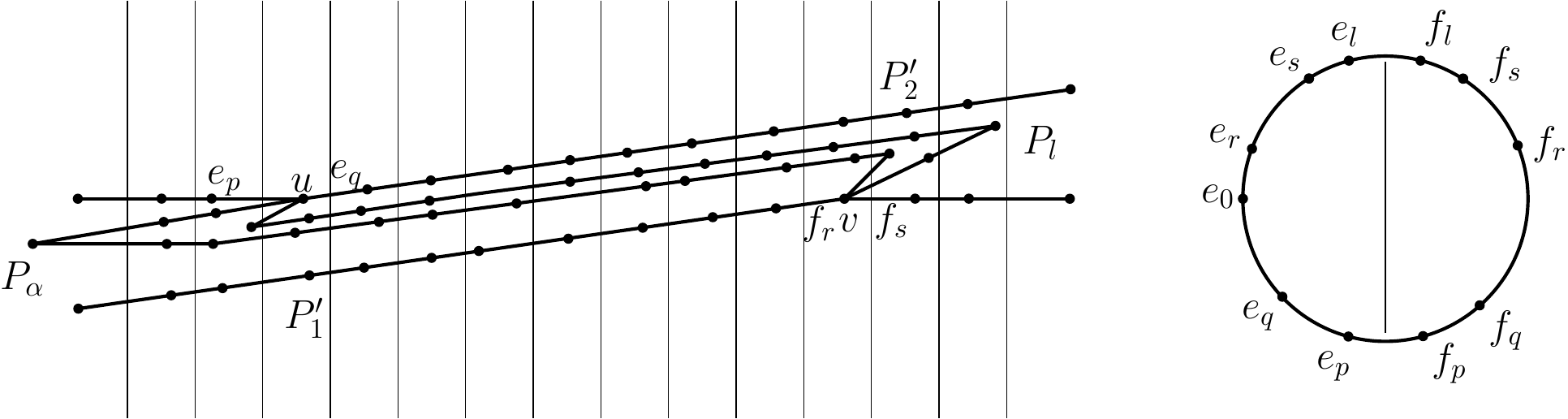}
\caption{The paths $P_1'$ and $P_2'$ meeting in $P_l$ such that $I_{l}= (\min(P_l), \max(P_l))$ 
does not contain $I_\alpha$ (left). The corresponding cyclic order of the leaves  corresponding to the edges incident to $u$ and $v$ captured by $T_0$ (right).}
\label{fig:thetaTree4} 
\end{figure}

Finally, if $I_{l}$ does not contain $I_\alpha$, the previous argument does not apply if the interval \\
$(\min(P_1'\cup P_2'), \max(P_1' \cup P_2'))$ does not contain 
$I_\alpha$. We assume that \\ $$\max(P_1' \cup P_2')>\max(I_l)>\max(I_\alpha)\ge \min(I_l)> \min(P_1' \cup P_2'))\ge \min(I)$$ and handle the remaining cases by the symmetry. We assume that $P_1'$ is a cap
 passing through edges $e_p$ and $f_r$.  We assume that $P_2'$ is a cup
 passing through edges $e_q$ and $f_s$. 
 
 The ordering $\mathcal{O}$ satisfies $\{e_pe_{\alpha}\}\{e_le_q\}$ and $\{f_rf_{\alpha}\}\{f_lf_s\}$.
 By  $T_1$ we have $\{e_qe_s\}\{e_{\alpha}e_l\}$. By $T_0$ and $T_C$,  $e_i$'s and $f_i$'s appear consecutively and they are reverse of each other in $\mathcal{O}$. Thus, we have $\{e_lf_l\}\{e_qf_s\}$.
 
 Refer to Figure~\ref{fig:thetaTree4}.
  A simple case analysis reveals that the observations in the previous paragraph gives us the following.
 If  $\{e_pf_r\}\{e_qf_s\}$ is
  satisfied by $\mathcal{O}$,
  then $\mathcal{O}$ satisfies $\{e_lf_l\}\{e_pe_q\}$ if and only if $\{e_lf_l\}\{f_sf_r\}$.
   On the other hand, if $\{e_pf_r\}\{e_qf_s\}$ is not satisfied by $\mathcal{O}$,
   then  $\mathcal{O}$ satisfies $\{e_lf_l\}\{e_pe_q\}$ if and only if it does not satisfy $\{e_lf_l\}\{f_sf_r\}$.
    By the symmetry there are four cases to check (see Figure~\ref{fig:thetaTree5}).
   Using the language of the formal logic the previous fact about $\mathcal{O}$
   is expressed by the following formula.

$$  \{e_pf_r\}\{e_qf_s\} \Rightarrow  (\{e_lf_l\}\{e_pe_q\} \Leftrightarrow \{e_lf_l\}\{f_sf_r\})  \ \bigwedge \  \neg \{e_pf_r\}\{e_qf_s\} \Rightarrow  (\{e_lf_l\}\{e_pe_q\} \Leftrightarrow \neg \{e_lf_l\}\{f_sf_r\}) $$
   
   \begin{figure}[htp]
\centering
\includegraphics[scale=0.7]{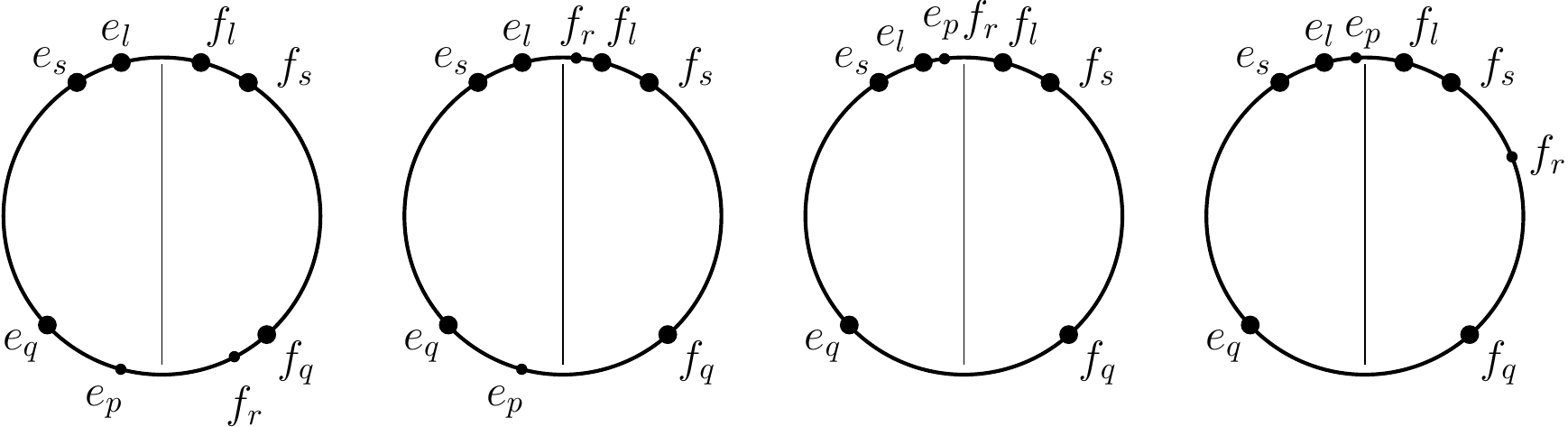}
\caption{The orderings of the leaves in $T_0$ corresponding to the true assignment of the propositional formula. Since only $e_p,e_q,f_r,f_s,e_l$ and $e_l$ appear in the formula.
We can fix $e_q,f_s,e_l$ and $f_l$ and use top-bottom symmetry.}
\label{fig:thetaTree5} 
\end{figure}

By the formula, it easily follows that the orientation  around $u$ and $v$, respectively, of the edges $e_l,e_p,e_q$ and $f_l,f_r,f_s$ is opposite of each other and we are done.
\end{proof}

\subsection{Hanani--Tutte}

We note that in our instance $\mathcal{I}$ in every tree $T$ all the $Q$-nodes are fixed by $Q$-nodes in only one of its children. This property remains to be true in the expansion graph. Thus,  $\mathcal{I}$ is a ``yes'' instance if all its
trees in the expansion graph are.

By the same token as in Section~\ref{sec:tree} the corresponding variant of the Hanani--Tutte theorem for strip clustered theta graphs
 follows once we extend Observation~\ref{obs:separation} to all the constraints implied
 by the trees in the expansion graph of $\mathcal{I}$~\cite{BR14}. By the proof of Theorem~\ref{thm:tree} we can extend
  Observation~\ref{obs:separation} to the constraints implied by trees $T_0,T_2,\ldots,T_m$, if
a strip clustered independently even
drawing of $(G',T')$ (corresponding to $T_0$) exists.
To this end in the initial drawing of $(G,T)$ we draw a copy of every path between $u$ and $v$  of $G$
in its close neighborhood. The desired drawing of $(G',T')$ is contained in the doubled
drawing.

  Let us fix an independently odd strip clustered drawing of $(G,T)$.
We show that constraints implied by $T_1$ satisfy the condition of Observation~\ref{obs:separation}.
First, suppose that $\{e_1e_2\}\{e_3e_4\}$, where $e_1,e_2,e_3$ and $e_4$ are arbitrary edges incident to $u$ is implied by $T_1$.
Let $C$ denote the cycle in $G$ passing through $e_1$ and $e_2$.
W.l.o.g. assume that $e_3$ and $e_4$, respectively belongs to the path $P_3$ and $P_4$ in $G$ containing a vertex $u'$ and $v'$
such that $\gamma(u'),\gamma(v')<\min (C)$ or $\gamma(u'),\gamma(v')>\max(C)$. 
Since in an independently
even strip clustered drawing of $(G,T)$ the paths $P_3$ and $P_4$ eventually visit the unbounded region in the complement of $C$ in the plane, we have $cr(e_1,e_3)+cr(e_2,e_3)=cr(e_1,e_4)+cr(e_2,e_4)$
if $\{e_1e_2\}\{e_3e_4\}$ is respected by the rotation at $u$.

\paragraph{Dummy edge.}
Let $e_0$ denote a dummy edge incident to $u$ 
corresponding to the column of $M'$ representing the position of the outer face in the rotation at $u$.
The edge $e_0$ is added to the drawing and spans all the clusters. We do not worry that it
violates the properties of strip clustered drawing. We just need $e_0$ to cross every edge of $G$
not adjacent to it an even number of times. To achieve this is fairly easy.
We notice that the edges of $G$ not adjacent to $e_0$ form a tree $G\setminus u$. Thus, no matter how we draw $e_0$, the edges crossing it an odd number of times will form an edge cut in $G\setminus u$.
By  switching  $e_0$ with all the vertices on one side the cut we make $e_0$ independently even.\\

Suppose that $\{e_0e_1\}\{e_2e_3\}$ is implied by $T_1$, and respected by the rotation at $u$.
Similarly as above we obtain $cr(e_0,e_2)+cr(e_1,e_2)=cr(e_0,e_3)+cr(e_1,e_3)$.
Indeed, both $e_0$ and $e_1$ belong to a path that visits a vertex outside of the cycle containing $e_2$ and $e_3$.

As we already mentioned in our instance of $\mathcal{I}$ every $Q$-node of a PC-tree is fixed only by $Q$-nodes of only one of its children. This property is kept in the expansion tree of $\mathcal{I}$.
It remains to show that  Observation~\ref{obs:separation} is extendable to constraints implied by expansion trees of $\mathcal{I}$. To this end we  need to show that the projection of  
the representatives of the fixed edges of $\mu$, when constructing  $(\mu,T'',T''')$ (see ~\cite[Section 3]{BR14}) maintains the validity of Observation~\ref{obs:separation} in the resulting PC-tree.
By Lemma~\ref{lemma:modular}, Observation~\ref{obs:separation}
holds for edge modules.

We proceed by a  top-to-bottom induction in the corresponding DAG.
We took care of the base case above.
Suppose that $\{e_1e_2\}\{e_3e_4\}$ is implied by $T''$, and in the tree
$(\mu,T'',T''')$, the constraint projects to $\{e_1'e_2'\}\{e_3'e_4'\}$, where
$e_i$ is mapped to $e_i'$.
We have $cr(e_1,e_3)+cr(e_2,e_3)+cr(e_1,e_4)+cr(e_2,e_4)=0$ (by induction hypothesis),
if  $\{e_1e_2\}\{e_3e_4\}$ is respected in the drawing.
We want to show that $cr(e_1',e_3')+cr(e_2',e_3')+cr(e_1',e_4')+cr(e_2',e_4')=0$,
if $\{e_1'e_2'\}\{e_3'e_4'\}$ is respected in the drawing.

If $e_1\not=e_1'$, we have $\{e_1e_1'\}\{e_3e_4\}$,
implied by a common parent of $T''$ and $T'''$. 
Thus, $cr(e_2,e_3)+cr(e_2,e_4)=cr(e_1,e_3)+cr(e_1,e_4)=cr(e_1',e_3)+cr(e_1',e_4)$,
if $\{e_1'e_2\}\{e_3e_4\}$ is respected by the drawing.
Similarly, we obtain $cr(e_1',e_3)+cr(e_1',e_4)=cr(e_2,e_3)+cr(e_2,e_4)=cr(e_2',e_3)+cr(e_2',e_4)$,
if $\{e_1'e_2'\}\{e_3e_4\}$ is respected by the drawing.
Repeating the same argument for $e_3$ and $e_4$ we obtain 
$cr(e_1',e_3')+cr(e_2',e_3')+cr(e_1',e_4')+cr(e_2',e_4')=0$,
if $\{e_1'e_2'\}\{e_3'e_4'\}$ is respected in the drawing.

\subsection{Beyond theta graphs and trees}
 
 One might wonder if our algorithm and/or Hanani--Tutte variant for theta graphs can be extended to strip clustered
 graphs with an arbitrary underlying abstract planar graph.
 
 It is tempting to consider the following  definition of a strip clustered tree $(G',T')$ corresponding 
 to a strip clustered instance $(G,T)$.
 Let us suppress every vertex of degree two in $G$. 
Let $G''$ denote the  resulting multi-graph. Each edge $e$ in $G''$ has a corresponding
path $P:=P(e)$ in $G$ (which is maybe just equal to $e$).
Let us assign a non-negative weights $w(e)$ to every edge $e$ of $G''$ equal to $\max(P)-\min(P)$.
Let $G'''$ denote the minimum weight spanning tree of $G''$.

The tree $G'$ is obtained as the union of the sub-graph of $G$ corresponding to
$G'''$, and pairs of copies of paths $P(e)$, for each $e=uv$, where
one copy of each pair is attached to $u$ and the other one to $v$, but otherwise 
disjoint from $G'''$. The assignment of the vertices to the clusters  in $(G',T')$ 
is inherited from $(G,T)$.

Note that if $G$ is a theta graph the resulting $G'$ is almost equal to the one defined
in the previous subsections, except that earlier we shortened copies of $P(e)$'s for $e$ not in $G'''$
so that they do not contain one end vertex of $P(e)$. In the case of theta graphs this does
not make a difference, but for more general class of graphs our new definition of $G'$ might
be more convenient to work with. 

Now, we would like to use the construction of~\cite[Section 4.1]{BR14} enriched by
the constraints of the strip clustered tree $(G',T')$, and other necessary contraptions if 
the graph is not two-connected.

If the graph is two-connected it seems plausible that
it is enough if we  generalize our construction of $T_1$ (taking care of trapped vertices),
and prove that the constraints of $(G',T')$ together with other constraints account for 
all unfeasible interleaving pairs of path.
To construct $T_1$ is not difficult, since we just repeat our construction of the matrix corresponding to $T_1$ for each consistency tree corresponding to a pair of vertices participating in the two-cut,
and combine the resulting matrices.

Then it seems that the only challenging part is to adapt the last paragraph in the proof of Theorem~\ref{thm:theta_alg}, which does not seem to be beyond reach. To this end it is likely that  a more efficient version of the characterization in Theorem~\ref{thm:characterization} is also needed, where by ``more efficient'' we mean a version
that restricts the set of interleaving pairs of paths considered. This should be possible,
since in our proofs for trees and theta graphs a considerably  restricted subset
of interleaving feasible pairs implied feasibility for the rest.

When the graph $G$ is not guaranteed to be two-connected, generalizing $T_1$ does not
seem to be a way to go. However, our general strategy of combining
simultaneous PC-orderings with our characterization could work, if the constraints preventing an occurrence of trapped vertices and unfeasible interleaving pairs  are treated simultaneously using a more efficient version of 
 Theorem~\ref{thm:characterization}.

\section{Open problems}

\label{sec:open}

We proved the weak variant of the Hanani-Tutte theorem for strip clustered graphs, and verified the corresponding strong variant for
three-connected graphs, trees and theta-graphs.
A next natural step would be the case, when $G$ is a series-parallel graph,
possibly with cycle-free components in the cactus decomposition,
which should be within reach using the developed techniques.

If we look at c-planarity testing of clustered graphs as a graph augmentation problem  (as explained in the introduction)
 one might feel that the problem, if tractable, could be solved by a variant of matroid intersection algorithm~\cite{L75}.
 However, until now the matroid intersection algorithm was shown to imply a tractability of c-planarity testing only in a very limited special case~\cite{FKP12J}.
 Our work supports this connection, even though less directly, since a perfect matching whose existence is guaranteed by Hall's theorem can be found as a set  in the intersection
of  two partition matroids.



Also it would be interesting to decide if Corollary~\ref{thm:treeHT} can be extended to flat clustered graphs not containing
a family of counter-examples from~\cite{FKP12J} as described in the full version~\cite{FKP12J}.

\paragraph{Acknowledgment.}

We would like to express our sincere gratitude to the organizers and
participants of the 11th GWOP workshop, where we could discuss the research
problems treated in the present paper. In particular, we especially benefited from the discussions with Bettina Speckmann, Edgardo Rold\'{a}n-Pensado and Sebastian Stich.

Furthermore, we would like to thank J\'{a}n Kyn\v{c}l for useful discussions at the initial stage of this work and many useful comments, Juraj Stacho
for directing my attention to 0--1 matrices with consecutive ones property, Arnaud de Mesmay for pointing out a gap in one of the proofs, and
 G\'abor Tardos for comments that helped to improve the presentation of the results.

\bibliographystyle{plain}

\bibliography{bib}





\end{document}